%% file: main.tex
\documentclass[letterpaper,twocolumn,10pt]{article}
\input{headers}
\input{macros}

\begin{document}

\date{}

\title{\Large \bf Sequential Auditing for f-Differential Privacy}

\author{
     {\rm Tim Kutta}$^{1}$\thanks{Corresponding author: \texttt{tim.kutta@math.au.dk}}, 
     {\rm Martin Dunsche}$^{2}$, 
     {\rm Yu Wei}$^{3}$, 
     {\rm Vassilis Zikas}$^{3}$ \\
     \\
     $^1$Aarhus University \\
     $^2$Ruhr University Bochum\\
     $^3$Georgia Institute of Technology\\
}

\maketitle

\input{00-abstract}

\input{01-introduction}

\input{02-background}
\input{03-methodology}

\input{04-experiments}

\section*{Acknowledgments}

Tim Kutta's work has been supported by the AUFF Nova grant 47222. Yu Wei and Vassilis Zikas were supported in part by NSF Award No. 2448339 and No. 2531010, Halcyon Futures via the AI Security Institute, JPMorgan Chase, the Stellar Development Foundation, and Sunday Group, Inc. Martin Dunsche was partially Funded by the Deutsche Forschungsgemeinschaft (DFG, German Research Foundation) under Germany’s Excellence Strategy - EXC 2092 CASA - 390781972

\input{05-ethical-considerations}
\input{06-open-science}

\printbibliography

\appendix

\input{100-proofs-and-technical-details}

\input{101-additional-algorithms}
\input{102-theory-adaptive-classifiers}

\input{103-additional-simulations}
\input{104-full-version-supplementary-material}

\end{document}

%% file: headers.tex
\usepackage{usenix}

\usepackage{tikz}
\usepackage{amssymb}
\usepackage{amsmath}
\usepackage{float}

\usepackage{booktabs}
\usepackage{multirow}

\usepackage{algorithm}
\usepackage{algpseudocode} 
\usepackage{amsthm}
\usepackage{caption}
\usepackage{subcaption}
\usepackage{xspace}
\usepackage[available]{usenixbadges}

%% file: macros.tex
\newcommand{\printbibliography}{%
  \bibliography{reference}
}

\theoremstyle{definition}
\newtheorem{theo}{Theorem}[section]
\newtheorem{lemma}[theo]{Lemma}

\newtheorem{definition}[theo]{Definition}
\newtheorem{example}[theo]{Example}
\newtheorem{rem}[theo]{Remark}
\newtheorem{ass}[theo]{Assumption}

\DeclareMathOperator*{\argmax}{arg\,max}
\DeclareMathOperator*{\argmin}{arg\,min}

\newcommand{\vlineref}[1]{\ensuremath{\langle #1\rangle}\xspace}

\newcommand{\ve}{\varepsilon}

\newcommand{\Teins}{\hat T_\alpha}
\newcommand{\Tzwei}{\hat T_\beta}
\newcommand{\claim}{\ensuremath{ }}

%% file: 00-abstract.tex
\begin{abstract}
We present new auditors to assess Differential Privacy (DP) of an algorithm based on output samples. Such empirical auditors are common to check for algorithmic correctness and implementation bugs. Most existing auditors are batch-based or targeted toward the traditional notion of $(\varepsilon,\delta)$-DP; typically both. In this work, we shift the focus to the highly expressive privacy concept of $f$-DP, in which the entire privacy behavior is captured by a single tradeoff curve. Our auditors detect violations across the full privacy spectrum with statistical significance guarantees, which are supported by theory and simulations. Most importantly, and in contrast to prior work, our auditors do not require a user-specified sample size as an input.
Rather, they adaptively determine a near-optimal number of samples needed to reach a decision, thereby avoiding the excessively large sample sizes common in many auditing studies. This reduction in sampling cost becomes especially beneficial for expensive training procedures such as DP-SGD. Our method supports both whitebox and blackbox settings and can also be executed in one-run frameworks.
\end{abstract}

%% file: 01-introduction.tex
\section{Introduction}
As data is becoming more valuable than gold, privacy is an increasingly pressing concern: how can we optimally utilize potentially sensitive data without compromising individuals’ privacy? Differential Privacy (DP)~\cite{dwork2014algorithmic} has emerged as a broadly accepted notion that hits a sweet spot between privacy and statistical utility.
The idea behind DP is simple: The answer to a (privatized) query should not severely compromise the privacy of any individual record in the original dataset. A bit more formally, a mechanism satisfies DP if changing any single record has only a ``small'' effect on the outcome of the (privatized) query. Measuring the impact on outcomes naturally requires a notion of distance, which in classical DP is captured by two parameters $\varepsilon$ and $\delta$. 

Following the initial study of $(\varepsilon,\delta)$-DP---and fueled by the adoption of DP as an acceptable privacy notion by academia, industry and even governments~\cite{Dwork2006_approx,erlingsson2014rappor,Abowd2018,abowd2022topdown,Bolin2017,apple2017learningprivacy}---a number of DP variants have emerged as alternatives to tackle limitations of the $(\varepsilon,\delta)$-DP notion. Examples include divergence-based relaxations of DP (e.g., Rényi-DP~\cite{mironov:2017}, zero-concentrated DP~\cite{BunS16}), $f$-DP~\cite{dong:roth:su:2022}, and instance-based privacy concepts (e.g., DDP~\cite{BassilyGKS13}, PAC Advantage Privacy~\cite{XiaoD23}). 

Many initial results in DP considered feasibility, optimal utility and privacy tradeoffs, composability, and comparison between different privacy notions~\cite{dwork2014algorithmic,mcsherry2007mechanism,hardt2010geometry,dwork2010boosting,Vadhan17,kamath2020primer,cummings2024advancing}. These results were typically either theoretical in nature---designing mechanisms and proving their privacy/utility---or they provided proof-of-concept implementations of such ideas and designs. The adoption of DP, however, for privacy-sensitive applications has put the notion under the scrutiny of the security research community, raising an additional question: How can we verify that (a given implementation of) a mechanism does not have a bug that makes it leak more than its specification? 

\paragraph{DP Auditing}
The above question ignited a rich literature in security that aimed at testing/auditing whether implementations of common DP mechanisms satisfy their theoretical privacy guarantees. This led to a growing interest in empirical privacy auditing with results that range from using classical bug detection approaches---e.g., inspired by code obfuscation or worst-case instance discovery techniques---to deploying statistical and machine learning tools to compare neighboring distributions with respect to privacy-relevant events~\cite{DP-Finder,DP-Sniper,steinke2023privacy,mahloujifar2024auditing}.

The auditing toolbox is quite diverse, yet DP auditors can be broadly separated into two classes: {\em whitebox} auditors assume that the analyst knows specifics of the audited mechanism (e.g., parameters of the noising distribution), or even that full access to the concrete implementation is granted~\cite{Liu2019,DP-Finder,CheckDP,nasr2023tight}. In contrast, {blackbox} auditors only require blackbox (input/output) access to the DP mechanism.\footnote{In some DP auditing work, blackbox still assumes the output distribution is known (e.g., audited mechanism is DP-SGD, hence corresponds to (subsampled) Gaussian tradeoff; likewise, LiRA-style attacks often model the distribution of losses as Gaussian~\cite{nasr2023tight,carlini2022membership}. Our auditing methods rely on much weaker assumptions.}

Naturally, the latter class poses more challenges, as tools from bug-finding methodologies cannot be used there and one needs to rely almost exclusively on statistical and machine learning-based methods as discussed below. 
\paragraph{f-DP Auditing}
Despite the increasing literature on DP auditing, it is fair to say that most of this literature focuses on testing (implementations of) mechanisms for the classical $(\varepsilon,\delta)$-DP notion. As discussed below, this puts several limitations, especially in the realm of blackbox and/or optimally efficient auditing. Notwithstanding, a number of recent works~\cite{dong:roth:su:2022,WangSuYeShokri2023,mahloujifar2024auditing,AwanRamasethu2024} have demonstrated that \emph{$f$-Differential Privacy}, in short, $f$-DP, can be more expressive than $(\varepsilon,\delta)$-DP, while being more friendly to statistical methods for DP auditing. As such, $f$-DP has emerged as an excellent candidate for efficiently certifiable implementations in both whitebox and blackbox settings~\cite{nasr2023tight}. In fact, as discussed in \cite{dong:roth:su:2022} $f$-DP is a particularly suitable privacy notion for methods like {\em differentially private stochastic gradient descent} (in short, DP-SGD) that are increasingly popular in machine learning. 

In a nutshell, $f$-DP connects the closeness of a (DP) mechanism's output distributions on neighboring datasets to the accuracy of (optimal) classifiers.
If classifiers can reliably distinguish the output distributions, changes to the databases have a strong impact and
privacy is low
 (we refer to Section~\ref{sec:2.1} for a detailed definition). In classical statistics, the theory of optimal binary classifiers is well studied and well understood, providing a diverse toolbox for quantifying distinguishability between distributions. Thus it becomes evident that $f$-DP is far more amenable to tools from statistics and classical machine learning than $(\varepsilon,\delta)$-DP. 

Yet, from an auditing perspective the most crucial advantage of $f$-DP is the following:
The notion of  $(\varepsilon, \delta)$-DP offers a very ``localized'' view of the algorithm's privacy as it does not make any statements about the potential trade-offs between $\varepsilon$ and $\delta$. Instead, $f$-DP provides a privacy notion that, intuitively, charts the entire $(\varepsilon, \delta)$ spectrum. This allows detection of privacy violations at a range of different scales simultaneously, and thereby yields much more powerful auditing procedures.
\paragraph{Sequential f-DP Auditing}
Motivated by the above, our work aims to push the boundaries of $f$-DP auditing both in the blackbox and the whitebox realm, by adapting powerful tools from statistics and machine learning. More concretely, we move from fixed-sample (fixed-batch) auditing procedures to a sequential auditor that adaptively determines how many samples are needed, with no major computational overhead. This approach promises substantially higher sample-efficiency and is supposed to complement existing auditing approaches by providing a generic sequential wrapper for any $f$-DP auditor. In the following, we outline the main ideas of our work and how they advance the state of the art.

Standard privacy auditors/tests take a fixed, predetermined number of outputs (from the mechanism), say $n$, and given these outputs run some statistical hypothesis test to assess privacy, e.g., \cite{Eureka, lowerbounds,askingeneral}. Mathematically speaking, such methods can be proved to be optimal ---by using optimal testing theory--- but in practice they suffer from the "fixed sample size"-issue, which is a common bottleneck in testing: the ideal sample size (just large enough for detection, but not larger) of any statistical/hypothesis test is strongly case dependent (we refer to Section \ref{sec:2:2} for a detailed discussion). If the user chooses $n$ too small, they risk inconclusive or false auditing results. If they choose $n$ too large, they waste computational resources. Naturally, experiments and benchmarks err on the side of caution (to eliminate inconclusive or false results) and use extremely large sample sizes, often two orders of magnitude larger than necessary. 
Such oversampling may be harmless in foundational studies, but it naturally deters adoption, especially for auditing expensive algorithms such as DP-SGD.

Countering this trend toward ever more  costly auditors, a recent strand of research aims to make DP auditing less expensive, among them  \cite{steinke2023privacy,nasr2023tight,mahloujifar2024auditing}, which focus on complex machine learning techniques. While their techniques already cut the computational cost dramatically, potentially down to training the algorithm only once, mostly they rely on evaluating canaries (see, e.g., \cite{steinke2023privacy}). 
As they note, while one would ideally score all $n$ examples (i.e., set $m=n$ as number of canaries), doing so may be prohibitive and one often settles for $m<n$:
\begin{quote}
\textit{To get the strongest auditing results we would set $m = n$, but we usually want to set $m < n$. For example, computing the score of all $n$ examples may be computationally prohibitive, so we only compute the scores of $m$ examples.}
\end{quote}
This gap suggests that sequential, anytime-valid procedures are useful even in the one-run setting: they can adaptively decide \emph{how many} canaries to score without losing rigorous error control.

Finally, the very recent, parallel, and independent work \cite{gonzalez2025sequentially} also made use of sequential tests for DP auditing. We only became aware of that work during the later stages of our paper. Notwithstanding, in the following we discuss the differences in scope and methodology:  Their work from \cite{gonzalez2025sequentially} focuses on the $(\varepsilon,\delta)$-DP notion, which is distinct from our auditing of $f$-DP claims. This provides both qualitatively and quantitatively different results and we provide a detailed comparison in Section~\ref{sec:4}. In a nutshell, on the additive-noise benchmarks used by~\cite{gonzalez2025sequentially}, our auditor rejects buggy mechanisms with substantially fewer samples in the small-$\varepsilon$ regimes they report, using roughly $10\%$--$25\%$ of the paper-reported samples. Under a finite sample, it also continues to detect the buggy variants over a wider range of larger $\varepsilon$. This comparison, however, should not be read as a universal dominance guarantee: the two methods audit different privacy objects, rely on different statistical witnesses, and their sample efficiency also depends on the concrete construction instantiating the approaches.

We summarize the main contributions of this paper:
\begin{itemize}
\item [1)] We provide the first sequential $f$-DP auditor.
   \item[2)] We provide mathematical proofs for the performance of our auditor: In particular,
   \begin{itemize}
       \item[a)] we show that our auditor holds a user-determined significance level (false rejection rate) of $\gamma \in (0,1)$ (Theorem \ref{theo:main})~,
       \item[b)] and detects violations with minimal sample size, up to a logarithmic-factor (Theorem \ref{theo:main} and Remark \ref{rem:theo}). 
   \end{itemize}
   \item[3)] Subroutines of our procedure involve a new tuning approach for optimal classifiers for DP auditing.
   Our new strategy even improves performance when used for fixed-batch $f$-DP auditing. 
   \item[4)] We demonstrate the effectiveness of our approach on synthetic and real-world experiments, including single run auditors.
   \item[5)] We provide an open-source implementation.\footnote{\url{https://github.com/martindunsche/sequential_fdp_auditing}}
   \end{itemize}

\paragraph{Paper structure.}
We briefly outline the structure of this work. In Section \ref{sec:2}, we recall the key concepts underlying our methodology. Section \ref{sec:3} presents our main theoretical results, with the central contribution stated in Theorem \ref{theo:main}. In Section \ref{sec:4}, we experimentally investigate the performance of our approach: first using standard ground-truth settings, then comparing against the method of \cite{gonzalez2025sequentially} using their favorable setup, and finally examining a real-world example in the spirit of \cite{nasr2023tight} and \cite{gonzalez2025sequentially}. Improvements, relative to the existing $f$-DP auditors in \cite{askingeneral} are demonstrated in the Appendix. There, we also provide mathematical details and proofs for our new procedures.

%% file: 02-background.tex
\section{Background}\label{sec:2}

\subsection{$f$-Differential Privacy} \label{sec:2.1}
Differential Privacy (DP) is a concept to describe information leakage from a data processing mechanism $\mathcal{M}$. Qualitatively, privacy is high when for any pair of adjacent databases $D,D'$ the output distributions $P \sim \mathcal M(D)$ and $ Q \sim \mathcal M(D')$ are nearly indistinguishable. This ``near indistinguishability'' has been formalized in different ways (see e.g. \cite{dwork2014algorithmic, mironov:2017, dong:roth:su:2022} among others) and here we adopt the popular description of $f$-DP introduced in \cite{dong:roth:su:2022}. Intuitively, $f$-DP states that $P,Q$ are hard to distinguish when there is no reliable way to classify whether data were generated by $P$ or $Q$. 
Formally, suppose that upon making an observation $X$ a user wants to decide which distribution was used to generate $X$, i.e. 
\begin{align} \label{e:classtask}
     X \sim P \qquad \textnormal{versus} \qquad 
     X \sim Q.
\end{align}
The decision is encoded by a binary classifier $\phi$ where $\phi(X)=0$ means deciding that $X \sim P$ and $\phi(X)=1$ means deciding that $X \sim Q$. 
No matter how good the classifier is, there is always room for error. Namely, one could wrongly classify $X \sim P$ as being generated by $Q$, and conversely $X \sim Q$ as being generated by $P$. We can express the misclassification probabilities as
\begin{align} \label{e:class:errors}
    \alpha_\phi := \mathbb{E}_{X \sim P} [\phi(X)] \quad \textnormal{and} \quad 
    \beta_\phi := 1-\mathbb{E}_{X \sim Q} [\phi(X)].
\end{align}
We can now assess the discrepancy between the distributions $P$ and $Q$ by virtue of the classification errors. Roughly speaking, if there exists a classifier $\phi$ that can distinguish data drawn from $P$ versus $Q$ with small error probabilities, then the distributions $P$ and $Q$ must differ substantially. Ideally, we would like to use the best possible classifier $\phi$ for this task. However, it turns out that there is no single best classifier that minimizes $\alpha_\phi$ and $ \beta_\phi$ simultaneously, because decreasing one error generally increases the other one. Rather, for any value $\alpha \in [0,1]$ there exists a best (possibly randomized) classifier $\phi^*$ with smallest attainable error $\beta_{\phi^*}$ under the side-constraint $\alpha_{\phi^*} \le \alpha$. This motivates the use of an entire \textit{tradeoff function} that is a collection of all those pairs $(\alpha_{\phi^*}, \beta_{\phi^*})$. 

\begin{definition}[Tradeoff function] \label{def:trad}
    For any two distributions $P,Q$ defined on the same space, the tradeoff function is defined as 
    \begin{equation*}
        T(\alpha):= \inf\{\beta_{\phi}: \phi \text{ is a classifier with } \alpha_\phi \leq \alpha\}, \, \alpha \in [0,1]~.
    \end{equation*}
\end{definition}
$T$ is a measure for the distance of output distributions $P,Q$. Accordingly, 
it can be used for privacy quantification. The corresponding notion of privacy is \textit{$f$-DP}, where $f:[0,1]\to [0,1]$ is a convex and non-increasing function, with $f(0)=1$ and $f(1)=0$.

\begin{definition}[$f$-DP]\label{def:fdp}
    A mechanism $\mathcal M$ satisfies $f$-DP if for every pair of neighboring datasets $D,D'$, the corresponding tradeoff function $T$ satisfies $T(\alpha)\geq f(\alpha)$ for all $\alpha\in [0,1]$. 
\end{definition}

\subsection{Hypothesis testing and power analysis} \label{sec:2:2}

Next, we consider the concept of statistical testing. Testing is related to classification in that a binary decision has to be reached based on observed data. Yet, there are two important differences: First, a decision is not based on a single observation $X$ but rather on an entire sample $X_1,...,X_n$. Second, the decision process involves a default belief, making the problem inherently asymmetric: one hypothesis is assumed to hold unless the data provides sufficient evidence to reject it. 
Abstractly, we can write down the two competing, mutually exclusive hypotheses as follows:
\begin{align}\label{eq:hypotheses_setup}
H_0&: \textnormal{the baseline model}, 
\\ 
H_1&: \textnormal{a collection of alternative models}.\nonumber
\end{align}

The baseline \(H_0\) represents the default belief and is rejected in favor of \(H_1\) only if observed data provide sufficiently strong evidence against it. For illustration, consider in the context of DP-auditing a postulated tradeoff curve $f$. 
Then, the null hypothesis \(H_0\) states that a mechanism \(\mathcal{M}\)  satisfies $f$-DP, while \(H_1\) states that it does not. Accordingly, if $H_0$ is rejected, we have high confidence that $\mathcal{M}$ does not satisfy $f$-DP. 

As mentioned above, the statistical test reaches its decision between $H_0$ and $H_1$ based on a data sample \(X_1, \ldots, X_n\). The ability to tell apart $H_1$ from $H_0$ depends on two main ingredients: First, the sample size $n$, where larger samples include more information (leading to a more informed decision) but are also more costly to collect. Second, the effect size, which is the deviation between the true data generating process and the model $H_0$. Recall that $H_1$ (in contrast to $H_0$) is not a single model - it is a collection of many possible models. Some of these models are similar to the baseline model $H_0$ and some are very different. Large discrepancies (large effect sizes) are naturally easier to identify than subtle ones (small effect sizes). 
In contrast to the sample size, the analyst has no influence on the effect size. Indeed, the effect size is unknown to them, since the true data generating process is unknown. This turns out to be a major problem in hypothesis testing.

The performance of a statistical test is captured by two fundamental probabilities: \\[-6ex]
\begin{align}\label{e:sign}
{} \\
\text{significance level (false rejection rate):} &\,\,\, \mathbb{P}(\text{reject } H_0 \mid H_0), \nonumber\\
\text{power (correct rejection rate):}  &\,\,\, \mathbb{P}(\text{reject } H_0 \mid H_1).\nonumber
\end{align}
Across empirical sciences, it is common to fix the significance level at some low value, typically $5\%$. Under this constraint,
the main aim becomes attaining a certain high level of power (say $80\%$). Reaching a desired power is non-trivial because
power is interconnected with both the sample size and the effect size; see Figure \ref{fig:power_triangle}. 
More precisely, larger samples and larger effect sizes both lead to more power.
In an idealized setting where the effect size were known to the analyst, they could analytically determine the minimal required sample size, denoted by \(n_{\min}\), to attain a desired power. However, in reality the effect size is unknown. To hedge against this uncertainty, analysts often adopt a conservative strategy and employ substantially larger samples than theoretically necessary, such that \(n \gg n_{\min}\).

In DP auditing, it is not uncommon to encounter sample sizes that exceed the theoretical minimum by two to three orders of magnitude, placing a big computational burden on the analyst. The impossibility of selecting the optimal sample size before testing motivates an adaptive choice that is made during testing.  This strategy leads to close-to-optimal results and is the subject of sequential testing.
 
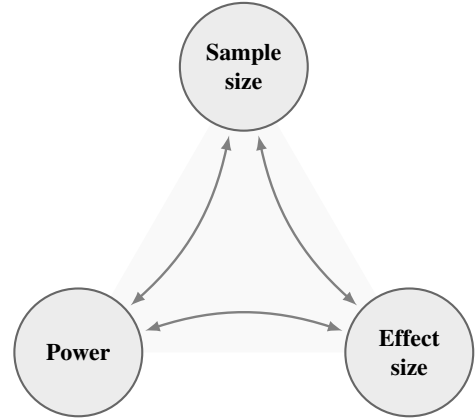
\begin{figure}[ht]
  \centering
  \begin{tikzpicture}[
    scale=1.6,
    thick,
    >=latex,
    font=\sffamily,
    every node/.style={font=\small},
    concept/.style={
      circle,
      draw=black!60,
      fill=gray!15,
      font=\bfseries\small,
      align=center,
      inner sep=0pt,
      minimum size=17mm
    },
    arrowstyle/.style={
      <->,
      color=black!50,
      line width=0.9pt,
      shorten >=2pt,
      shorten <=2pt,
      bend left=18
    }
  ]

    \coordinate (A) at (0,0);
    \coordinate (B) at (2.2,0);
    \coordinate (C) at (1.1,1.9);

    \fill[gray!5] (A)--(B)--(C)--cycle;

    \node[concept] (Power) at (A) {Power};
    \node[concept] (Effect) at (B) {Effect\\size};
    \node[concept] (Sample) at (C) {Sample\\size};

    \draw[arrowstyle] (Power) to (Effect);
    \draw[arrowstyle] (Effect) to (Sample);
    \draw[arrowstyle] (Sample) to (Power);

  \end{tikzpicture}
  \caption{  \label{fig:power_triangle}  Power, sample size, and effect size form a triad in which any two quantities determine the third. Thus, an optimal sample size can only be computed once both the desired power and effect size are specified.}
  \vspace{-0.8cm}
\end{figure}

\subsection{Sequential testing} \label{sec:2:3}

Consider the setup from the previous section. To avoid the use of excessively large samples for hypothesis testing, sequential procedures have been developed that adaptively minimize the sample size during the testing process. The core idea is simple: In practice, an analyst rarely collects all data at once; rather, observations arrive sequentially. First, $X_1$ is observed, then, $(X_1,X_2)$, and so on. By monitoring the accumulating data in real time, the analyst may already obtain sufficient evidence to reject $H_0$ in favor of $H_1$ after observing only a small prefix $(X_1,\dots,X_k)$. Here, the time point $k$ at which $H_0$ is first rejected is a random variable, but it typically concentrates near the theoretical minimum $n_{\min}$.

Designing valid sequential procedures is non-trivial. An ad-hoc approach might be to repeatedly apply a standard fixed-sample test, every time a new observation is made. Unfortunately, this naive strategy cannot work, because the false rejection rate grows uncontrollably. We illustrate this point shortly with the next example.
\begin{example}
    Consider data from a normal distribution $X_i \sim \mathcal{N}(\mu,1)$, with mean $\mu \in \mathbb{R}$. The hypotheses under investigation are $H_0: \mu=0$ vs. $H_1:\mu\neq 0$. To sequentially check $H_0$ vs. $H_1$, the analyst runs a $t$-test repeatedly whenever a new observation is made and stops after the first rejection. Each $t$-test individually has a significance level of $5\%$.
    Now, let us assume that the true data are indeed generated by the model $\mathcal{N}(0,1)$ and thus, any rejection of $H_0$ would be erroneous. We have run $1,000$ simulations to obtain empirical false rejection rates and plotted them in Figure \ref{fig:ttest}. As we can see, rather than holding the desired significance level (red horizontal line), the false rejection rate rises rapidly and is already $ \approx 24\%$ after the first $200$ datapoints. We do not display it here, but it can be mathematically proven that the false rejection rate increases ever further, converging to $100\%$ if testing goes on indefinitely. This means a false rejection is eventually guaranteed with the naive procedure.

 \begin{figure}
     \centering
     \includegraphics[width=.7\linewidth]{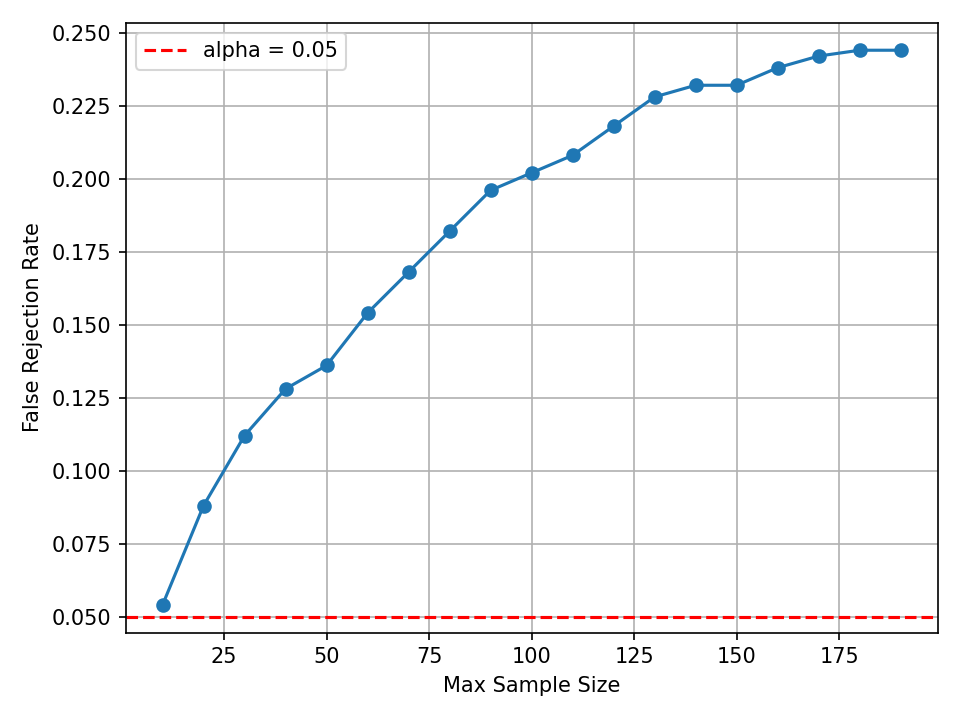}
     \caption{False rejection rates in sequential testing without proper adjustment.}
     \label{fig:ttest} 
    \vspace{-0.6cm}
\end{figure}

\end{example}

The above example illustrates the need for designated sequential tests to ensure that the significance level holds, even when testing goes on for a very long time. Rather than consisting of a  fixed-sample test that is applied over and over again, sequential tests compute an evolving statistic that can be approximated by a stochastic process. The probability of that process crossing a boundary, which is analytically tractable, then gives a closed-form for the false rejection rate of the sequential test. The development of such procedures is mathematically challenging and relies on modern insights from open-ended Gaussian approximations such as developed by \cite{berkes:liu:wu:2014}. 

%% file: 03-methodology.tex
\section{Methodology}\label{sec:3}

In this section, we present a sequential $f$-DP auditing approach that adapts the sample size (number of algorithmic runs) to the unknown level of privacy violation. Strong violations can be detected with a few samples, while subtler effects automatically trigger more sampling.

\subsection{Sequential privacy auditing}\label{sec:3.1}

We fix a pair of adjacent databases $D,D'$ and denote the corresponding output distributions of the mechanism $\mathcal M$ by $P \sim \mathcal M(D)$ and $Q \sim \mathcal M(D')$. Recall the central Definitions \ref{def:trad} and \ref{def:fdp}. In Definition \ref{def:trad}, we noticed that the discrepancy between $P$ and $Q$ is captured by the tradeoff function $T$. In Definition \ref{def:fdp}, we said that $f$-DP holds if the function $T$ lies geometrically above the function $f$. The task of an $f$-DP auditor is then to check for a candidate function $f$, whether $f$-DP holds.
This task is naturally framed as a hypothesis test (see Section~\ref{sec:2:2}), where we decide between the competing hypotheses of $f$-DP ($H_0$) and no $f$-DP ($H_1$)
\begin{align} \label{e:hyp:aud}
   & H_0: T(\alpha) \geq f_{\claim}(\alpha) \quad \forall \alpha \in [0,1], \\
   & H_1: \exists \,\, \alpha^* \in [0,1] \ \textnormal{such that}\  T(\alpha^*) < f_{\claim}(\alpha^*). \nonumber
\end{align}
To decide between $H_0$ and $H_1$, the analyst may include prior knowledge in their decision process. In an ideal whitebox scenario, the analyst already knows a powerful classifier $\phi$ to distinguish $P$ and $Q$ (see Section \ref{sec:2.1}), while in a blackbox setting only minimal knowledge is available, e.g. that $P$ and $Q$ have probability densities. Our methodology is compatible with both settings and details are discussed below. 
In either case, the analyst bases their final verdict on samples of algorithmic outputs. More precisely, they draw independently $X_1,X_2,... \sim P$ and $Y_1, Y_2,.... \sim Q$ and employ a statistical test for the hypotheses pair \eqref{e:hyp:aud}. This test should have the two properties described after Eq. \eqref{e:sign}: First, it should bound the false rejection rate by a small, user-determined number $\gamma \in (0,1)$, the significance level. Second, when $f$-DP is violated, the test should reject $H_0$ using as few samples of algorithmic outputs as possible. We therefore use a sequential testing procedure (background in Section \ref{sec:2:3}) that we describe next.

\begin{algorithm}[h]
\small
\begin{algorithmic}[1]
\Require \; \parbox[t]{\dimexpr\linewidth-\algorithmicindent-3em}{Distribution $P,Q$
Burn-in size $M$, bandwidth $h$, candidate vector  $\eta$,
refresh period $M_1$, evaluation period $M_2$, claimed curve $f_{\claim}$, critical $1-\gamma/2$ value $q_{1-\gamma/2}$, max iterations $k_{max}$.
} 
\Ensure \, Online monitoring for violations of a claimed privacy curve.

\Function{Sequential-Audit}{$M, h, \eta$, $M_1$, $M_2$, $f_{\claim}$, $q_{1-\gamma/2}$, $k_{max}$}
    \State Draw $M$ samples from $P$ and $Q$~.
    \State Define classifier $$\phi(\cdot) \gets \textproc{BuildClassifier}( \{X_i\}_{i=1}^{M}, \{Y_i\}_{i=1}^{M}, h, \eta, f_{\claim} )~.$$
    \State Compute initial statistics:
        \[
        \hat \alpha_\phi := \frac{1}{M} \sum_{i=1}^{M} \phi(X_i), \quad 
        \hat \beta_\phi := 1 - \frac{1}{M} \sum_{i=1}^{M} \phi(Y_i).
        \]
    \State Compute confidence-adjusted estimators:
        \[
        \Teins, \Tzwei \gets \textproc{ConfAdj}(\hat \alpha_\phi, \hat \beta_\phi, q_{1-\gamma/2}, M). \qquad
        \text{\Comment{Algorithm \ref{alg:conf-adj}}}\]
    \For{$k = M+1, M+2, \dots, k_{max}$}
        \State Draw $X_k \sim P$, $Y_k \sim Q$.
        \If{$k \bmod M_1 = 0$}
            \State Update $\phi(\cdot)$~.
        \EndIf
        \If{$k \bmod M_2 = 0$}
            \State Update $\hat \alpha_\phi, \hat \beta_\phi$ using $\{X_i, Y_i\}_{i \le k}$.
            \State Update $\Teins, \Tzwei \gets \textproc{ConfAdj}(\hat \alpha_\phi,\hat \beta_\phi, q_{1-\gamma/2}, M)$.
            \If{$\Tzwei < f_{\claim}(\Teins)$}
                \State \Return \texttt{"DP violation"} at time $k$.
            \Else
                \State Continue monitoring.
            \EndIf
        \EndIf
    \EndFor
\EndFunction
\end{algorithmic}
\caption{APT: Advanced Privacy Testing}
\label{alg:seq-audit}
\end{algorithm}

\subsection{Sequential auditing scheme}\label{sec:3:2}
In the following, we verbally describe our new sequential auditing scheme, termed {\em Advanced Privacy Testing,} in short APT (see Algorithm \ref{alg:seq-audit} for a blueprint of APT).  
For convenience, we use the line numbering in Algorithm \ref{alg:seq-audit} to refer to specific steps, following the notation  \vlineref{APT.L.i} (Algorithm \ref{alg:seq-audit}, line number $i$).

Prior to the testing phase, APT generates a small burn-in sample $(X_1, Y_1), \ldots, (X_M, Y_M)$, which is used for parameter calibration \vlineref{APT.L.$2-3$}. This sample is later reused during the testing phase \vlineref{APT.L.12}. Next, APT constructs a classifier $\phi$, using the burn-in sample and (if available) prior knowledge of the analyst \vlineref{APT.L.3}. 
Note that there are very different ways to instantiate the classifier in the APT blueprint and the concrete scenario dictates what a good choice is. 
For instance, in a whitebox scenario, the analyst might know an optimal classifier $\phi$ and may input it to APT without reference to the data. On the opposite end of the spectrum, 
in a blackbox scenario, an initial classifier may need to be constructed from scratch 
based on 
the burn-in sample, and may be refined as additional data become available. 
We refer to the following sections for the concrete classifier choices that instantiate the APT blueprint in our results.

After $\phi$ has been selected, the main task becomes approximation of its classification errors $( \alpha_\phi,  \beta_\phi)$ \vlineref{APT.L.4,12} defined in \eqref{e:class:errors}. Recall that, by definition, the pair $(\alpha_\phi,  \beta_\phi)$ lies on or above the 
tradeoff 
curve $T$. 
Thus if APT's measurements indicate that $( \alpha_\phi,  \beta_\phi)$ is below $f$, we interpret this  
as a signal that $T$ is 
also likely below $f$ and thus $f$-DP is violated.
Note, however, that due to randomness in our samples, the quantities $(\hat \alpha_\phi, \hat \beta_\phi) $ are noisy estimators 
of $( \alpha_\phi,  \beta_\phi)$; therefore, to control false rejections,  we 
 reject $f$-DP only when  $(\hat \alpha_\phi, \hat \beta_\phi) $ lies below  $f$ by a sufficiently large margin \vlineref{APT.L.14-15}---the precise amount is determined by Algorithm \ref{alg:conf-adj} \vlineref{APT.L.5, 13}.
If no $f$-DP violation is detected, the algorithm keeps running \vlineref{APT.L.17} until a maximum number of iterations $k_{\max}$ is reached. In our theory, we may set $k_{\max}=\infty$ (arbitrarily long DP-auditing) and still hold statistical significance guarantees.

The mathematical crux of our approach lies in the confidence adjustment by Algorithm \ref{alg:conf-adj} that allows the control of the false rejection rate. To obtain confidence regions for the empirical classification error, it is standard to use a Hoeffding bound, as e.g. done in \cite{Eureka,askingeneral}. However, these confidence bounds are of fixed-batch type and therefore inappropriate for the more stringent sequential testing setup. Here sequential decision boundaries are needed, which have recently become popular in the literature under the label of "anytime valid inference" (see \cite{howard:ramdas:mcauliffe:sekhon:2021} for an exposition). Our  Algorithm \ref{alg:conf-adj} compares empirical classification errors, as more and more data arrive, to the fluctuations of a Brownian motion. This comparison is based on state-of-the-art tools from Gaussian approximation theory, developed by \cite{berkes:liu:wu:2014}.
Some background on the Brownian motion can be found in the Appendix. For the Brownian motion, one can derive boundary functions that it only crosses with a probability of $\gamma$ - exactly corresponding to our significance level. Using the same boundary functions in Algorithm \ref{alg:conf-adj}  to describe the fluctuations of our empirical classification errors yields a valid sequential inference procedure.

 \begin{figure}[h]
     \centering
     \includegraphics[width=.65\linewidth]{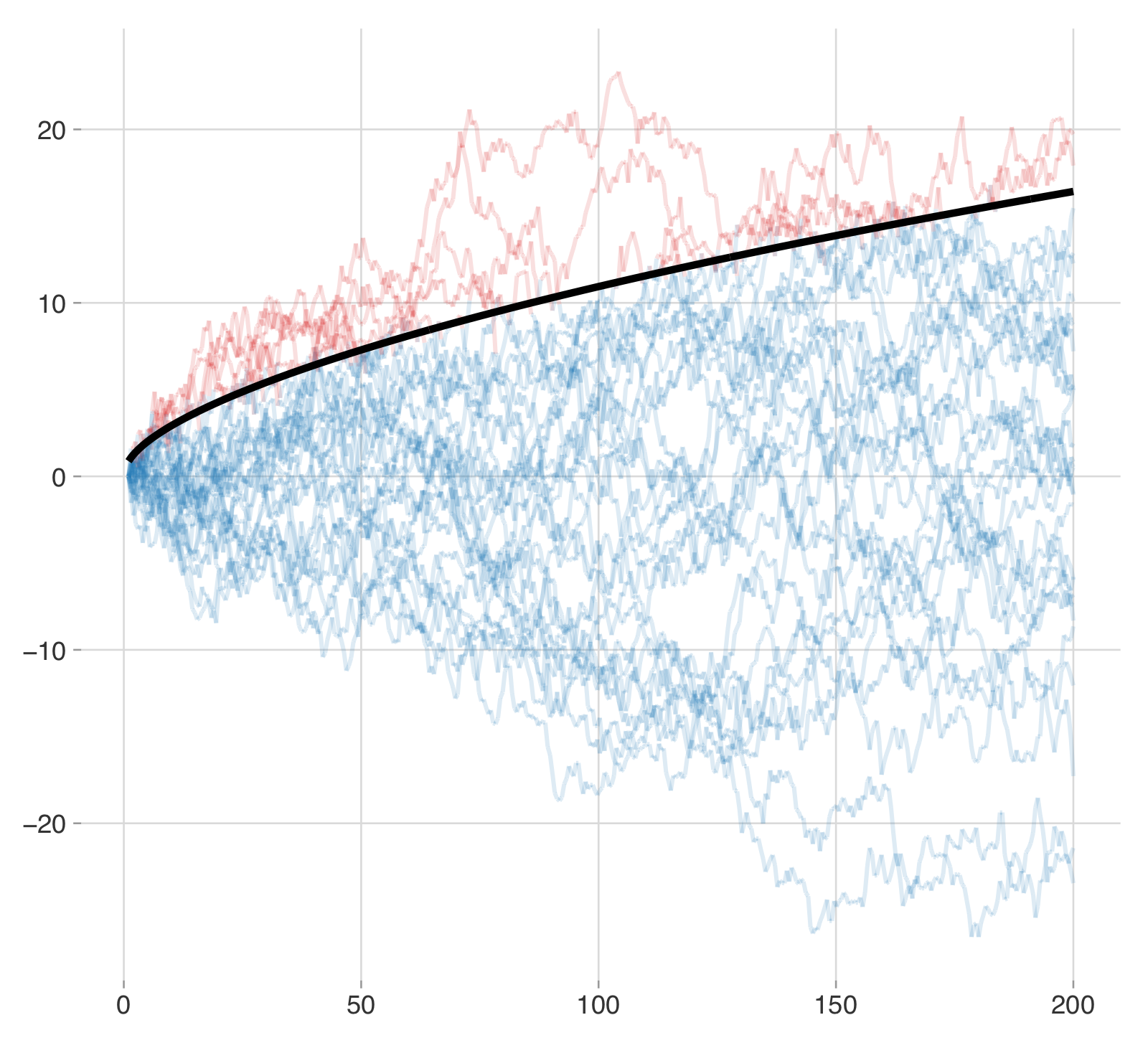}
     \caption{Realizations of the standard Brownian motion on the interval $[1,200]$ compared to a boundary function $g(x) = \sqrt{(x/4)\log(20+x)}$ shown in solid black. Blue paths stay below the boundary function and red paths cross the boundary at some point.}
     \label{fig:bound}
 \end{figure}

\noindent Below, in Theorem \ref{theo:main}, we provide a formal performance guarantee for our approach. The theoretical result is formulated here for one fixed classifier $\phi$, i.e. without updating, a situation most typical in a white-box setting. An advantage of this presentation is its relative simplicity compared to scenarios with updating rules. The proof, presented in the Appendix, Section \ref{sec:proof:main}, is based on classical open-ended Gaussian approximations that are available for i.i.d. random variables. The case of updating rules is more complicated, because it introduces serial dependence. For instance, suppose that at time $k$ the classifier $\hat\phi$ draws on previous data $X_1,\ldots,X_{k-1}$; then the evaluations $\hat\phi(X_i),\hat\phi(X_j)$ for two indices $i \neq j$ are no longer independent, and the distribution of outputs may evolve over time with $\hat\phi$. In the Appendix, Section \ref{sec:adaptive_monitoring}, we present theory for classifiers that evolve over time. The main assumption is that the sequence of classification outcomes remains sufficiently regular for Gaussian approximations to apply (now for weakly dependent data). Practically, this points toward update rules that stabilize over time, so that historical samples have only limited influence on the classifier at later stages. Some examples of such classifiers are discussed in the next section.
The assumption of weak dependence is formalized using classical concepts from time series analysis such as strong mixing. The justification for adaptive classifiers, also given in Section \ref{sec:adaptive_monitoring}, is then much more involved, using recent advances in Prokhorov-metric approximation for partial sum processes of dependent, non-stationary data. We also discuss technical assumptions there in more detail. \\
For now, let us focus on the exemplary case of one fixed $\phi$, which is the subject of Theorem \ref{theo:main} below. The theorem has three statements. The first two, (a)-(b), involve $o(1)$ terms, which are negligible for larger $M$. In practice, moderate values such as $M=50$ are sufficient for a tight approximation. The parameter $\gamma \in (0,1)$ is the user-determined significance level and practically fixed at some small value, such as $5\%$ (see Section \ref{sec:2:2}). According to parts (a) and (b), $\gamma$ bounds the risk of a false rejection if $f_{\claim}$-DP actually holds. Smaller values of $\gamma$ ensure fewer false rejections, but come at the cost of slower detection of actual privacy violations. 
Part (c) holds for any $M$ and states that a privacy violation will be detected eventually (for large enough $k$) with certainty.

\begin{theo} \label{theo:main}
    Let $\phi(\cdot)$ be a classifier and $\gamma \in (0,1)$ a significance level.  
    Then the following statements hold:
    \begin{itemize}
        \item[a)] If $(\alpha_\phi, \beta_\phi) \in (0,1)^2$ satisfies $f_{\claim}(\alpha_\phi)<\beta_\phi$ (DP claim satisfied by a positive margin), it follows that 
        \[
       \mathbb{P} \big(\textnormal{APT outputs for some $k$ \texttt{"DP violation"}}\big) =o(1).
        \]
        \item[b)] If $(\alpha_\phi, \beta_\phi) \in (0,1)^2$ satisfies $f_{\claim}(\alpha_\phi)=\beta_\phi$ (DP claim exactly satisfied), it follows that 
        \[
        \mathbb{P} \big(\textnormal{APT outputs for some $k$ \texttt{"DP violation"}}\big) = \gamma+o(1).
        \]
        \item[c)] If $(\alpha_\phi, \beta_\phi) \in (0,1)^2$ satisfies $f_{\claim}(\alpha_\phi)>\beta_\phi$ (DP claim violated), it follows that 
        \[
       \mathbb{P} \big(\textnormal{APT outputs for some $k$ \texttt{"DP violation"}}\big) = 1.
        \]
    \end{itemize}
\end{theo}

\begin{rem}\label{rem:theo}  Theorem \ref{theo:main} implies that a violation of $f-$DP is detected by the APT algorithm eventually (i.e. if samples are large enough). In the following we discuss how large the samples required by APT are, compared to the theoretical optimum -- a fixed-batch test with optimal sample size. Recall from our above discussion that using such an optimal test is practically impossible, because the effect size is unknown. In summary, we see that the samples used by APT are within a logarithmic factor of the theoretical optimum.

Let us now suppose that $f-$DP is violated and that we have a candidate classifier $\phi$, with $(\alpha_\phi, \beta_\phi)$ below the claimed curve $f_{\claim}$. The effect size (Section \ref{sec:2:2}) is the distance between $(\alpha_\phi, \beta_\phi)$ and  $f_{\claim}$, and we denote it here by $\delta$. 
    An ideal test would now be based on a confidence region around $(\alpha_\phi, \beta_\phi)$, e.g. using a Hoeffding bound (or $t$-test) with sample size dependent on $\delta$. Some simple calculations show that the minimal required sample size to detect the violation would have to be of size $\mathcal{O}(\delta^{-2})$ as $\delta$ becomes small. 
    Our sequential approach yields similar values. Indeed, depending on the relation of $\delta$ and $M$, it follows that the sample size in the sequential testing is at best $\mathcal{O}_P(\delta^{-2})$ (equal to minimum up to a constant) and at worst
     $\mathcal{O}_P(\delta^{-2}\log(\delta^{-1}))$ (within a log-factor of the minimum). We thus see that there is only a minimal sampling overhead of sequential testing compared to the (infeasible) mathematically optimal procedure.
\end{rem}



\subsection{Our classifiers suite}

We next describe how 
we instantiate the APT blueprint with 
(close to) optimal classifiers in different relevant scenarios. In a nutshell, our classifiers adapt the Neyman-Pearson Lemma from classical statistics ({\em cf.} \cite[p.60]{lehmann:romano:2005}) to the different scenarios. More concretely: suppose that $P$ and $Q$ have probability densities, denoted by $p$ and  $q$, respectively. The Neyman-Pearson Lemma defines the Likelihood Ratio approach: For constants $\lambda \in [0,1]$ and $\eta>0$, and a random bit $B\sim \text{Ber}(\lambda)$, it takes the form
\begin{equation}\label{eq:classifier}
\phi_{\eta,\lambda}(x) := \begin{cases}
    0, \quad q(x)/p(x) < \eta\\
    0, \quad q(x)/p(x) =\eta \,\, \& \,\, B=0\\
    1, \quad \textnormal{else}. 
\end{cases}
\end{equation}

Adapting the lemma to the $f$-DP formulation and terminology it states that any $\phi$ of the above form is optimal, in the sense that the error pair $(\alpha_\phi, \beta_\phi)$ lies exactly on the tradeoff curve $T$. For $f$-DP auditing, it would therefore suffice to consider classifiers from the set $\{\phi_{\eta,\lambda}\}_{\eta,\lambda}$. Obtaining such optimal classifiers is different in blackbox and whitebox settings, and we discuss both cases below: 

\paragraph{Blackbox Approach} 
In a blackbox setting, employing optimal classifiers faces two main obstacles. 
First, 
the densities $p$ and $q$ are unknown, which means that the corresponding optimal classifiers $\phi_{\eta,\lambda}$ in \eqref{eq:classifier} are unknown as well. We obtain feasible substitutes for $p$ and $q$ by estimation - in our experiments, we employ kernel density estimation although other approaches (for instance, neural network based estimators) are equally suitable. 
Density estimates are updated sequentially as more data become available.

Second, even if the entire family $\{\phi_{\eta,\lambda}\}_{\eta,\lambda}$ were somehow available, one would still need to decide which specific pair $(\eta,\lambda)$ to use. 
For most models, $\eta$ is the decisive parameter. For example, in the Gaussian case discussed below, the choice of $\lambda$ is entirely irrelevant. Therefore, we focus here on $\eta$ exclusively, setting $\lambda=0$ and suppressing dependence on $\lambda$ in our notation with $\phi_\eta := \phi_{\eta,0}$. Adaptations of our method that include $\lambda \neq 0$ are probably possible, but they seem to be less relevant for DP-auditing.
Each value of $\eta$ can be associated with a unique operating point $(\alpha_{\phi_\eta}, \beta_{\phi_\eta})$, and in Section \ref{sub:crit_eta} below, we discuss a simple geometric strategy to choose the optimal $\eta$.

The corresponding methods are summarized below in Algorithm \ref{alg:PLRT} (estimation of $T$),
 Algorithm \ref{alg:eta_star_45} (selection of optimal $\eta$) and can then be integrated as subroutines into the APT Algorithm \ref{alg:seq-audit}. 

\paragraph{Whitebox Approach}
If the analyst has additional prior knowledge about the distributions $P$ and $Q$, we may depart from the non-parametric blackbox methods above and instead adopt a parametric model that offers higher sample efficiency. A classical assumption is that the outputs under $P$ and $Q$ are (approximately) Gaussian with a common variance $\sigma^{2}$ and different means $\mu_{P}$ and $\mu_{Q}$ (see e.g.\ \cite{carlini2022membership}). Under this model, the Likelihood Ratio classifier is equivalent to a simple threshold classifier on the scalar output $s$: $\phi_{\eta}(s) = \mathbf{1}\{s \ge \eta\},$
namely, we decide $P$ if $s \ge \eta$ and $Q$ otherwise. The appeal of this approach is that, under the Gaussian equal-variance model, the optimal Likelihood Ratio classifier depends only on the two means (after standardization by the common variance). Moreover, even when the true output distributions deviate moderately from normality, the Gaussian plug-in Likelihood Ratio classifier often remains a powerful tool for auditing purposes.
As before, we can use Algorithm \ref{alg:eta_star_45} to detect a critical $\eta$.

\subsection{Likelihood Ratio tuning}\label{sub:crit_eta}

In the previous section, we have focused on approximating density ratios $p(x)/q(x)$ (either non-parametric or parametric), used for the optimal classifier from the Neyman-Pearson Lemma. We will now focus on the choice of the optimal threshold $\eta^*$ in Eq.~\eqref{eq:classifier}.
 Recall that any $\eta$ implies a Likelihood Ratio classifier $\phi_\eta$, and to emphasize our focus on $\eta$, define
\[
\alpha_{\phi_\eta} =: \alpha(\eta), \qquad \beta_{\phi_\eta}=:\beta(\eta).
\]
Our aim is to select a classifier that maximizes our power against privacy violations. A first intuition for this is to select a pair $(\alpha(\eta), \beta(\eta))$ that lies as far as possible under the claimed privacy curve $f_{\claim}$. This intuition was used by \cite{askingeneral}, where the selection rule
\begin{equation}\label{eq:old_eta_max}
    \eta^* := \argmax_{\eta} 
    f_{\claim}(\hat \alpha(\eta)) - \hat{\beta}(\hat \alpha(\eta))
    ~,
\end{equation}
was suggested, where $\hat{\beta}(\cdot)$ is an estimator of the true unknown tradeoff curve. In \eqref{eq:old_eta_max}, $\eta^*$ is selected so that $(\alpha(\eta), \beta(\eta))$ lies as far as possible under $f_{\claim}$ in vertical direction. In the following, we use a modification of this rule that significantly improves the detection of DP violations. 

\begin{figure}[htp]
\centering
\includegraphics[width=0.6\linewidth]{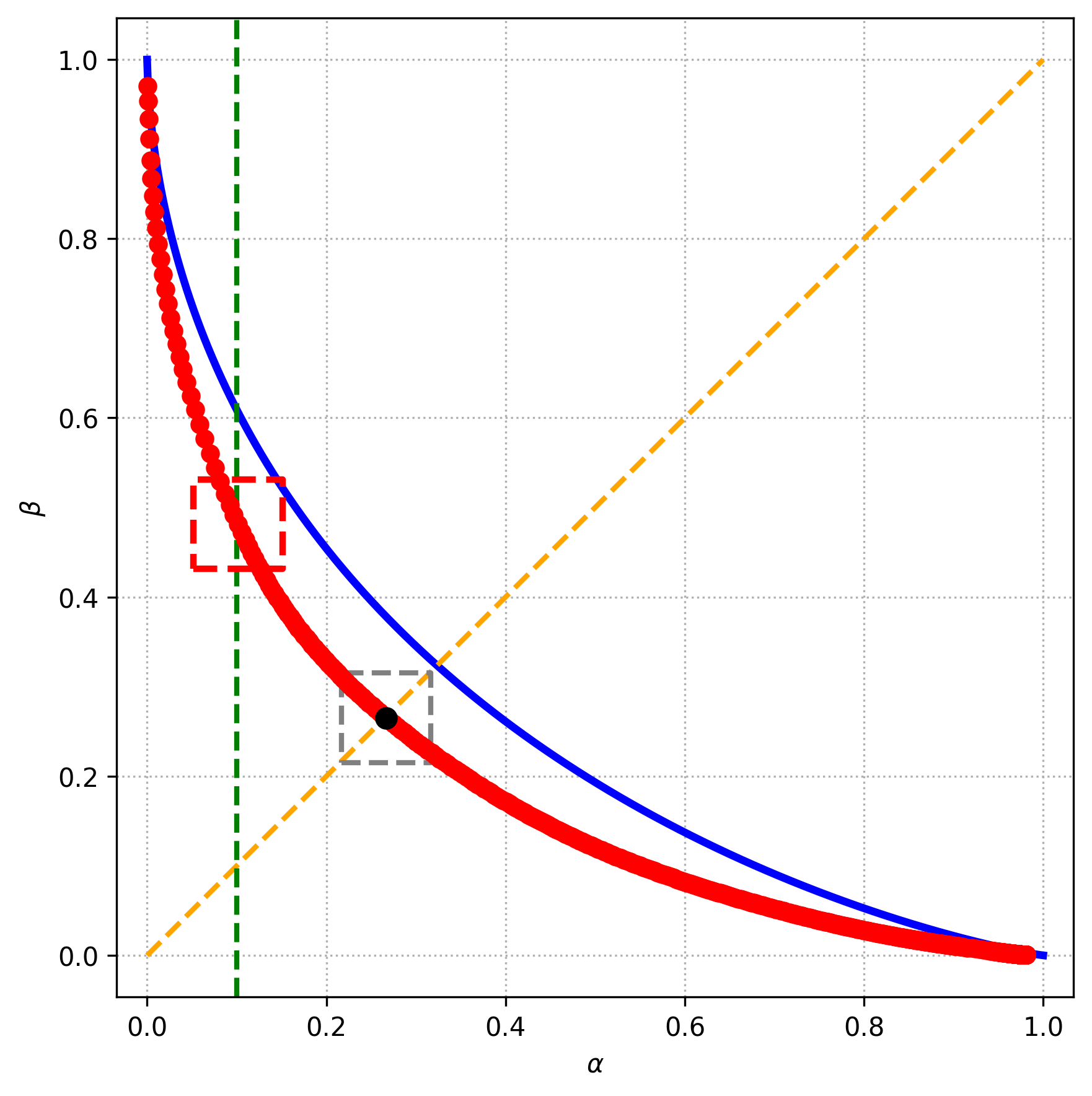}
    \caption{The blue line corresponds to $f_{\claim}$, while the red  line represents the estimated tradeoff curve $\hat T$. The vertical green line and diagonal yellow line correspond to the rules of selecting $\eta$ in \eqref{eq:old_eta_max} and \eqref{eq:new_eta_max}, respectively.\label{fig:eta_crit} The rule \eqref{eq:new_eta_max} maximizes the size of the square that can be fit below $f_{\claim}$ -- the central quantity for detecting DP-violations.} 
\end{figure}

For our modification to make sense, notice that detecting DP violations does not directly depend on the one-dimensional, vertical distance between  $(\alpha(\eta), \beta(\eta))$ and  $f_{\claim}$. Rather, it depends on a two-dimensional quantity, and a simple proxy is the largest axis-aligned square around the center $(\alpha(\eta), \beta(\eta))$, that still fits below  $f_{\claim}$. The upper-right corner of this square lies on the $45^\circ$ direction, from $(\alpha(\eta),\beta(\eta))$. This motivates selecting a point that is well separated from $f_{\claim}$ along this diagonal direction, rather than only in the vertical direction as in~\eqref{eq:old_eta_max}.
To make this geometric idea mathematically precise, consider for any estimate $(\hat{\alpha}(\eta), \hat{\beta}(\eta))$ the line
\[
\beta(\alpha) = \hat{\beta}(\eta) + \big(\alpha - \hat{\alpha}(\eta)\big)~,
\]
which passes through $(\hat{\alpha}(\eta), \hat{\beta}(\eta))$ with slope $1$. We can find the intersecting point of this line and $f$ by solving the equation 
\[
f_{\claim}(\alpha) = \hat{\beta}(\eta) + \big(\alpha - \hat{\alpha}(\eta)\big)
\]
for $\alpha$. Proceeding this way, we can for each $\eta$ compute the Euclidean distance between $(\hat{\alpha}(\eta), \hat{\beta}(\eta))$ and its intersection point $(\alpha_{\text{int}}(\eta), f_{\claim}(\alpha_{\text{int}}(\eta)))$.
The critical $\eta^*$ is now the following: 
\begin{equation}\label{eq:new_eta_max}
    \eta^* := \argmax_{\eta} \|(\hat{\alpha}(\eta), \hat{\beta}(\eta))-(\alpha_{\text{int}}(\eta), f_{\claim}(\alpha_{\text{int}}(\eta)))\|_2 
    ~.
\end{equation} 

In Figure \ref{fig:eta_crit} we illustrate our general idea. Our new approach selects a point that maximizes the size of the square under the tradeoff curve (which is directly linked with more power to detect DP violations). In contrast, the method by \cite{askingeneral} selects a less appropriate point, and we see that the same square now intersects with the curve $f_{\claim}$ (it does not fit under it anymore).

%% file: 04-experiments.tex
\section{Experiments}\label{sec:4}


In this section, we evaluate the proposed sequential auditor. We first consider controlled ground-truth benchmarks in both blackbox and whitebox settings. In the blackbox setting, the analyst observes only mechanism outputs, whereas in the whitebox setting, the analyst additionally has prior knowledge about the mechanism, such as the class of noising distributions. We then compare our method with the sequential $(\varepsilon,\delta)$-DP auditor proposed by~\cite{gonzalez2025sequentially}. Finally, we demonstrate in a realistic one-run DP-SGD canary audit, how our sequential pipeline can reduce the number of canary scores needed for auditing. We now formulate the common setup for the ground-truth experiments below.

\paragraph{Setup for controlled ground-truth experiments.}
Our experimental framework follows the well-established and accepted methodology for estimating and auditing various notions of differential privacy (see e.g. \cite{lowerbounds, askingeneral}): 
Throughout, we consider neighboring databases $D, D' \in [0,1]^r$ with $r = 10$ records. A standard choice in the literature, which maximizes the $\ell_1$-distance while preserving adjacency, is
\begin{equation}\label{eq_databases}
    D = (0,\ldots,0), 
    \qquad 
    D' = (1,0,\ldots,0),
\end{equation}
though similar behavior is observed for other adjacent pairs. For the choice of databases, we also refer to \cite{lowerbounds}.  
We denote by $S(x) = \sum_{i=1}^{10} x_i$ the summary statistic used in our additive-noise experiments.

\paragraph{Mechanisms}\label{sec:algorithms}
We evaluate three mechanisms commonly studied in privacy research: the Gaussian mechanism, the Laplace mechanism, and DP-SGD. These mechanisms produce quite different output distributions and therefore provide a meaningful benchmark. Below we briefly describe each mechanism together with the parameter settings used in our experiments. We also provide, for each mechanism, a mathematically closed form for the tradeoff curve, that can serve as ground truth in our below experiments.

\paragraph{Additive Noise Mechanisms}
For both the Gaussian and Laplace mechanisms, the released value takes the form
\[
    \mathcal{M}(x) = S(x) + Y,
\]
where $Y \sim \mathrm{Lap}(0,b)$ or $Y \sim \mathcal{N}(0,\sigma^2)$.  
When the constants $b, \sigma$ are appropriately calibrated, both mechanisms satisfy $f$-DP.  

For the Gaussian mechanism, we fix $\sigma = 1$.  
Following \cite{dong:roth:su:2022}, the corresponding tradeoff curve is then given by the closed form
\[
    T_{\text{Gauss}}(\alpha)
    = \Phi\!\bigl(\Phi^{-1}(1-\alpha) - \mu \bigr),
\]
where $\Phi$ is the cumulative distribution function of the standard normal. We notice that $T_{\text{Gauss}}(\alpha)$
 corresponds to the tradeoff curve between $\mathcal{N}(0,1)$ and $\mathcal{N}(\mu,1)$.

For the Laplace mechanism, we fix $b = 1$ and use another closed form derived by \cite{dong:roth:su:2022}, namely
\[
    T_{\text{Lap}}(\alpha)
    =
    \begin{cases}
        1 - e^{\mu}\alpha, & \alpha < e^{-\mu}/2, \\[4pt]
        e^{-\mu} / (4\alpha), & e^{-\mu}/2 \le \alpha \le 1/2, \\[4pt]
        e^{-\mu}(1-\alpha), & \alpha > 1/2.
    \end{cases}
\]

\paragraph{DP-SGD}
Deriving a closed-form expression for the tradeoff curve of DP-SGD is in general challenging and often infeasible.  
However, in a special case, \cite{askingeneral} have derived a closed form that we can use as ground truth for our validation.

First, recall that DP-SGD aims to privately approximate the solution to the empirical minimization problem
\[
    \theta^* = \argmin_{\theta \in \Theta} \mathcal{L}_x(\theta),
    \qquad
    \mathcal{L}_x(\theta)
    =
    \frac{1}{r} \sum_{i=1}^{r} \ell(\theta, x_i).
\]
In our experiments, we choose $\ell(\theta, x_i) = \tfrac{1}{2}(\theta - x_i)^2$, initialize $\theta_0 = 0$, and set $\Theta = \mathbb{R}$.  
The remaining hyperparameters are fixed to $\sigma = 0.2$, $\rho = 0.2$, $\tau = 10$, and $m = 5$. Then, according to \cite{askingeneral}, the tradeoff curve is given by
\[
    T_{\text{SGD}}(\alpha)
    =
    \sum_{I \subset \{1,\ldots,\tau\}}
    \frac{1}{2^\tau}\,
    \Phi\!\left(
        \Phi^{-1}(1-\alpha)
        -
        \frac{\mu_I}{\bar\sigma}
    \right),
\]
where
\begin{align}\label{mu_I}
    \mu_I := \sum_{t \in I} \frac{\rho (1 - \rho)^{\tau - t}}{m}, \qquad
    \bar{\sigma}^2 = \rho^2 \sigma^2
    \frac{1 - (1 - \rho)^{2 \tau}}{1 - (1 - \rho)^{2}}.
\end{align}

\paragraph{Experimental Parameters}
We use the same global parameters across all mechanisms:  
the refresh period $M_1$ is chosen adaptively (see Algorithm \ref{alg:audit_should_refit} with threshold $0.1$), evaluation period $M_2 = 10$ is static, maximum number of iterations $k_{\max} = 10,000$, and significance level $\gamma = 0.05$. All empirical rejection probabilities are derived using $1,000$ repetitions.

\smallskip
\noindent\textbf{Gaussian}
We introduce privacy violations by miscalibrating the privacy parameter $\mu$.  
Data are generated under $\mu = 1$ and we audit for the claimed curves when $\mu = 0.5,\ 0.8,\ 1$.

\smallskip
\noindent\textbf{Laplace}
We use the same strategy as for the Gaussian mechanism. The true calibration uses $\mu=1$, while we audit under the claims $\mu = 0.5,\ 0.8,\ 1$.

\smallskip
\noindent\textbf{DP-SGD}
For DP-SGD, instead of $\mu$, we vary the number of iterations $\tau$.  
The true mechanism runs with $\tau=10$, while we audit under the claims $\tau = 5,\ 7,\ 10$.

\medskip
Summing up, for each mechanism we thus consider three scenarios:
\emph{strong violation}, \emph{subtle violation}, and \emph{no violation}.

\subsection{Blackbox}\label{sec:4.1}

We now evaluate the three mechanisms under the blackbox setting.  
Before proceeding, we briefly recall the construction of the classifier.  
Algorithm~\ref{alg:PLRT} describes the perturbed Likelihood Ratio estimator introduced in \cite{askingeneral}, which approximates the optimal Likelihood Ratio classifier.  
In addition, we derive the critical decision rule $g$ required by the sequential auditing procedure in Algorithm~\ref{alg:seq-audit}.  
The complete blackbox auditing pipeline is summarized in Algorithm~\ref{alg:seq-audit} and by setting (cf. Algorithm \ref{alg:build_classifier_general})
$$\mathcal L(\{X_i\}_{i=1}^{m}, \{Y_i\}_{i=1}^{m})=\log(\hat p)-\log(\hat q)~,$$
where $\hat q$ and $\hat p$ are KDEs for the unknown densities $p,q$.
\paragraph{Blackbox Parameters}
For the KDE-based blackbox procedure, we introduce a few additional parameters.  
We set the maximum threshold value to $\eta_{\max} = 15$ (i.e. $\log(\eta_{\max})\approx 2.7$) and the perturbation level in Algorithm~\ref{alg:PLRT} to $h = 0.1$.  
Bandwidths for the KDEs are selected in a data-adaptive manner using the package \cite{kernsmooth}.
\subsection{Whitebox} \label{sec:4.2}
In the whitebox setting we focus on the Gaussian mechanism. For simplicity, we fix $\sigma=1$. 
We instantiate \textproc{BuildClassifier} with a parametric learning rule $\mathcal L_{\mathrm{Gauss}}$ that fits the means
\[
\hat\mu_P := \frac{1}{m}\sum_{i=1}^m X_i, 
\qquad 
\hat\mu_Q := \frac{1}{m}\sum_{i=1}^m Y_i,
\]
and uses the analytic Gaussian tradeoff. We consider an equidistant threshold grid $\mathcal T \subset[\min(X,Y),\max(X,Y)]~.$ 
For each $\eta\in\mathcal T$, the induced classifier is the threshold rule
\[
\phi_\eta(x) := \mathbf{1}\{x\ge \eta\}.
\]
Moreover, the corresponding estimated tradeoff point is available in closed form
\[
\hat\alpha(\eta) = 1-\Phi(\eta-\hat\mu_P),
\qquad
\hat\beta(\eta) = \Phi(\eta-\hat\mu_Q).
\]
Finally, \textproc{BuildClassifier} selects $\eta^*$ by applying \eqref{eq:new_eta_max} to the curve
$\{(\hat\alpha(\eta),\hat\beta(\eta)):\eta\in\mathcal T\}$ and returns $\phi(\cdot)=\phi_{\eta^*}(\cdot)$.

\begin{figure*}[!t]
    \centering
    \caption*{\textbf{Empirical Rejection Rates}} 

    \setlength{\tabcolsep}{2pt} 
    \renewcommand{\arraystretch}{1.1} 

    \begin{tabular}{|c|c|c|c|}
    \hline
     & \textbf{Large Violation} & \textbf{Small Violation} & \textbf{No Violation} \\
    \hline
    \rotatebox{90}{\hspace{16mm} \textbf{Gaussian}} &
    \begin{subfigure}{0.315\textwidth}
        \centering
        \includegraphics[width=\linewidth]{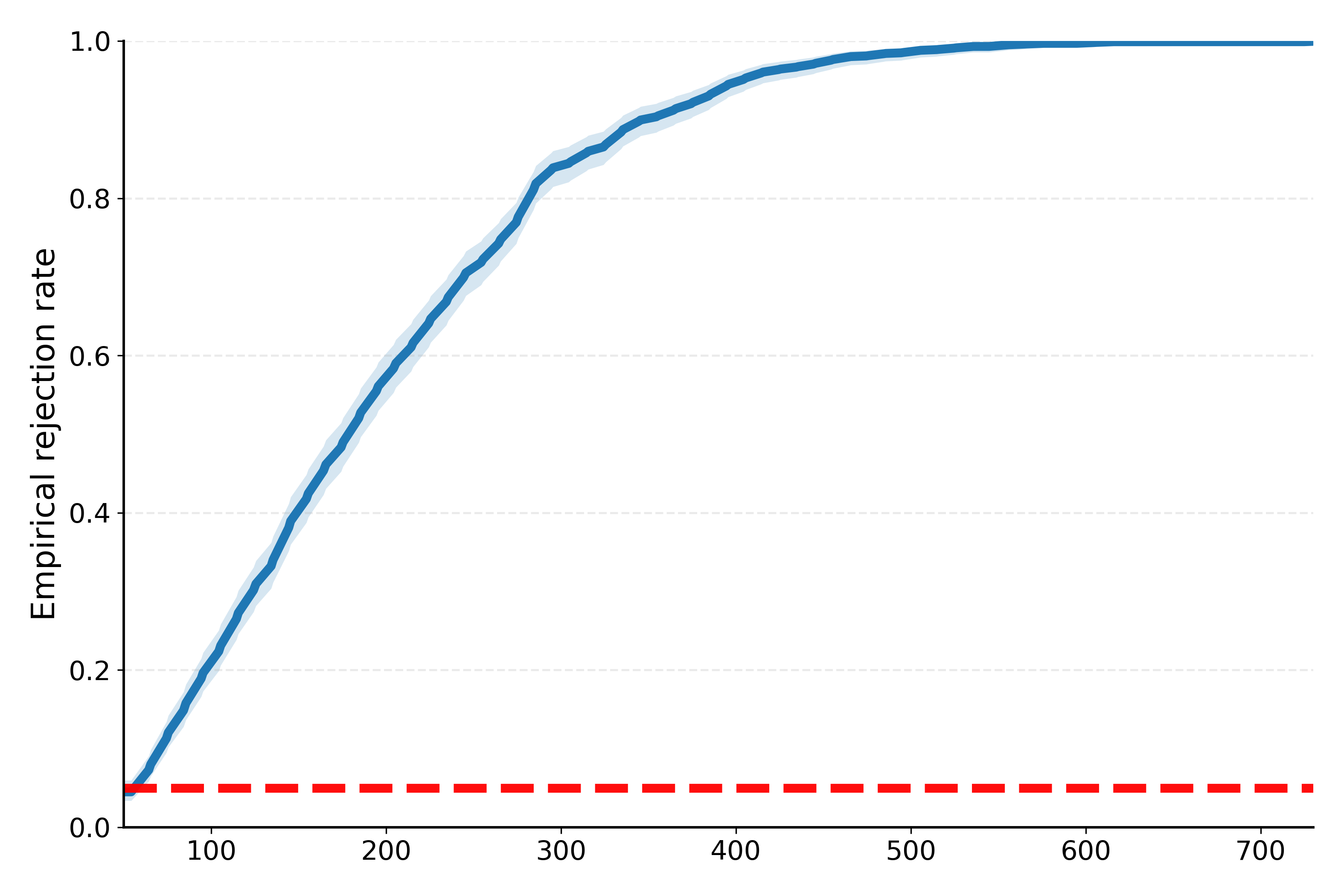}
        \subcaption*{
            $\mu = 0.5$ \\
            Whitebox
        }
    \end{subfigure} &
    \begin{subfigure}{0.315\textwidth}
        \centering
        \includegraphics[width=\linewidth]{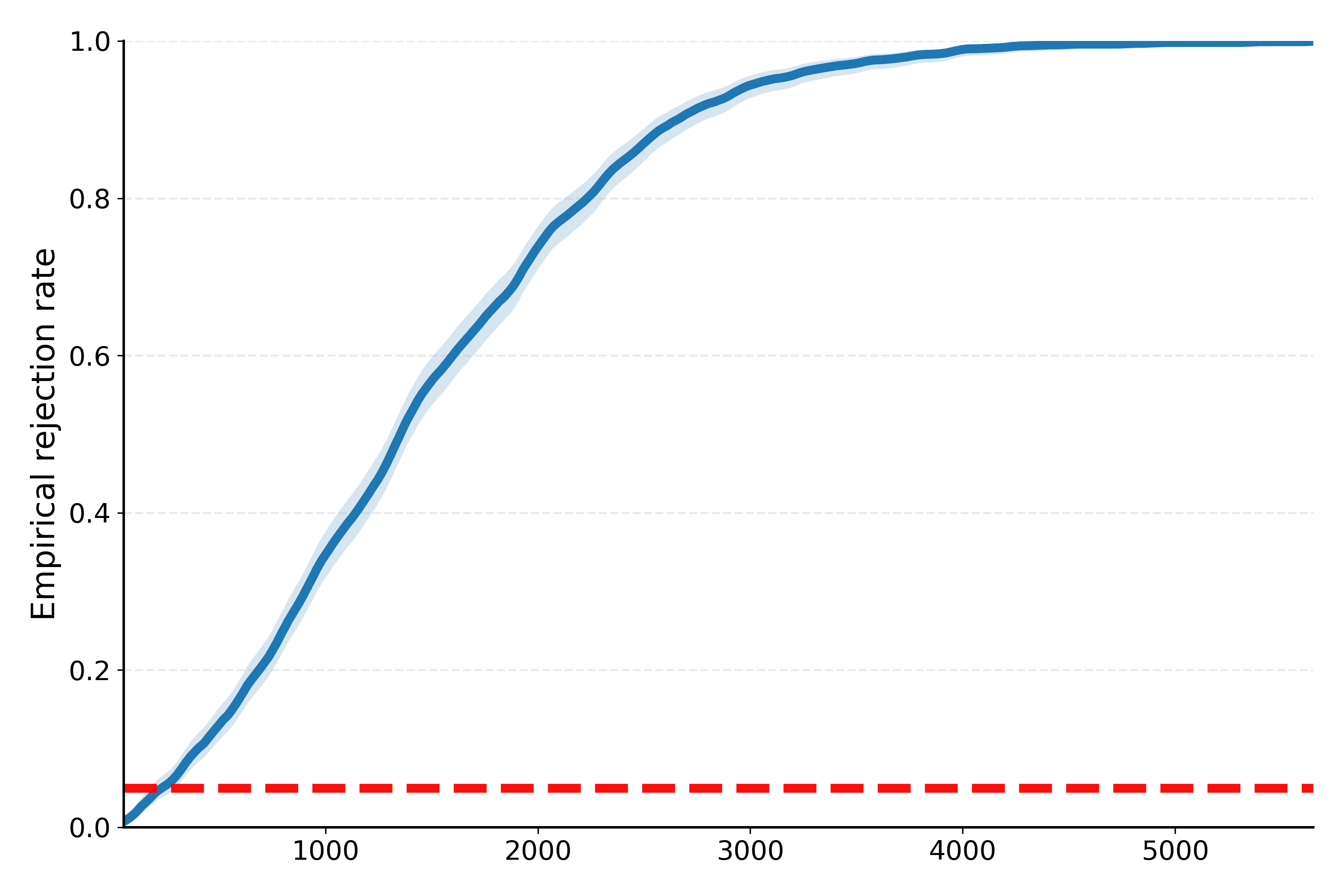}
        \subcaption*{
            $\mu = 0.8$ \\
            Whitebox
        }
    \end{subfigure} &
    \begin{subfigure}{0.315\textwidth}
        \centering
        \includegraphics[width=\linewidth]{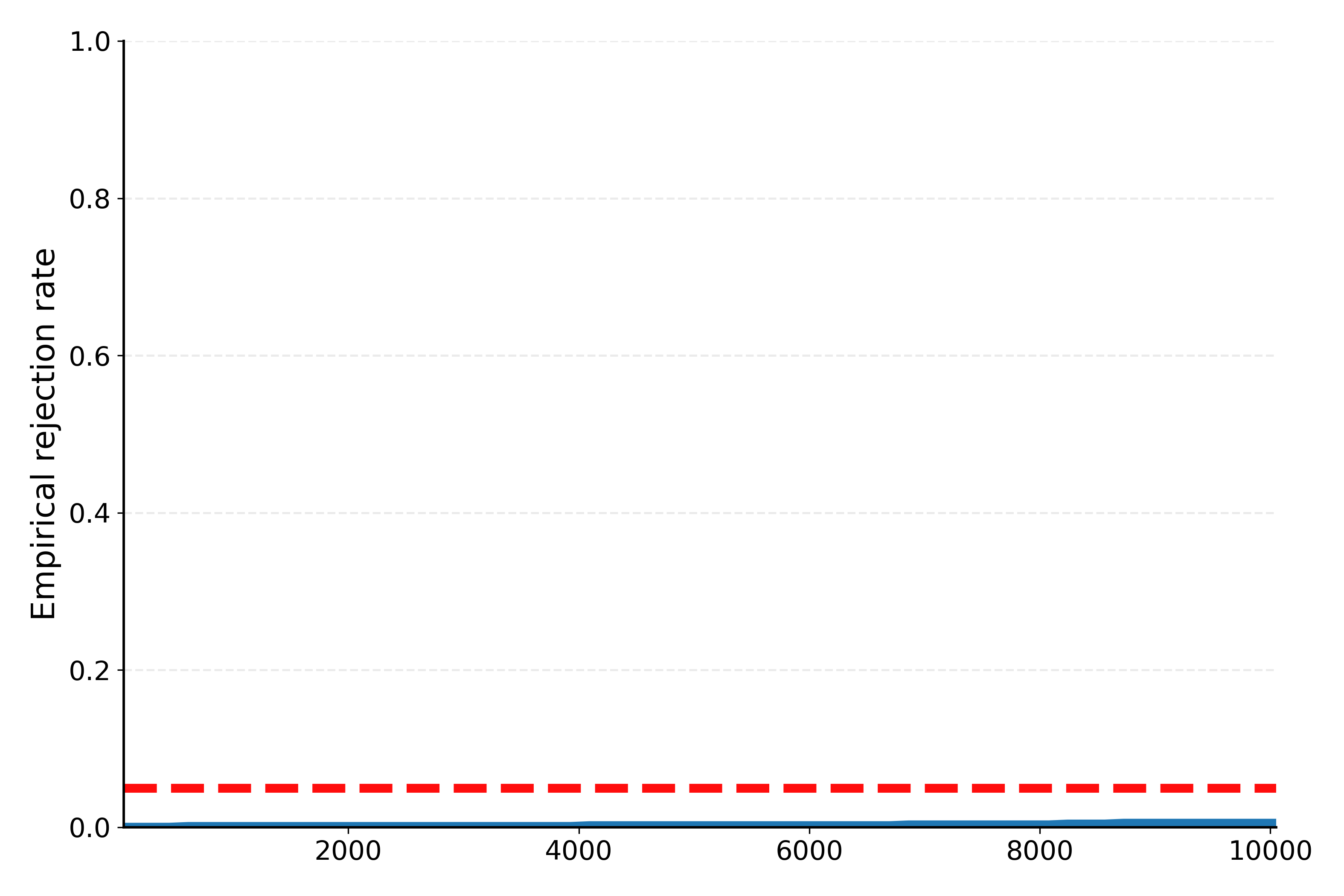}
        \subcaption*{
            $\mu = 1$ \\
            Whitebox
        }
    \end{subfigure} \\
    \hline
    \rotatebox{90}{\hspace{16mm}\textbf{Gaussian}} &
    \begin{subfigure}{0.315\textwidth}
        \centering
        \includegraphics[width=\linewidth]{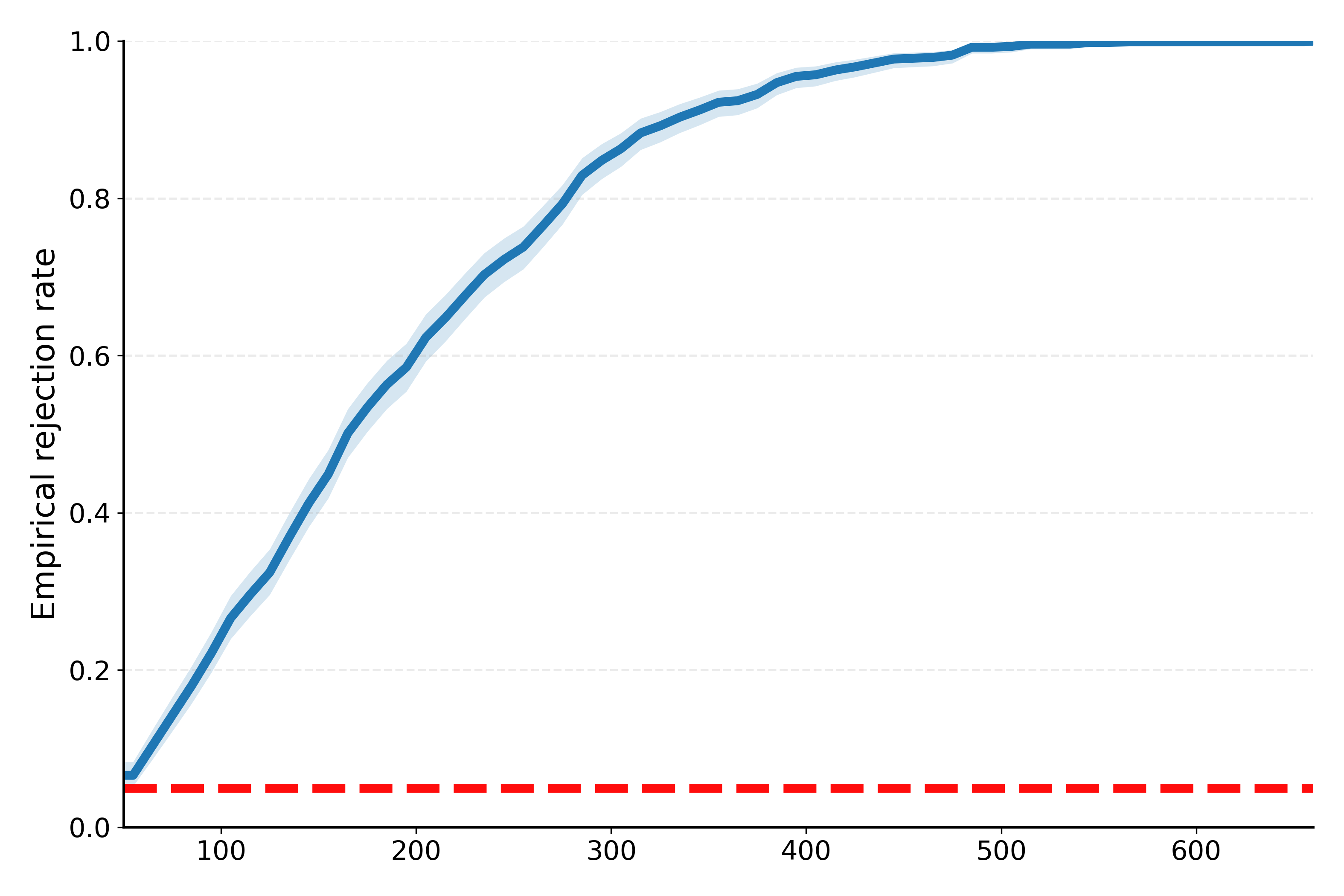}
        \subcaption*{
            $\mu = 0.5$ \\
            Blackbox
        }
    \end{subfigure} &
    \begin{subfigure}{0.315\textwidth}
        \centering
        \includegraphics[width=\linewidth]{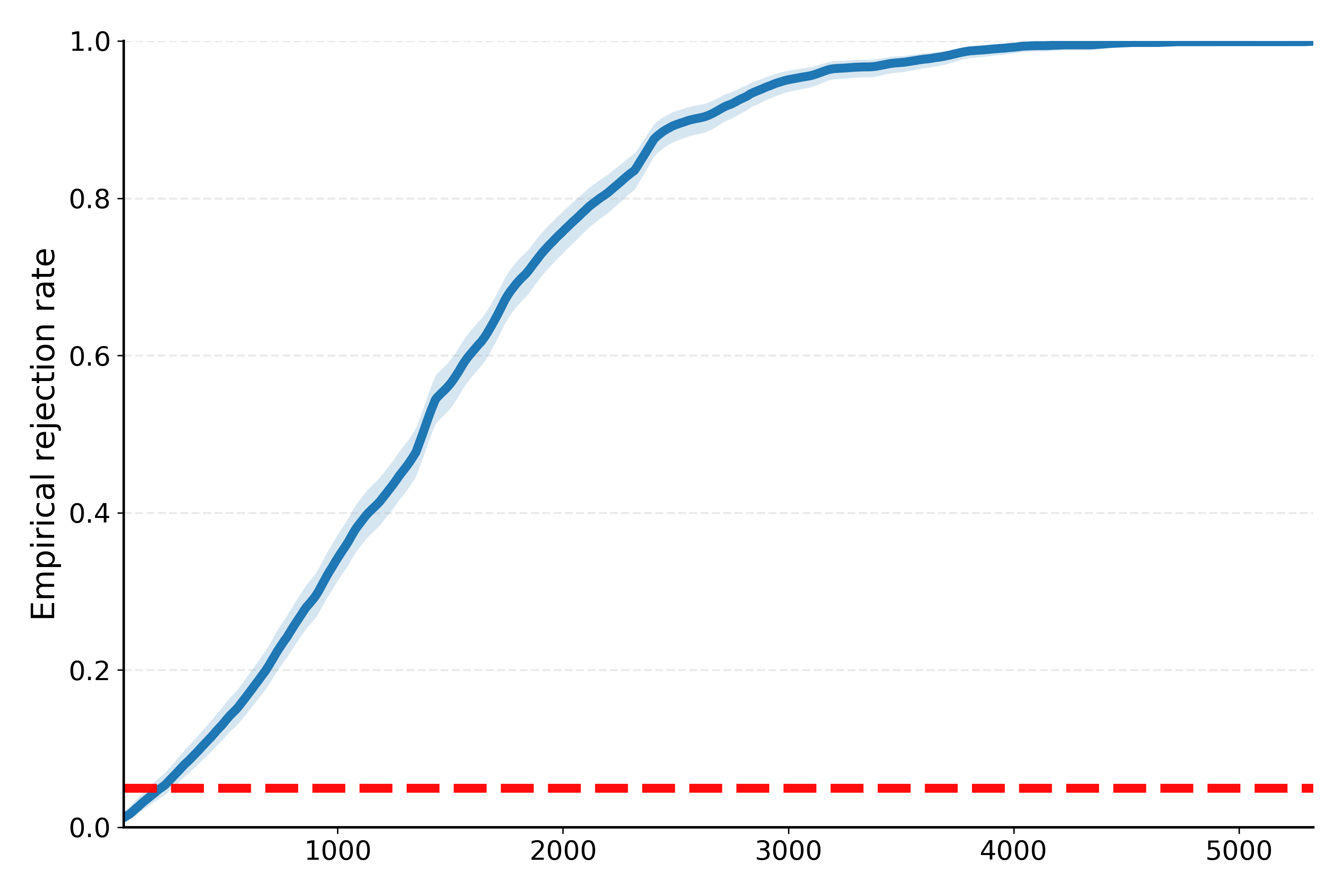}
        \subcaption*{
            $\mu = 0.8$ \\
            Blackbox
        }
    \end{subfigure} &
    \begin{subfigure}{0.315\textwidth}
        \centering
        \includegraphics[width=\linewidth]{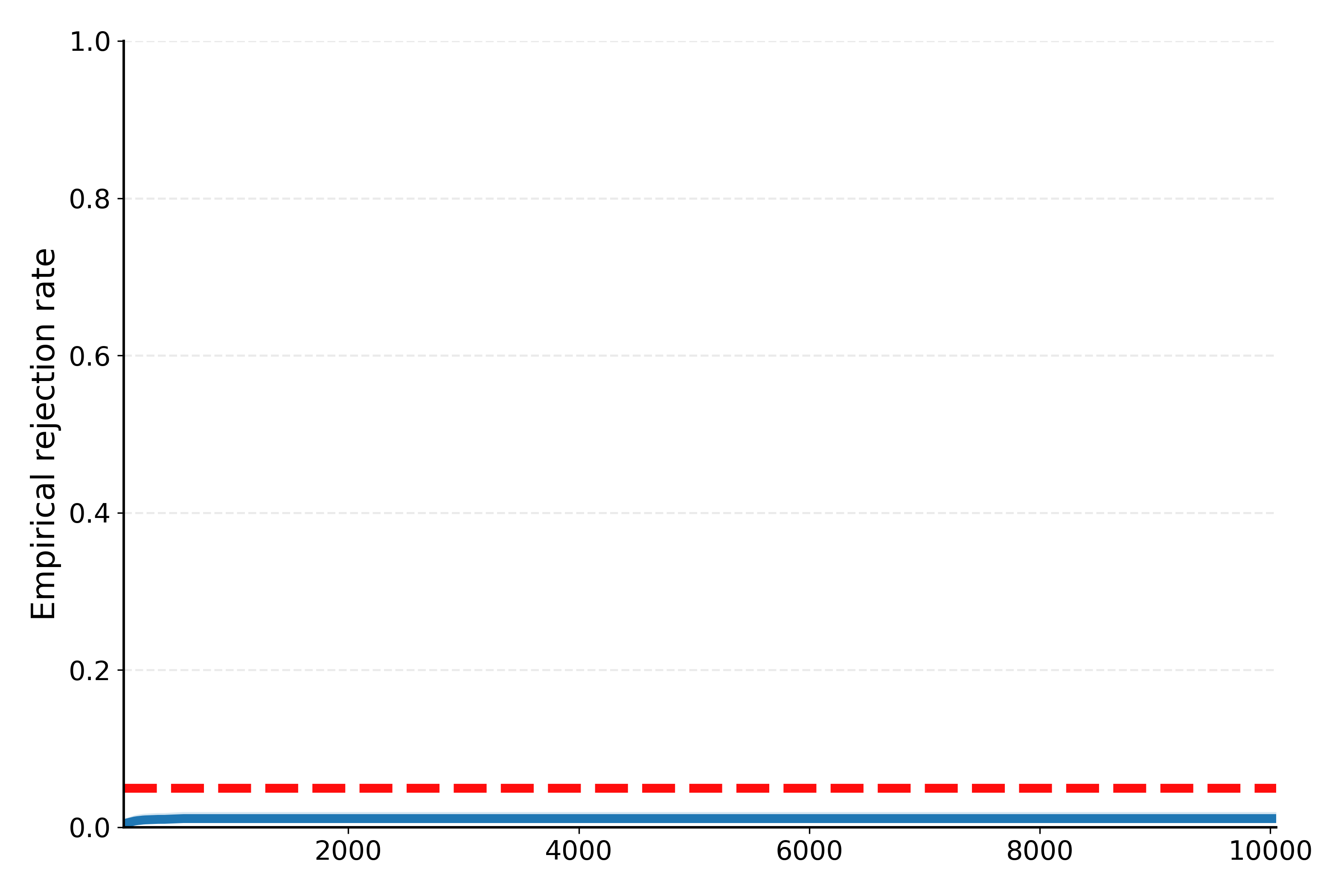}
        \subcaption*{
            $\mu = 1$ \\
            Blackbox
        }
    \end{subfigure} \\
    \hline
    \rotatebox{90}{\hspace{16mm}\textbf{Laplace}} &
    \begin{subfigure}{0.315\textwidth}
        \centering
        \includegraphics[width=\linewidth]{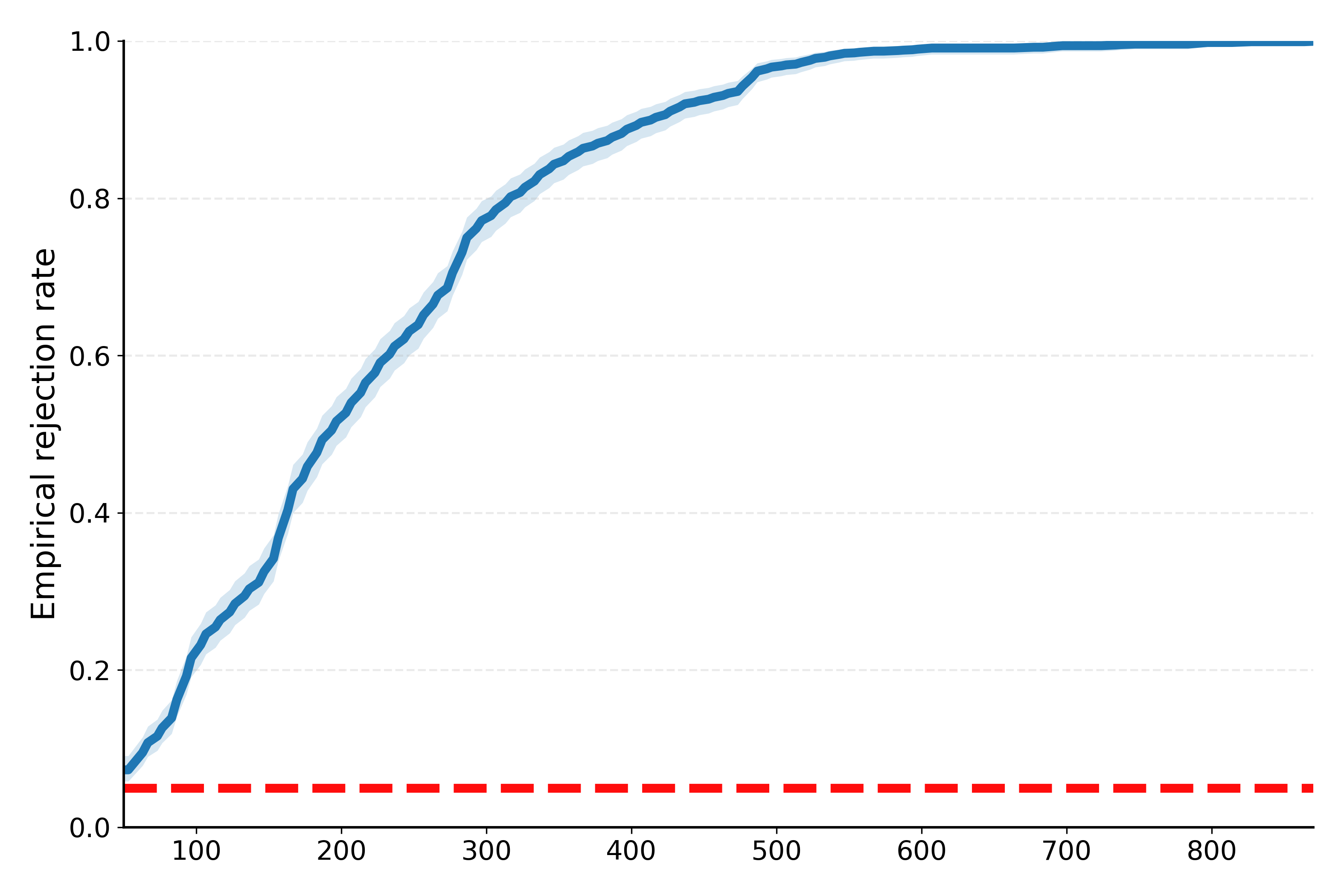}
        \subcaption*{
            $\mu = 0.5$ \\
            Blackbox
        }
    \end{subfigure} &
    \begin{subfigure}{0.315\textwidth}
        \centering
        \includegraphics[width=\linewidth]{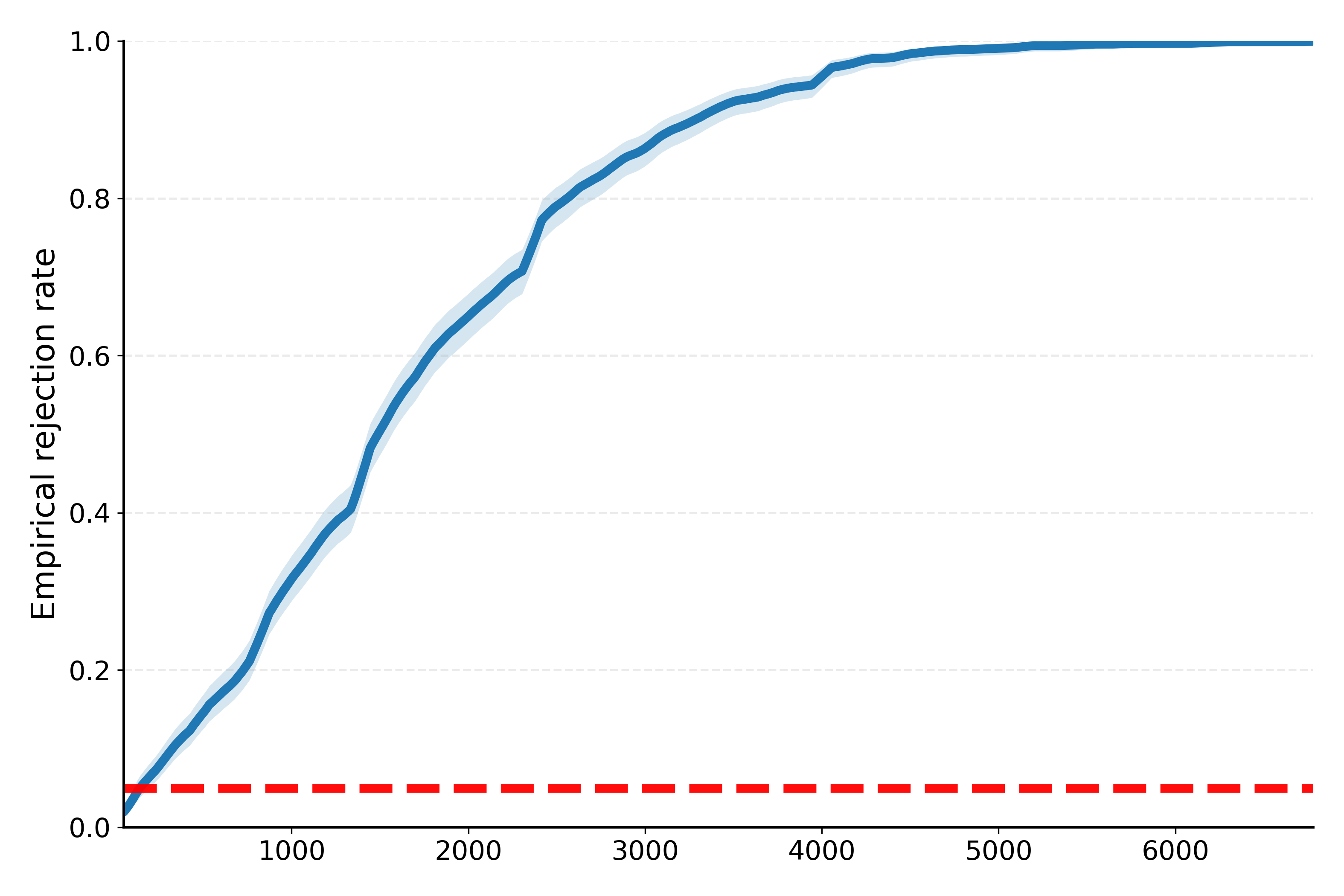}
        \subcaption*{
            $\mu = 0.8$ \\
            Blackbox
        }
    \end{subfigure} &
    \begin{subfigure}{0.315\textwidth}
        \centering
        \includegraphics[width=\linewidth]{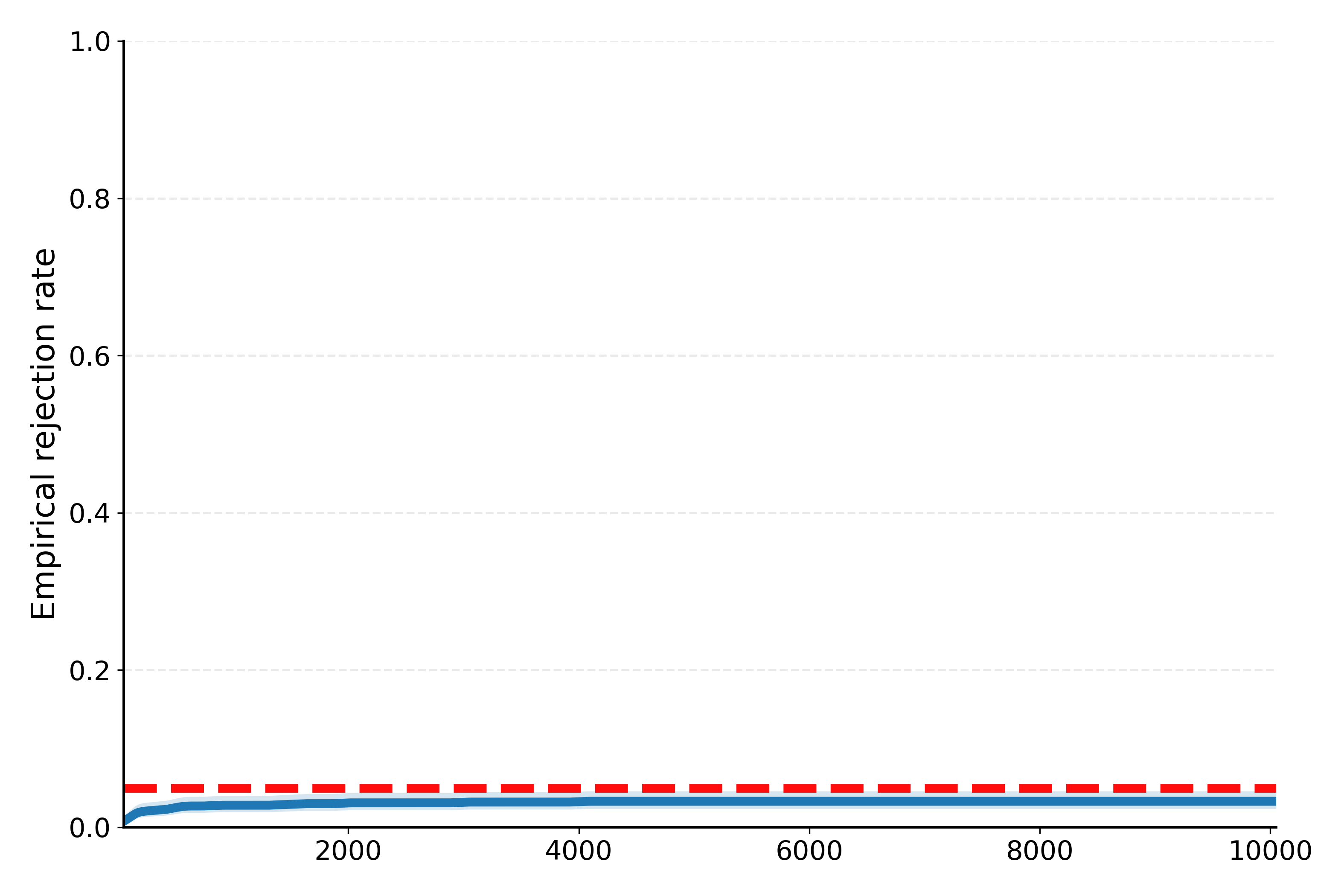}
        \subcaption*{
            $\mu = 1$ \\
            Blackbox
        }
    \end{subfigure} \\
    \hline
    \rotatebox{90}{\hspace{16mm}\textbf{DP-SGD}} &
    \begin{subfigure}{0.315\textwidth}
        \centering
        \includegraphics[width=\linewidth]{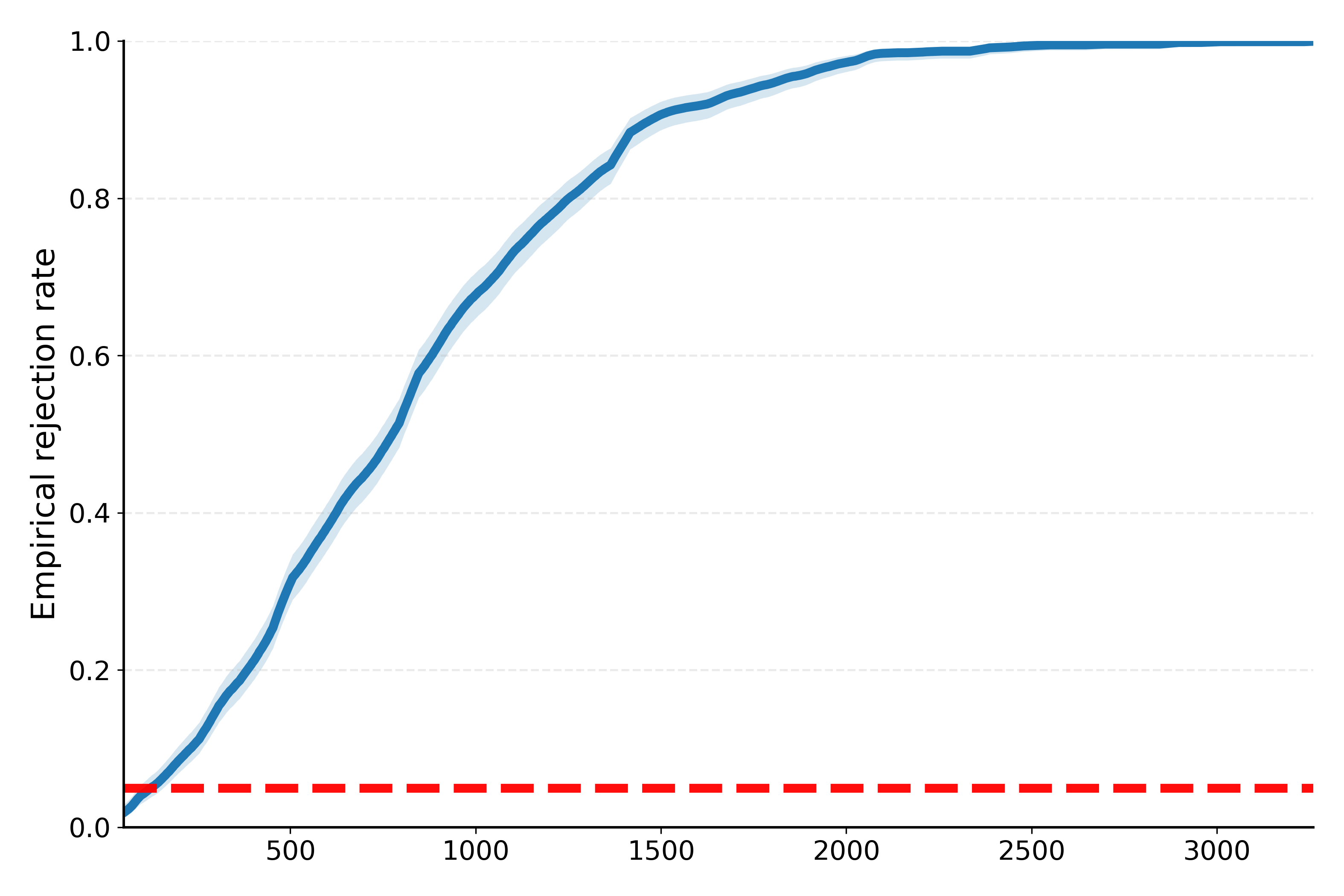}
        \subcaption*{
            $\tau = 5$ \\
            Blackbox
        }
    \end{subfigure} &
    \begin{subfigure}{0.315\textwidth}
        \centering
        \includegraphics[width=\linewidth]{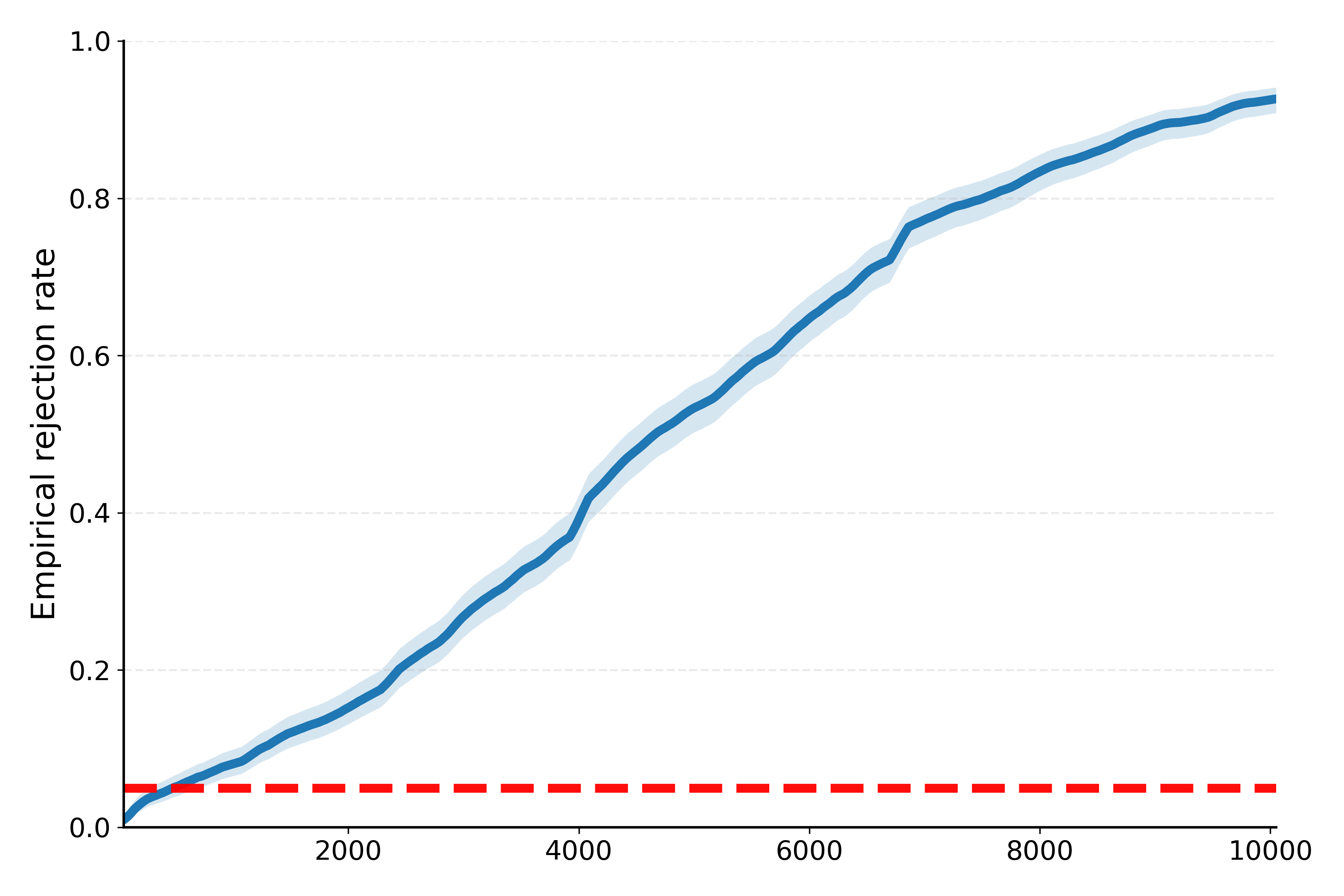}
        \subcaption*{
            $\tau = 7$ \\
            Blackbox
        }
    \end{subfigure} &
    \begin{subfigure}{0.315\textwidth}
        \centering
        \includegraphics[width=\linewidth]{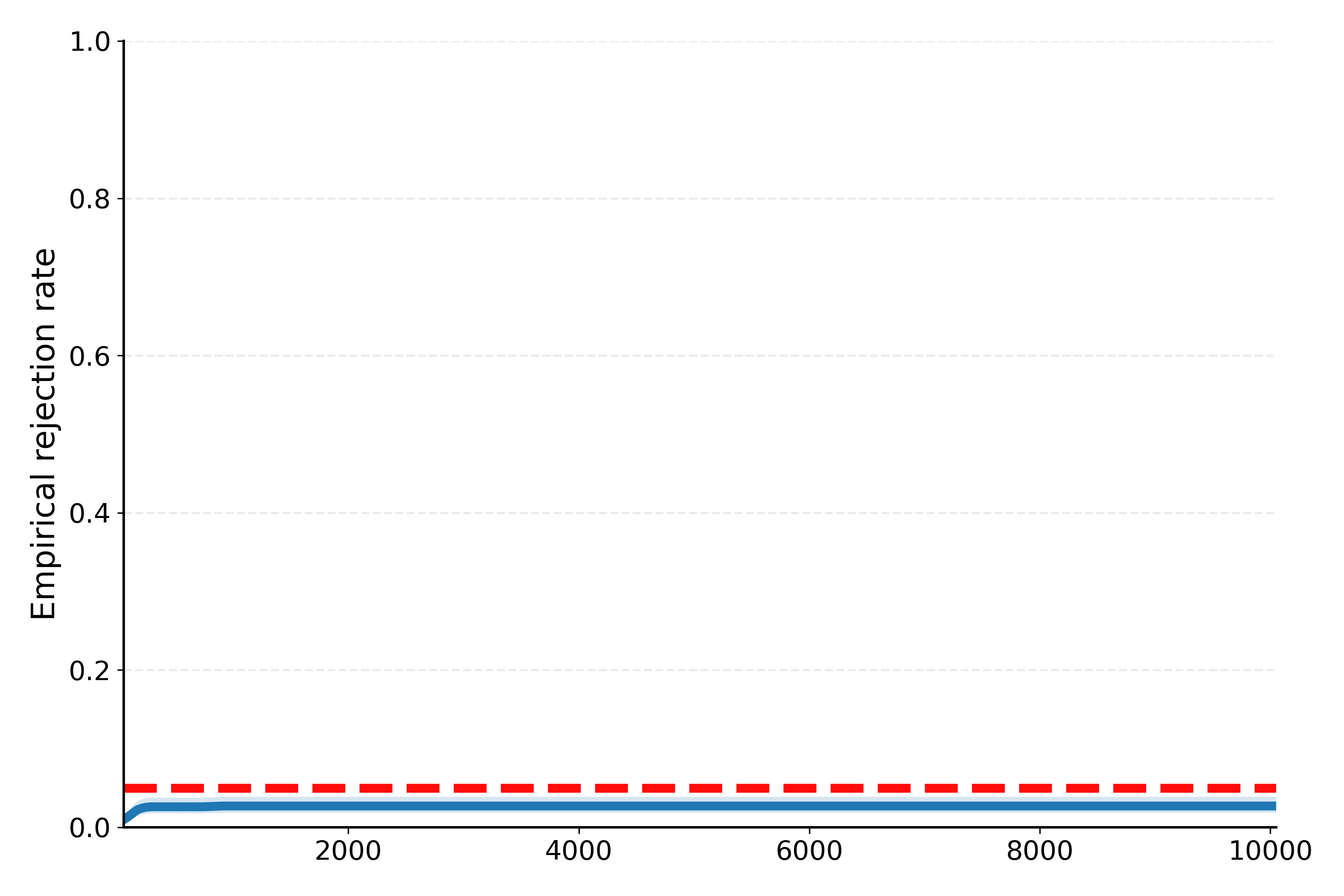}
        \subcaption*{
            $\tau = 10$ \\
            Blackbox
        }
    \end{subfigure} \\
    \hline
    \end{tabular}

    \caption{\centering Empirical rejection rates relative to sample sizes across $1,000$ runs for a fixed significance level $\gamma=0.05$.  
Rows correspond to different mechanisms and columns to the strength of privacy violation. Shaded ribbons indicate pointwise $95\%$ confidence intervals.}

\label{fig:power_curves}
   
\end{figure*}

To illustrate the impact of sequential auditing, we also compare our \emph{sequential} APT procedure to a \emph{fixed-batch} (offline) baseline and a previous auditor which also operates in \emph{fixed-batch} (see Appendix~\ref{app:add_sim} for details).
\paragraph{Interpretation of the results}

Each curve summarizes $1,000$ independent repetitions of the auditing experiment. For a given stopping time $n$ on the horizontal axis, we record whether the procedure has rejected by time $n$, and plot the resulting empirical rejection rate on the vertical axis. Thus, the blue curve can be interpreted as a power-over-time curve. Essentially it captures how quickly the auditor accumulates evidence and starts rejecting as more samples are observed.
For each fixed value of $n$, the empirical rejection rate is itself random.
If the true rejection probability is $p_n$, then the number of rejections
over $R$ repetitions follows a binomial distribution 
$\operatorname{Bin}(R,p_n)$, and the empirical rejection rate is this count
divided by $R$. To visualize the resulting uncertainty, we have included shaded
ribbons showing pointwise $95\%$ confidence intervals for the rejection
probabilities. These intervals are very narrow in our experiments.
The red dashed horizontal line indicates the significance level $\gamma=0.05$. Under the null hypothesis (third column), the blue curve should stay near this line, while under alternatives (first and second column) it should rise above it.

Across all simulation settings, the sequential auditor behaves as intended under $H_0$: the empirical rejection probability remains at (or below) the prescribed level $\gamma$. Here, we do not expect to fully use the prescribed false positive rate $\gamma$, as the procedure is supposed to hold even when our sample size becomes arbitrarily large, i.e. $n\to \infty$. Under different alternatives, the procedure detects the signal quickly, often after only a few hundred samples. Consequently, only minor effort is needed in many cases. Stronger 
privacy violations (larger effect size) impact the detection rate positively, which is a classical and expected phenomenon.

\subsection{Comparison to Prior Work}
\label{sec:comparison-prior}

An empirical comparison with existing work is only meaningful for other sequential methods. The only existing work in this direction seems to be the (independent and parallel) method by \cite{gonzalez2025sequentially} for $(\ve, \delta) $-DP. A direct apples-to-apples comparison is inherently difficult, because the two works audit different privacy objects.  We test a claimed $f$-DP tradeoff curve: the null is that the hypothesized tradeoff function $T$ for neighboring output distributions $P$ and $Q$ satisfies that 
\begin{align*}
    H_0: T(\alpha) \geq f_{\claim}(\alpha), \quad \forall \alpha \in [0,1], \tag{c.f.\eqref{e:hyp:aud}}
\end{align*}
whereas rejection means that a classifier-induced pair $(\alpha_\phi,  \beta_\phi)$ provides evidence that the claimed tradeoff curve is violated. 

By contrast, \cite{gonzalez2025sequentially} audit approximate DP.  Since the hockey-stick divergence defining \((\varepsilon,\delta)\)-DP is difficult to test directly in a black-box setting, their method tests an MMD surrogate:
\begin{align}
    \label{e:hyp:other work aud}
    H^{\mathrm{MMD}}_{0}:\quad
    \mathrm{MMD}(P,Q) \leq \tau_{\mathrm{MMD}}(\varepsilon,\delta),    
\end{align}
where the constant $\tau_{\mathrm{MMD}}(\varepsilon,\delta)$ is defined as $\sqrt{2}\left(1-\frac{2(1-\delta)}{1+e^{\varepsilon}}\right)$.
Thus $\mathrm{MMD}(P,Q)>\tau_{\mathrm{MMD}}(\varepsilon,\delta)$ indicates a violation of the claimed $(\varepsilon,\delta)$-DP guarantee. We note that, however, the converse does not  hold: a mechanism may violate approximate DP while still having small MMD.

Because the two methods differ in their testing objectives, privacy notions, and underlying statistical objects, we do not aim to establish a universal sample-efficiency dominance claim. Indeed, the sample efficiency of \cite{gonzalez2025sequentially} is governed by the MMD margin
\begin{align*}
    \Delta_{\mathrm{MMD}} := \mathrm{MMD}(P,Q)-\tau_{\mathrm{MMD}}(\varepsilon,\delta),
\end{align*}
which depends on the choice of kernel and bandwidth for the MMD test-function class, the neighboring output distributions $P,Q$, and the claimed privacy level. In contrast, the sample efficiency of our auditor is governed by the separation between the classifier-induced error pair and the claimed tradeoff curve, which depends on the construction of the classifier and the neighboring output distributions $P,Q$.

Thus, favorable regimes in terms of sample-efficiency may differ. We do, however, want to point out that for \cite{gonzalez2025sequentially}, the MMD margin typically becomes small in larger-$\varepsilon$ regimes, as the MMD threshold approaches $\tau_{\mathrm{MMD}}(\varepsilon,\delta)$. This "saturation effect" means that the  effect size becomes small, which in turn reduces the test's power. Indeed, \cite{gonzalez2025sequentially} report poor behavior beyond roughly $\varepsilon>0.6$ in their DP-SGD experiments (so fairly high privacy regimes). Our auditor does not rely on this MMD-to-DP surrogate and therefore remains effective for more general privacy parameters. Conversely, MMD-based methods may be advantageous for certain high-dimensional settings with an appropriate kernel choice, whereas a KDE-based APT instantiation is more suited to low-dimensional settings. With these comments in mind, the below experiments should be interpreted as illustrating how the two sequential auditors behave under one specific  test constructions, that was proposed  by~\cite{gonzalez2025sequentially} to evaluate their method.

\paragraph{Additive-noise benchmark.}
For a direct quantitative comparison in the small-$\varepsilon$ regimes where \cite{gonzalez2025sequentially} report strong performance, we follow their additive-noise mean-estimation benchmarks. To align the sample-efficiency comparison with the privacy object audited by~\cite{gonzalez2025sequentially}, we do not audit a tight tradeoff curve here as the preceding experiments . Instead, for the same claimed $(\varepsilon,\delta)$-DP guarantee, our auditor audit the corresponding tradeoff lower bound
\begin{align*}
    T_{\varepsilon,\delta}(\alpha) = 
    \max\left\{
        0,\,
        1-\delta-e^{\varepsilon}\alpha,\,
        e^{-\varepsilon}(1-\delta-\alpha)
    \right\}.
\end{align*}
This choice is less tailored and may worsen our sample efficiency, but it makes the comparison with the $(\varepsilon,\delta)$-DP audit of~\cite{gonzalez2025sequentially} more direct.\\
We now present the following randomized algorithms from \cite{gonzalez2025sequentially}
\begin{align*}
\mathrm{DPLaplace}(X) &:= \frac{\sum_{i=1}^n X_i}{\tilde{n}} + \rho_1, \\
\mathrm{NonDPLaplace1}(X) &:= \frac{\sum_{i=1}^n X_i}{n} + \rho_2.
\end{align*}
Here, \(\tilde{n} = \max\{10^{-12},\, n + \tau\}\) with \(\tau \sim \mathrm{Laplace}(0, 2/\varepsilon)\),
\(\rho_1 \sim \mathrm{Laplace}(0, 2/\lceil \tilde{n} \varepsilon \rceil)\), and
\(\rho_2 \sim \mathrm{Laplace}(0, 2/\lceil n \varepsilon \rceil)\). Analogously, \cite{gonzalez2025sequentially} define DPGaussian and NonDPGaussian, which should not be confused with the standard equal-variance Gaussian shift mechanism associated with $\mu$-GDP. In particular, we use the same claimed $T_{\varepsilon,\delta}$ described above but not the Gaussian/GDP curve. We notice in passing that the adjacency notion in \cite{gonzalez2025sequentially} deviates slightly from ours (exchanging records from a database vs. removing one), but this has little impact on the auditing itself. Below, we consider the remove database notion, in the choice of databases $D=(0), D'=(0,1)$ by \cite{gonzalez2025sequentially}. We summarize the results of our simulation across $1,000$ runs by reporting the average rejection time $\bar n$ and the empirical standard deviation of these runs. Results are reported in Table \ref{tab:seq_audit_two_col_grouped}. For ease of comparison, we also cite the corresponding numbers by \cite{gonzalez2025sequentially} in the same table.

\begin{table*}[htp]
\centering
\caption{Sample efficiency comparison for the additive-noise benchmark. The left block reports Algorithm~\ref{alg:seq-audit} over $1{,}000$ repetitions, using the tradeoff lower bound induced by $(\varepsilon,10^{-5})$-DP for Gaussian and by pure $\varepsilon$-DP for Laplace. The right block contains the paper-reported values from Table~1 of \cite{gonzalez2025sequentially}.}
\label{tab:seq_audit_two_col_grouped}
\small
\setlength{\tabcolsep}{5pt}
\renewcommand{\arraystretch}{1.15}
\begin{tabular*}{\textwidth}{@{\extracolsep{\fill}} l cc cc cc cc @{}}
\toprule
& \multicolumn{4}{c}{Algorithm~\ref{alg:seq-audit} \textnormal{(ours)}}
& \multicolumn{4}{c}{Paper-reported values from \cite{gonzalez2025sequentially}} \\
\cmidrule(lr){2-5}\cmidrule(lr){6-9}
Mechanism
& \multicolumn{2}{c}{\(\varepsilon = 0.01\)}
& \multicolumn{2}{c}{\(\varepsilon = 0.1\)}
& \multicolumn{2}{c}{\(\varepsilon = 0.01\)}
& \multicolumn{2}{c}{\(\varepsilon = 0.1\)} \\
\cmidrule(lr){2-3}\cmidrule(lr){4-5}\cmidrule(lr){6-7}\cmidrule(lr){8-9}
& Rej. rate & \(\bar n\) to rej.
& Rej. rate & \(\bar n\) to rej.
& Rej. rate & \(\bar n\) to rej.
& Rej. rate & \(\bar n\) to rej. \\
\midrule
DPGaussian
& 0 & \textemdash
& 0 & \textemdash
& 0 & \textemdash
& 0 & \textemdash \\
NonDPGaussian1
& 1 & \(60.86 \pm 19.33\)
& 1 & \(71.60 \pm 29.79\)
& 1 & \(264 \pm 9.3\)
& 1 & \(562 \pm 29.2\) \\
DPLaplace
& 0 & \textemdash
& 0.021 & \textemdash
& 0 & \textemdash
& 0 & \textemdash \\
NonDPLaplace1
& 1 & \(81.38 \pm 39.43\)
& 1 & \(115.34 \pm 70.92\)
& 1 & \(331 \pm 14.5\)
& 1 & \(920 \pm 61.6\) \\
\bottomrule
\end{tabular*}
\end{table*}

\paragraph{Interpretation of Table~\ref{tab:seq_audit_two_col_grouped}.}
In these small-\(\varepsilon\) regimes, both procedures detect the non-private Gaussian and Laplace benchmarks with rejection rate one. The stopping times differ substantially:  for \(\varepsilon=0.01\), our average stopping times are approximately $23\%$ and $25\%$ of the paper-reported values in \cite{gonzalez2025sequentially} for \(\mathrm{NonDPGaussian1}\) and \(\mathrm{NonDPLaplace1}\), respectively.  For $\varepsilon=0.1$, the corresponding ratios are approximately $13\%$ and $13\%$. This supports the empirical claim that, on these additive-noise benchmarks and in the regimes reported in \cite{gonzalez2025sequentially}, our $f$-DP auditor can reject with fewer samples.\\
The DP-compliant controls also behave as expected. The (theoretical) significance level for both procedures is $\gamma= 0.05$. In finite samples, the observed rejection rates fluctuate due to randomness across simulations and for our $1,000$ runs, the binomial sampling variability suggests that values below $0.064$ should be considered acceptable. Rejection rates by both methods are always well below $0.064$. For~\cite{gonzalez2025sequentially}, the reported rejection rates are exactly zero under $H_0$ in all scenarios. This is consistent with a conservative test, which can be desirable for false-positive control but is also associated with lower power or slower detection under alternatives.


\begin{figure}[htp]
\centering
\setlength{\tabcolsep}{2pt}
\begin{tabular}{cc}
\includegraphics[width=0.51\columnwidth]{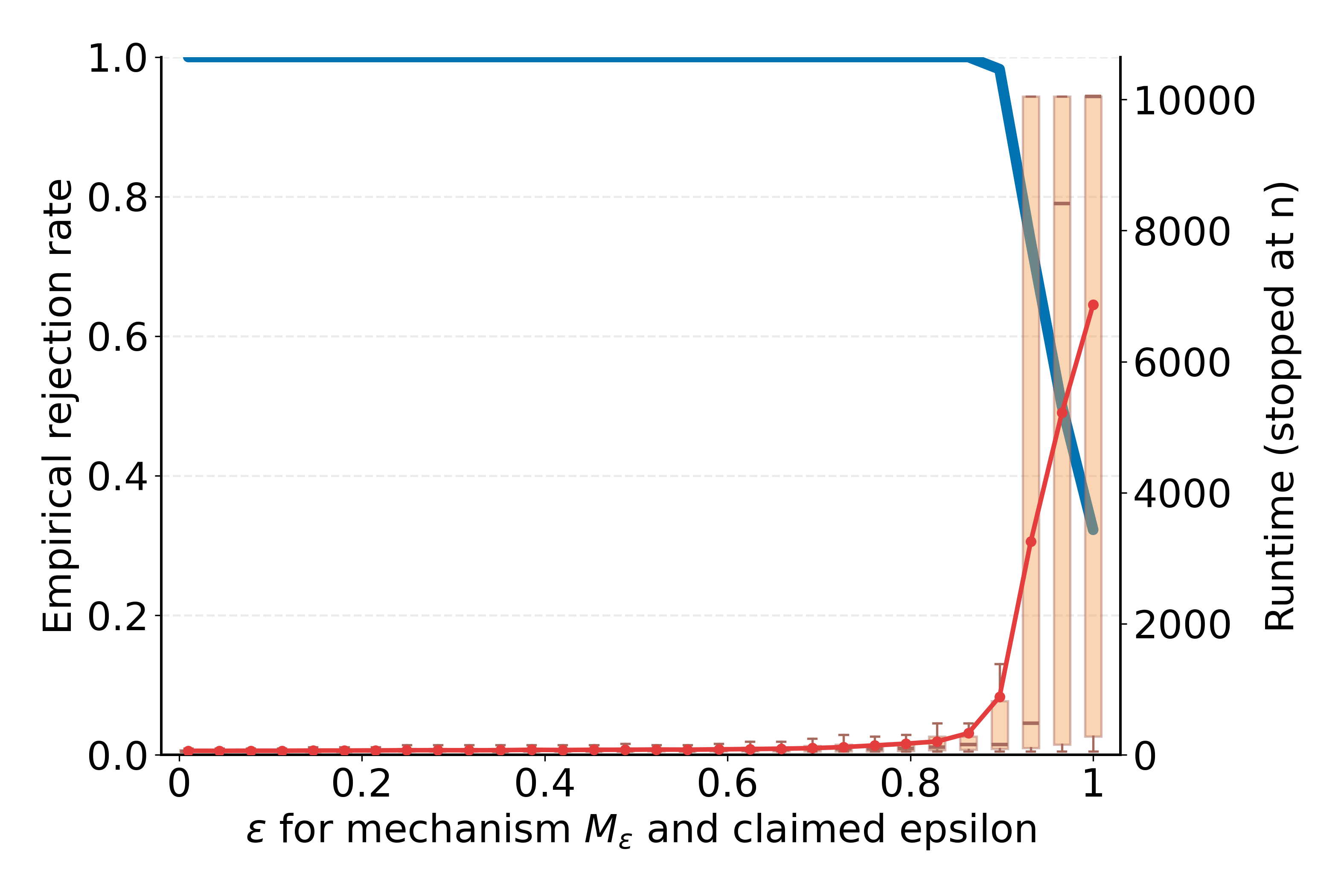} &
\includegraphics[width=0.51\columnwidth]{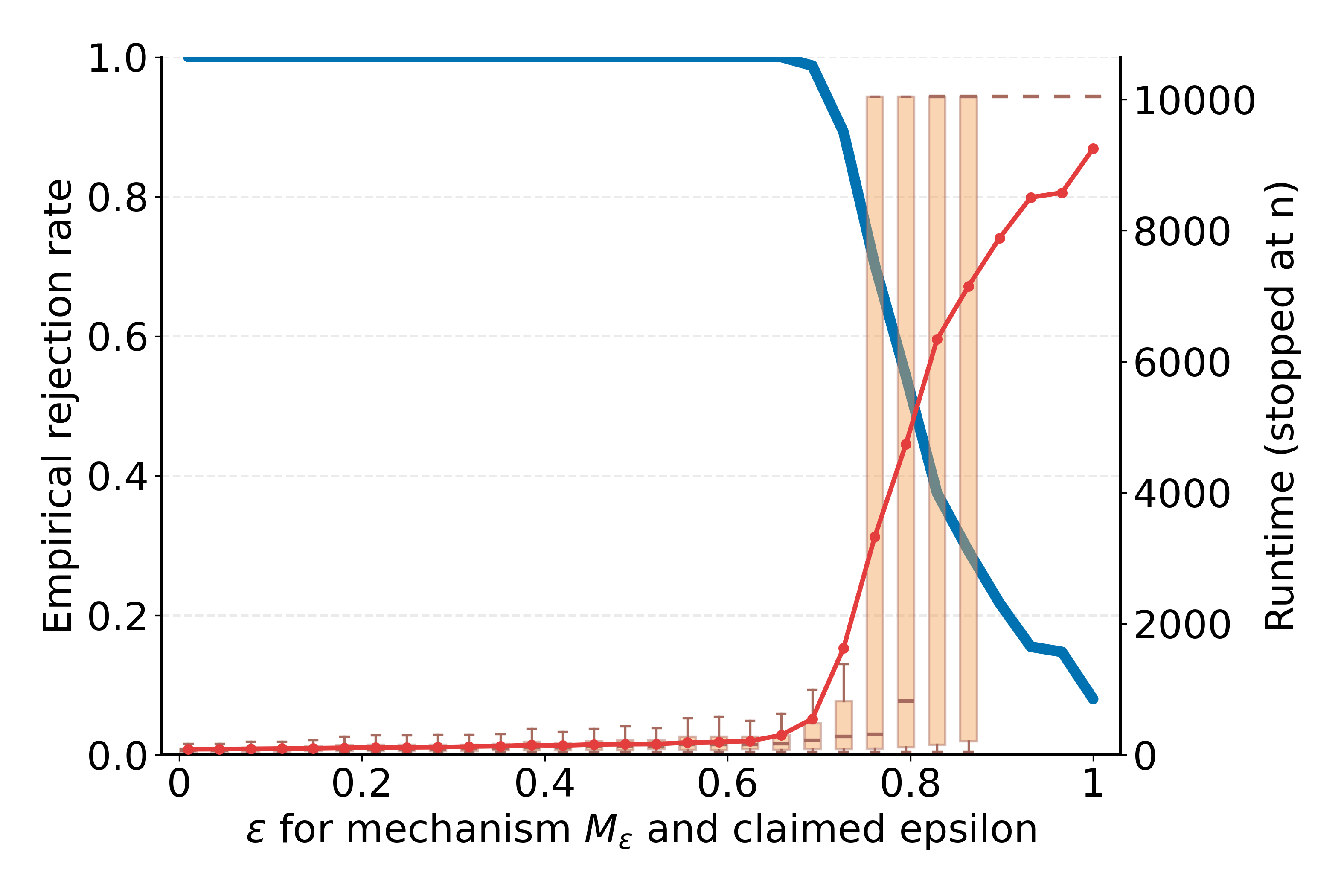} \\
\end{tabular}
\caption{Larger-\(\varepsilon\) behavior of Algorithm~\ref{alg:seq-audit} on \(\mathrm{NonDPGaussian1}\) (left) and \(\mathrm{NonDPLaplace1}\) (right). Blue curves show empirical rejection rates; red curves and boxplots show stopping times, capped at \(10{,}000\) samples.}
\label{fig:large-epsilon-ours}
\end{figure}

\begin{figure}[htp]
\centering
\setlength{\tabcolsep}{2pt}
\begin{tabular}{cc}
\includegraphics[width=0.49\columnwidth]{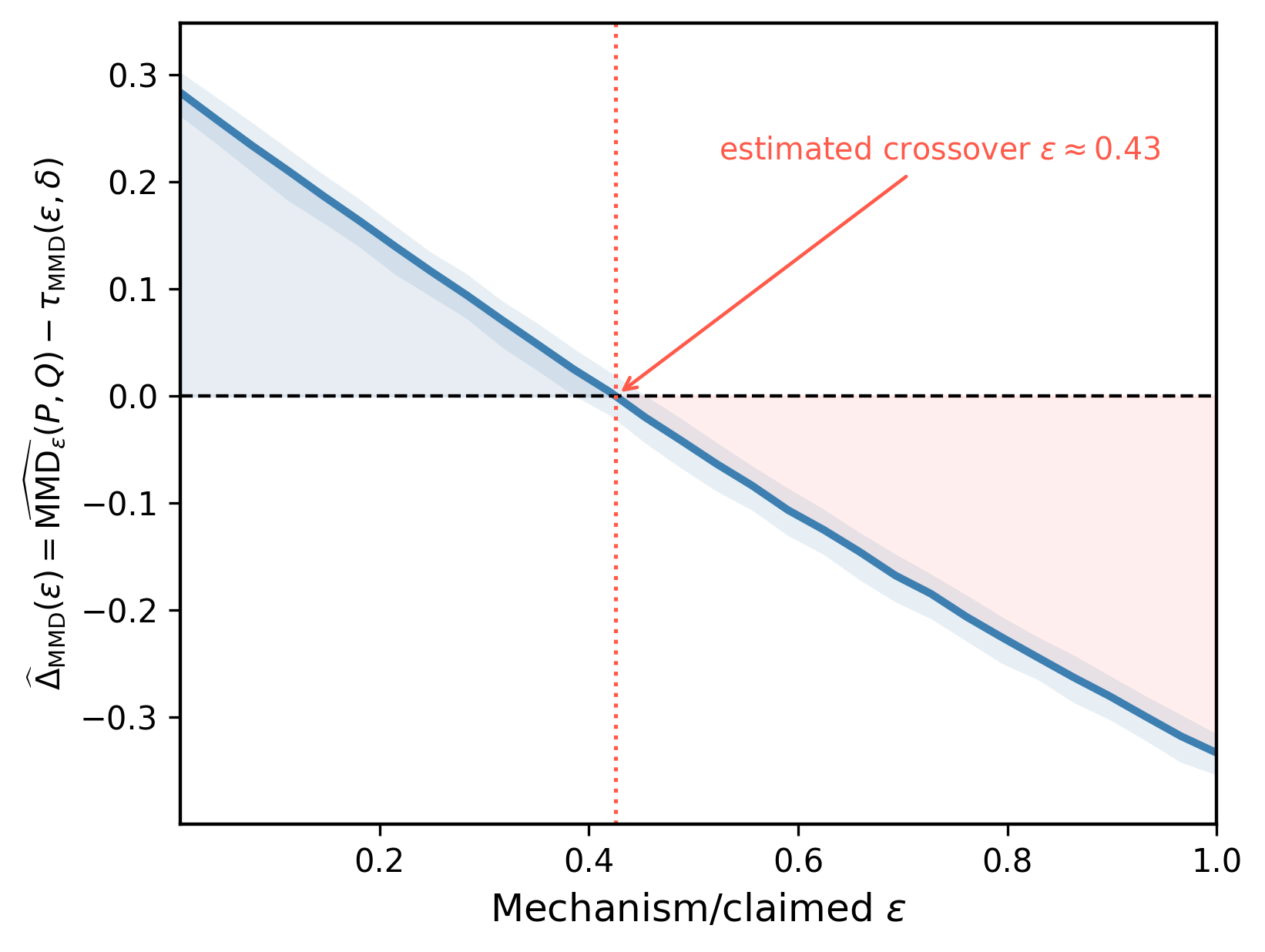} &
\includegraphics[width=0.49\columnwidth]{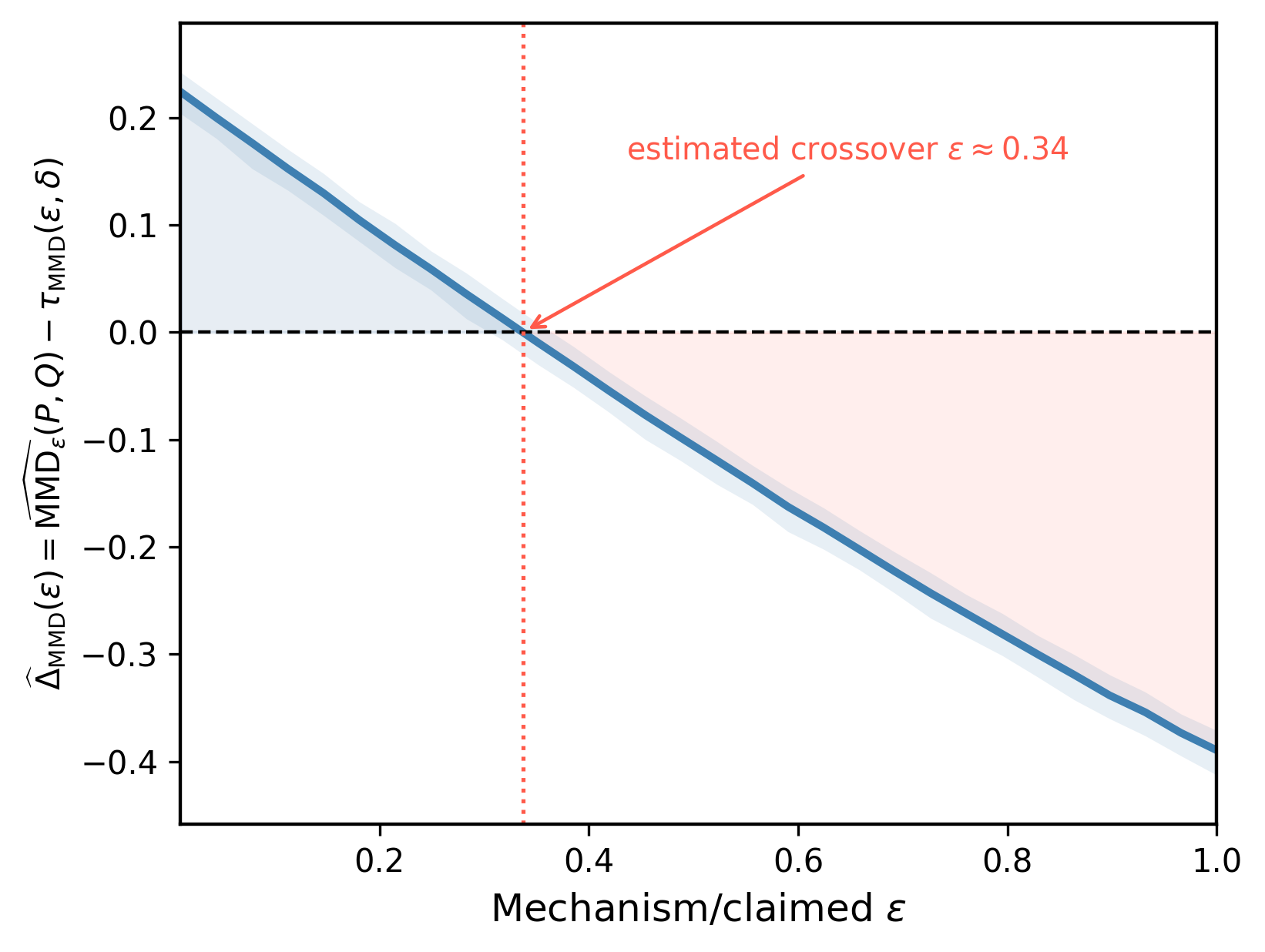} \\
\end{tabular}
\caption{Estimated MMD margins for \(\mathrm{NonDPGaussian1}\) (left) and \(\mathrm{NonDPLaplace1}\) (right); shaded ribbons denote 95\% confidence intervals. Once the margin crosses zero, around $\varepsilon \approx 0.43$ and $\varepsilon \approx 0.34$, respectively, the MMD-based approach can no longer reliably detect the violations.}
\label{fig:mmd-margin-epsilon}
\end{figure}

\paragraph{Larger-$\varepsilon$ regimes and the MMD saturation effect.}
We also sweep over larger claimed privacy levels for the two non-private mechanisms,
$\mathrm{NonDPGaussian1}$ and $\mathrm{NonDPLaplace1}$, to explain the difference between the two auditors. Figure~\ref{fig:large-epsilon-ours} shows that our auditor continues to detect violations over a wider range of $\varepsilon$ than the range in which the MMD surrogate has a positive margin. The blue curve reports the empirical rejection rate, while the red curve and boxplots summarize the stopping time of the auditors, with monitoring capped at $10{,}000$ samples. As $\varepsilon$ increases, the effective violation margin becomes smaller. Consequently, although these mechanisms remain buggy, their violations of the audited lower-bound tradeoff become harder to detect under our auditor instantiation.

Figure~\ref{fig:mmd-margin-epsilon} reports the corresponding estimated MMD margins
$\widehat{\Delta}_{\mathrm{MMD}}(\varepsilon) =\widehat{\mathrm{MMD}}(P,Q)-\tau_{\mathrm{MMD}}(\varepsilon,\delta)$ for the same benchmark, with $\widehat{\mathrm{MMD}}(P,Q)$ computed using the estimation procedure of~\cite{gonzalez2025sequentially}. $\widehat{\Delta}_{\mathrm{MMD}}(\varepsilon)$ crosses zero at approximately $\varepsilon\approx 0.43$ for $\mathrm{NonDPGaussian1}$ and $\varepsilon\approx 0.34$ for $\mathrm{NonDPLaplace1}$.  Once the MMD margins are non-positive, collecting more samples cannot make the MMD surrogate reject reliably even if the mechanism still violates $(\varepsilon,\delta)$-DP.

\subsection{Real-World DP-SGD / Auditing in one run}
\label{sec:dp-sgd}
In this section, we demonstrate that our sequential pipeline fits together with the recent emerging one-run DP-SGD auditing paradigm~\cite{steinke2023privacy},  in which a single DP-SGD training~\cite{AbadiCGMMT016, DeBerradaHayesSmithBalle22} run simultaneously evaluates many (Dirac) canaries and provides multiple sequences of per-training-step observations for auditing. Concretely, we instantiate our method in a realistic DP-SGD setting using the now-standard whitebox canary methodology~\cite{nasr2023tight} from recent DP-SGD auditing work: the auditor injects Dirac canaries and records step-wise intermediate statistics from the noised and clipped gradient during the private training.

Following~\cite{nasr2023tight}, at each DP-SGD step $t$, we extract a scalar observation for a fixed canary by taking the dot product between the (clipped and noised) gradient and a fixed canary direction. We then aggregate the per-step observations of a canary into a single scalar by averaging over training steps. We denote these averaged scores by $\{X_i\}_{i=1}^n$ for the canary-included world and by $\{Y_i\}_{i=1}^n$ for the no-canary world, where $n$ is the number of (independent) canaries used per side. The Dirac canary direction is chosen to create a detectable gradient-inner-product signal for the resulting neighboring pair, which does not necessarily yield a worst-case neighboring pair.

To instantiate our sequential pipeline in Algorithm~\ref{alg:seq-audit}, it suffices to define a classifier using \textproc{BuildClassifier} from samples $\{X_i\}$ and $\{Y_i\}$. We adopt the simple threshold rules presented in Section \ref{sec:4.2}, where $\{\phi_\eta\}_{\eta\in\mathbb{R}}$ is given by
\begin{align*}
    \phi_\eta(z) = \mathbf{1}\{z \ge \eta\}.
\end{align*}

We have trained a DP model on the CIFAR-10 dataset following the architecture and hyperparameters of~\cite{DeBerradaHayesSmithBalle22}: WideResNet-16-4 with GroupNorm/Weight-Standardization, batch size $4096$, clipping norm $C=1.0$, noise multiplier $\sigma=3.0$. We additionally insert $10{,}000$ canaries during training, following the canary setup of~\cite{nasr2023tight}~\footnote{Prior DP-SGD auditing work often uses $\approx 5{,}000$ canaries. We observe comparable behavior at $5{,}000$; we use up to $10{,}000$ here to more clearly illustrate the gains from sequential auditing over a wider range of claimed $\mu$ and number of canaries.}. As in the one-run auditing methodology, the Dirac canaries are placed on distinct coordinates, which reduces cross-canary interference and makes the resulting scores closer to the independent-sample setting considered in our theory. We nevertheless do not claim robustness to arbitrary dependence among canary scores.


\paragraph{Experiment settings.}
To illustrate the impact of sequential auditing, we run the following experiment. Denote by $\widehat P$ and $\widehat Q$ the empirical distributions of the canary-included and no-canary scores produced by the single training run. We then perform $250$ independent resampling repetitions, and run our sequential auditor on each resampled subsets, respectively. We use these repetitions to study the auditor's variability without retraining DP-SGD, enabling the boxplots in Figure~\ref{fig:real_world}.
 
For the parameters, we choose the default settings as throughout the whole section. In Figure \ref{fig:real_world} we illustrate the empirical rejection rates (blue) and average runtime until rejection, expressed by the sample size $n$ (red). The x-axis represents the claimed $\mu$ (smaller $\mu$ means a larger DP violation). Additionally, we include boxplots of the runtimes over the 250 runs, to illustrate variation in runtimes. 
In this setup, the precise DP level is unknown. Yet a theoretical upper bound, in terms of $(\varepsilon, \delta)$-DP, can be given with $\ve^*=6.56, \delta^*=10^{-5}$.\footnote{This $\ve^*=6.56$ corresponding to $\delta^*=10^{-5}$ is computed using the privacy accountant provided by \texttt{Opacus}~\cite{opacus} for DP-SGD, with the noise multiplier, sampling rate, and number of training steps derived from the prescribed training parameters.}  Using the relation between $\mu-$GDP and $(\ve,\delta)-$DP (see \cite{dong:roth:su:2022}), this bound translates into a parameter value of exactly $\mu=1.4$; this means we expect that for $\mu=1.4$, no DP violation exists. By the same calculation, we notice that $\mu=1.2$ corresponds to $\ve=5.413$ and $\delta=10^{-5}$, and $\mu=1.1$ corresponds to $\ve=4.88$ and $\delta=10^{-5}$.

\paragraph{Interpretation of the results} First, consider the blue curve of rejection probabilities in Figure \ref{fig:real_world}. At $\mu=1.4$ (where we believe no DP violation exists), we do not reject in any of our 250 runs, and in particular hold the significance level. Afterwards, the number of rejections rises rapidly to about $35\%$ for $\mu=1.2$, reaching $>99.5\%$ for $\mu=1.1$. For smaller values of $\mu$, i.e. larger DP violations, the rejection rate is always $100\%$, but larger violations are associated with shorter runtimes. To see this, consider the red curve of average runtimes until rejection. For $\mu=1.4$, runtimes are maximal, because no rejection takes place (this is good because any rejection would be a false rejection in this case). Afterwards, average runtimes drop fast, to less than $3,000$ for $\mu=1.1$ and going all the way to a few hundred for $\mu=0.8$. \emph{This result is crucial: it demonstrates that our adaptive DP auditor drastically reduces the sample size $n$, if small samples suffice.}

An even more nuanced picture emerges when considering the runtime boxplots, which express the variability in runtimes. Variability decreases as the DP violation grows larger, with all runs very short for $\mu \le 1$. At $\mu=1.1$ variability is higher, but we can see that about $50\%$ of runs finish for $n \le 2,000$, and $25\%$ even for $n \le 1,000$. Thus, even for smaller DP violations, there exists a large proportion of reasonably short audits.  Boxplots for the cases $\mu \ge 1.2$ are less meaningful, because the majority of runs end without rejection.

\begin{figure}
    \centering
    \includegraphics[width=0.75\linewidth]{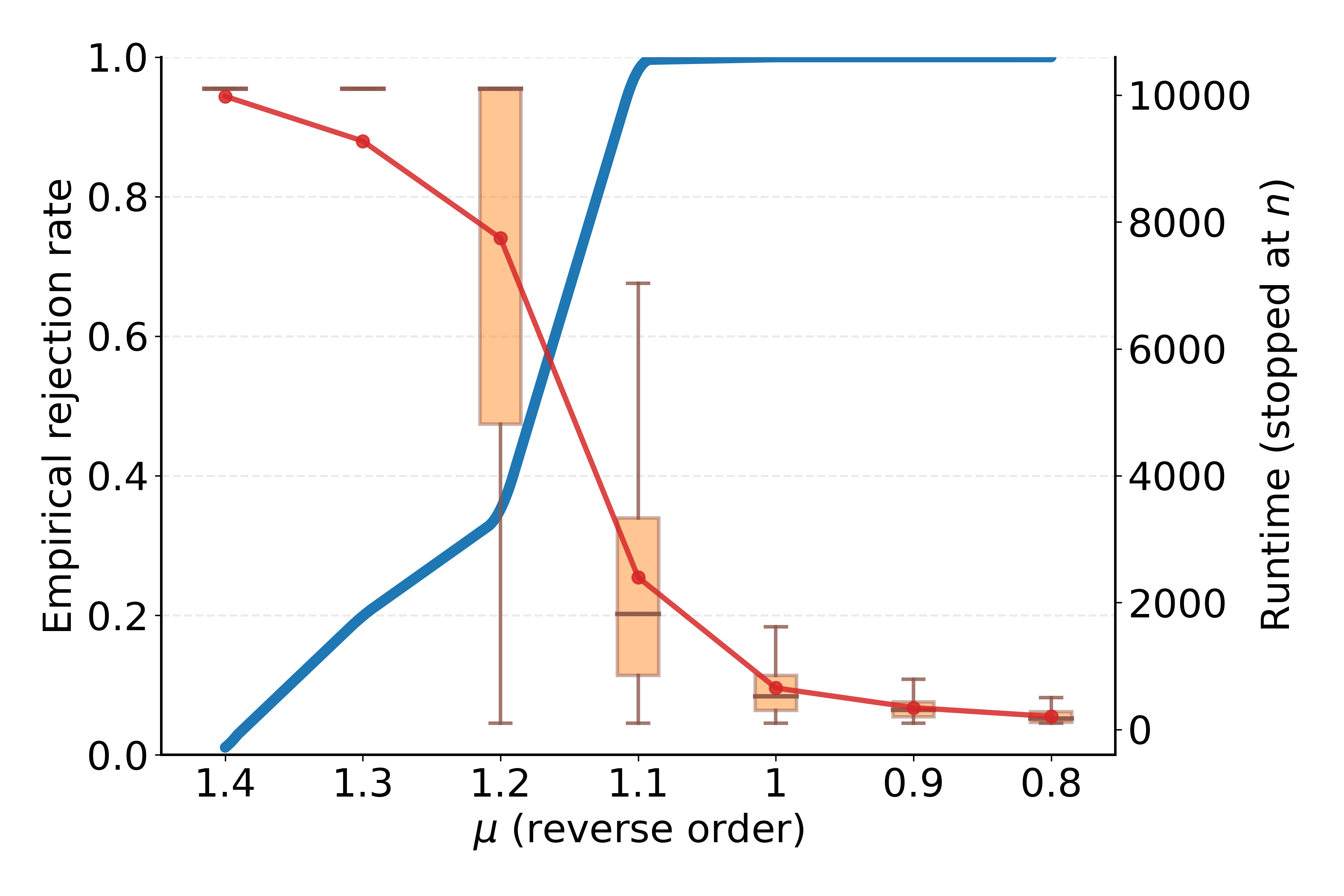}
    \caption{Empirical rejection rate (blue curve) and average runtime (red curve) for different claimed $\mu$. Variation in runtimes is represented by orange boxplots. Notice that smaller values for $\mu$ correspond to larger privacy violation.}
    \label{fig:real_world}
\end{figure}


%% file: 05-ethical-considerations.tex
\section*{Ethical Considerations}
\paragraph{Stakeholder analysis.}
The primary stakeholders are data subjects, whose privacy the audited mechanisms are designed to protect, model trainers seeking to validate their privacy guarantees. The intended benefit of this work is to reduce the cost of detecting implementation bugs, or invalid audit practices that could otherwise lead to excess privacy leakage. 
\vspace{-0.4cm}

\noindent \paragraph{Risks and mitigations.}
Likeur other privacy auditing tools, o auditor could in principle be misused to search for weak or misconfigured deployments. This risk is mitigated by the fact that the auditor requires repeated blackbox or whitebox access to the mechanism and reports only statistical evidence of a privacy violation, rather than reconstructing private records or revealing the underlying dataset. 

\vspace{-0.4cm}
\noindent \paragraph{Decision to conduct and publish the research.}
We conducted this work because sample-efficient and statistically valid privacy auditing can help practitioners identify and fix privacy failures before deployment. We publish the method and artifacts to support reproducibility and help the community avoid invalid early-stopping practices in DP auditing. Our intended use case is authorized auditing by mechanism designers, system owners, or compliance teams.


%% file: 06-open-science.tex
\section*{Open Science}

To comply with the Open Science Policy, we make our research artifacts available to the public~\footnote{\url{https://github.com/martindunsche/sequential_fdp_auditing}; permanent archive: \url{https://doi.org/10.5281/zenodo.20398258}}. The artifact includes the code implementing our auditing pipeline, together with scripts and instructions to reproduce the experimental results reported in the paper. In our repository, we ensure the transparency and reproducibility by describing the dataset used, the specification of our machines, and citing all algorithms tested.


%% file: 100-proofs-and-technical-details.tex
\section{Proofs and Technical Details}

The main probabilistic result, used in our theory are the KMT \cite{kmt:1975, kmt:1976} approximations. A modern formulation and extension of these results is given e.g. in \cite{berkes:liu:wu:2014}. The last reference matters in connection with our result to settings with dependent data, which we discuss in~Appendix \ref{sec:adaptive_monitoring}.
For the next result, we restrict our attention to i.i.d. data. As a first step we cite the KMT approximations, stated in the above papers.

\begin{theo} \label{thm:kmt}(KMT approximation) Suppose $(X_i)_i$ are i.i.d. random variables with $|X_i|\le C<\infty$ a.s., $\mathbb{E}X_i=0$ and $\mathbb{E}X_i^2=: \sigma^2$. Then, possibly after enlarging the probability space, there exist i.i.d. Gaussians $(Z_i)_i$ with $Z_i \sim \mathcal{N}(0, \sigma^2)$ such that
{\small
\begin{align*}
    \bigg|\sum_{i=1}^k \{X_i -Z_i\} \bigg|=\mathcal{O}(\log(k)), \qquad a.s.
\end{align*}
}
\end{theo}
We use the KMT approximation to conduct, roughly speaking, an infinite number of hypothesis tests in $k$.

\subsection{Proof of Theorem \ref{theo:main}} \label{sec:proof:main}

The main step in the proof will be the Gaussian approximation of the difference of empirical errors of $\phi$ minus the true error. We demonstrate the main step for the error $\alpha_\phi$, noting that the proof for $\beta_\phi$ works in exact analogy. For any $k \ge M$, the estimator of  $\alpha_\phi$ is the empirical version $\hat \alpha_\phi(k):=\frac{1}{k} \sum_{i=1}^{k} \phi(X_i)~.$
We now consider the rescaled difference
{\small
\begin{align*}
D(k):=&\dfrac{\sqrt{k}}{\sqrt{\log\!\left(20 + k/M\right)}} \big[\hat \alpha_\phi(k) - \alpha_\phi\big] \\
=& \dfrac{1}{\sqrt{k\log\!\left(20 + k/M\right)}}\sum_{i=1}^{k} \big[\phi(X_i)-\mathbb{E}\phi(X_i)\big]. 
\end{align*}
}

Since $\phi(x) \in [0,1]$, the random variables $\phi(X_i)-\mathbb{E}\phi(X_i)$ are independent, identically distributed and bounded. In particular, using Theorem \ref{thm:kmt}, we can find normally distributed random variables $Z_i \sim \mathcal{N}(0,\mathbb{V}ar(\phi(X_1)))$ for $1 \le i <\infty$ such that almost surely 
\[
\bigg|\sum_{i=1}^{k}\Big\{ \big[\phi(X_i)-\mathbb{E}\phi(X_i)\big]-Z_i\Big\}\bigg| = \mathcal{O}(\log(k))
\]
Defining $D_Z(k):= \dfrac{1}{\sqrt{k\log\!\left(20 + k/M\right)}}\sum_{i=1}^{k} Z_i,$
it follows that, for large $M$
\begin{align*}
& \sup_{k \ge M}|D_Z(k)-D(k)| = \mathcal{O}\bigg(\sup_{k \ge M}\frac{\log(k)}{\sqrt{k \log(20+k/M)} }\bigg) \\
=& \mathcal{O}\bigg(\sup_{k \ge M}\frac{\log(k)/k^{1/4}}{M^{1/4}} \bigg)=\mathcal{O}(M^{-1/4}) =o(1)
\end{align*}
almost surely. On an intuitive level, this approximation completes the proof. We first explain why and then provide some remaining mathematical details. We have shown that $D(k)$ may be replaced for all $k \ge M$ by $D_Z(k)$. Now, taking the upper $(1-\gamma/2)$ quantile of $\sup_{k \ge M} D_Z(k)$, say $q_{1-\gamma/2}$, implies that 
\[
D(k) \le q_{1-\gamma/2}
\]
for all $k \ge M$ simultaneously with probability converging to $1-\gamma/2$ as $M \to \infty$. The quantile is calculated by Algorithm \ref{alg:conf-adj}. In practice there exists a minor gap, since $\mathbb{V}ar(\phi(X_1))$ is not known, but it can be estimated by the empirical variance over the burn-in period and Slutsky's theorem ensures consistency. Using an analogous approximation for $\beta_\phi$, we obtain a sequence of one-sided confidence rectangles $(\alpha_\phi, \beta_\phi)$ described by Algorithm \ref{alg:conf-adj}. Coverage is, by the Bonferroni bound, $1-\gamma$ for $M \to \infty$, and the confidence bounds converge to $(\alpha_\phi, \beta_\phi)$ as $k \to \infty$. Accordingly, the guarantees from Theorem \ref{theo:main} follow directly. \\
One mathematical subtlety remains: For our above arguments, we have used the real-valued quantiles of the object $\sup_{k \ge M}D_Z(k)$. Yet we have, strictly speaking, not established that these are finite (or uniformly bounded in $M$). To ensure that, we show weak convergence of $\sup_{k \ge M}D_Z(k)$, which in particular implies finite quantiles. A bit of mathematical background is therefore needed on the so-called Brownian motion, a key process in probability theory. The Brownian motion is here denoted by $W(t), 0 \le t < \infty$. It is a random function and we list some of its properties, useful for our below derivation:
\begin{enumerate}
  \item[B1)] $W(0)=0$ a.s. and $t\mapsto W(t)$ is continuous a.s. 
  \item[B2)] $W$ has independent, centered and Gaussian distributed increments.
  \item[B3)] $\operatorname{Cov}(W(s),W(t))=\min\{s,t\} \cdot \mathbb{V}ar(W(1))$.
  \item[B4)]  $W(t), t \ge 0$ has the same distribution as  $W(ct)/\sqrt{c}, t \ge 0$ for any  $c>0$.
  \item[B5)]
  \[
    \limsup_{t\to\infty} \frac{|W(t)|}{\sqrt{2t\log\log t}} = \sqrt{\mathbb{V}ar(W(1))}, \quad a.s.
  \]
\end{enumerate} 
Properties B1)-B3) are commonly used to characterize the Brownian motion. B4) and B5) are important consequences. B4) is often referred to as the scaling property, and intuitively states that the Brownian motion looks the same on large or small scales. B5) is called "the law of the iterated logarithm" and ensures that $W(t)$ is (eventually) contained in a domain only growing slightly faster than $\sqrt{t}$.\\
Returning to our proof,
we notice that the following equality holds in distribution
\[
\sup_{k \ge M} D_Z(k) \overset{d}{=}\sup_{k \ge M} \dfrac{W(k)}{\sqrt{k\log\!\left(20 + k/M\right)}}.
\]
Here, $W$ is a centered Brownian motion, with variance $\sigma^2:=\mathbb{V}ar(\phi(X_1))$. The reason is that (using B1)-B3)) one can easily show that $\sum_{i=1}^{k} Z_i, k \ge 1$ have jointly the same distribution as $W(k), k \ge 1$. 
Next, 
using the rescaling property B4) with $c=1/M$ it follows that in distribution
\[
\sup_{k \ge M} \dfrac{W(k)}{\sqrt{k\log\!\left(20 + k/M\right)}} \overset{d}{=}\sup_{x \ge 1, x=k/M} \dfrac{W(x)}{\sqrt{x\log\!\left(20 + x\right)}}.
\]
The random variable on the right is almost surely finite, due to the law of iterated logarithm for the Brownian motion in B5).
Now, let us fix the random seed and consider $W(t), t \ge 0$ just as a fixed function. As  $M\to \infty$, the right side converges (monotonically from below) to $L:= \sup_{x \ge 1} \dfrac{W(x) }{\sqrt{x\log\!\left(20 + x\right)}}.$
To obtain this convergence, we have used that the function $W(x)/(\sqrt{x\log(20 + x})$ is continuous according to B1) and goes to $0$ as $x \to \infty$ which is again following from an application of the law of the iterated logarithm B5). 
These considerations imply the weak convergence $\sup_{k \ge M} D_Z(k) \overset{d}{\to} L,$
finishing our proof.

%% file: 101-additional-algorithms.tex
\section{Additional Algorithms}

Our APT procedure makes use of a variety of subprocedures, e.g. to calculate confidence boundaries or to build classifiers in various scenarios. We give a list of blueprints below.

\begin{algorithm}[htp]
\small
\begin{algorithmic}[1]
\Require Samples $\{X_i\}_{i=1}^{m}$ and $\{Y_i\}_{i=1}^{m}$; candidate threshold set $\mathcal T$; claimed curve $f_{\claim}$; learning rule $\mathcal L$.
\Ensure A classifier $\phi(\cdot)$ and its selected threshold $\eta^*$.

\Function{BuildClassifier}{$
\{X_i\}_{i=1}^{m},\allowbreak
\{Y_i\}_{i=1}^{m},\allowbreak
\mathcal T,\allowbreak
f_{\claim},\allowbreak
\mathcal L
$}

  \State Train a scoring function $s(\cdot) \gets \mathcal L(\{X_i\}_{i=1}^{m}, \{Y_i\}_{i=1}^{m})$.
  \For{each $\eta \in \mathcal T$}
    \State Define $\phi_\eta(y) \gets \mathbf{1}\{ s(y) \ge \eta\}$.
    \State Compute $\hat\alpha(\eta) \gets \frac{1}{m}\sum_{i=1}^{m}\phi_\eta(X_i)~.$
    \State Compute $\hat\beta(\eta) \gets 1-\frac{1}{m}\sum_{i=1}^{m}\phi_\eta(Y_i)~.$
  \EndFor
  \State Select $\eta^*$ using equation \eqref{eq:new_eta_max} applied to the curve $\{(\hat\alpha(\eta),\beta(\eta)) : \eta \in \mathcal T\}$.
  \State Define $\phi(y) \gets \mathbf{1}\{ s(y) \ge \eta^* \}$.
  \State \Return $\phi(\cdot)$.
\EndFunction
\end{algorithmic}
\caption{\textproc{BuildClassifier}: Generic classifier construction for APT via a learned score function.}
\label{alg:build_classifier_general}
\end{algorithm}

\begin{algorithm}[htp]
\small
\begin{algorithmic}[1]
\Require \parbox[t]{\dimexpr0.9\linewidth-\algorithmicindent\relax}{%
  Blackbox access to $\mathcal M$; threshold $\eta>0$; sample size $n$; databases $D, D'$.}
\Ensure \parbox[t]{\dimexpr0.9\linewidth-\algorithmicindent\relax}{%
  Estimate $(\hat{\alpha}(\eta),\hat{\beta}(\eta))$ of $(\alpha(\eta),\beta(\eta))$ for $(P,Q)$,
  where $\mathcal M(D)\sim P$ and $\mathcal M(D')\sim Q$.}

\State Choose perturbation parameter $h$.
\State Set density estimator $\mathcal A$ (default: KDE).

\Function{PLRT}{$\mathcal A,\allowbreak \mathcal M,\allowbreak D,\allowbreak D',\allowbreak \eta_{\mathrm{vec}},\allowbreak n$}
  \State Run $\mathcal M$ on $D$ and $D'$ to obtain $n$ samples from each.
  \State Fit densities $\hat p$ and $\hat q$ using $\mathcal A$.
  \For{$\eta \in \eta_{\mathrm{vec}}$}
    \State $\hat{\alpha}(\eta)\gets$
    \Statex \hspace{\algorithmicindent}$\displaystyle
      \frac{1}{h}\int_{-h/2}^{h/2}\left(\int
      \mathbf{1}\!\left\{\log\!\frac{\hat q(z)}{\hat p(z)}>\log\eta+x\right\}\hat p(z)\,dz\right)\,dx$
    \State $\hat{\beta}(\eta)\gets 1-$
    \Statex \hspace{\algorithmicindent}$\displaystyle
      \frac{1}{h}\int_{-h/2}^{h/2}\left(\int
      \mathbf{1}\!\left\{\log\!\frac{\hat q(z)}{\hat p(z)}>\log\eta+x\right\}\hat q(z)\,dz\right)\,dx$
  \EndFor
  \State \Return $(\hat p,\hat q,\hat T)$ where $\hat T=\{(\hat{\alpha}(\eta),\hat{\beta}(\eta)):\eta\in\eta_{\mathrm{vec}}\}$.
\EndFunction

\end{algorithmic}
\caption{\textproc{PLRT}: Non-parametric estimator for an optimal $f$-DP classifier}
\label{alg:PLRT}
\end{algorithm}

\begin{algorithm}[htp]
\small
\begin{algorithmic}[1]
\Require \; \parbox[t]{\dimexpr\linewidth-\algorithmicindent-3em}{
Empirical means $\hat T_1, \hat T_2$;
critical value $q_{1-\gamma/2}$;\\ burn-in size $M$; current sample size $k$.
}
\Ensure \, Confidence-adjusted $(\Teins, \Tzwei)$.
\Function{ConfAdj}{$\hat T_1, \hat T_2, q_{1-\gamma/2}, M, k$}
    \State $b \gets \dfrac{\sqrt{\log\!\left(20 + k/M\right)}}{\sqrt{k}}$, standard deviations $s_1 \gets \sqrt{\hat T_1(1-\hat T_1)} $, $s_2 \gets \sqrt{\hat T_2(1-\hat T_2)}$.
    \State $\Teins \gets T_1 + q_{1-\gamma/2}\, s_1\, b$; $\Tzwei \gets T_2 + q_{1-\gamma/2}\, s_2\, b$
    \State $\Teins \gets \min\{1,\max\{0,\, \Teins\}\}$
    \State $\Tzwei \gets \min\{1,\max\{0,\, \Tzwei\}\}$
    \State \Return $(\Teins, \Tzwei)$
\EndFunction
\end{algorithmic}
\caption{Confidence Adjustment for Sequential Audit}
\label{alg:conf-adj}
\end{algorithm}

\begin{algorithm}[htp]
  \caption{\textsc{Quant}$(M,\alpha)$, where standard values $\alpha=0.01,0.05,0.1$ are saved.}
  \label{alg:quantile-threshold}
\small
  \begin{algorithmic}[1]
    \Require Burn‑in size $M\in\mathbb{N}$, significance level $\alpha\in(0,1)$
    \Ensure  $(1-\alpha)$‑quantile $q_{1-\alpha}$ of
      $W=\displaystyle\sup_{k\ge M}\tfrac1{\sqrt{(k)\log(20+k/M)}}\sum_{i=1}^{k}Z_i$ 

    \State $R \gets 10^{5}$ \Comment{number of Monte‑Carlo replications}
    \State $T \gets 10^{4}$ \Comment{truncation for the supremum}
    \For{$r = 1,\dots,R$}
        \State Draw $Z^{(r)}_1,\dots,Z^{(r)}_{T}$ i.i.d.\ $\mathcal N(0,1)$
        \State $W_r \gets \displaystyle\max_{k\ge M}
                 \tfrac1{\sqrt{k\log(20+k/M)}}
                 \sum_{i=1}^{k} Z^{(r)}_i$
    \EndFor
    \State $q_{1-\alpha} \gets$ empirical $(1-\alpha)$‑quantile of $\{W_1,\dots,W_R\}$
    \State \Return $q_{1-\alpha}$
  \end{algorithmic}
\end{algorithm}
\begin{algorithm}[h]
\begin{algorithmic}[1]
\Require \; \parbox[t]{\dimexpr0.9\linewidth-\algorithmicindent}{Estimated tradeoff points $(\hat\alpha(\eta), \hat\beta(\eta))$ for all $\eta \in \eta_{\mathrm{vec}}$, and the claimed curve $f_{\claim}$.}\\[0.1cm]
\Ensure \; Critical $\eta^*$ maximizing diagonal (45°) distance between empirical and claimed curves.

  \Function{Compute-$\eta^*$-Diagonal}{$\hat\alpha(\eta), \hat\beta(\eta), f_{\claim}, \eta_{\mathrm{vec}}$}
    \State Initialize $\text{dists} \gets 0$.
    \For{$\eta_i$ in $\eta_{\mathrm{vec}}$}
      \State Set $\hat\alpha \gets \hat\alpha(\eta_i)$ and $\hat\beta \gets \hat\beta(\eta_i)$.
      \State Define $f_{45}(\alpha') := f_{\claim}(\alpha') - [\,\hat\beta + (\alpha' - \hat\alpha)\,]$.
      \State Solve for $\alpha'$ satisfying $f_{45}(\alpha') = 0$ on $[0,1]$ \Comment{(e.g., via \texttt{uniroot}).}
      \If{a valid root $\alpha'$ is found}
        \State Set $\beta' \gets f_{\claim}(\alpha')$.
        \State Compute $d(\eta_i) := \sqrt{(\alpha' - \hat\alpha)^2 + (\beta' - \hat\beta)^2}$.
      \Else
        \State Set $d(\eta_i) := \text{NA}$.
      \EndIf
    \EndFor
    \State Return $\eta^* := \argmax_{\eta_i \in \eta_{\mathrm{vec}}} d(\eta_i)$.
  \EndFunction

\end{algorithmic}
\caption{\textproc{Compute-$\eta^*$-Diagonal}: 45°-based selection of critical $\eta^*$}
\label{alg:eta_star_45}
\end{algorithm}

\begin{algorithm}[htp]
\small
\begin{algorithmic}[1]
\Require Last refit sample size $n_{\text{last}}$; refit threshold $\tau \in (0,1)$.
\Ensure $M_1$.

\Function{AuditShouldRefit}{$n_{\text{last}},\tau$}
  \State Compute shrinkage for $n>n_{\text{last}}$:
  \[
    \texttt{shrink}(n) \gets 1 - \Big(\frac{n_{\text{last}}}{n}\Big)^{1/5}.
  \]
  \State Set decision: $ M_1 \gets \argmin_{n}\mathbf{1}\{\texttt{shrink}(n) > \tau\}.$
  \State \Return $M_1$.
\EndFunction
\end{algorithmic}
\caption{\textproc{AuditShouldRefit}: Dynamic choice of refresh period $M_1$.}
\label{alg:audit_should_refit}
\end{algorithm}


%% file: 102-theory-adaptive-classifiers.tex
\section{Theory for Adaptive Classifiers} 
\label{sec:adaptive_monitoring}
In Theorem~\ref{theo:main}, we justified our sequential auditing
procedure under the simplifying assumption that the classifier
$\phi$ is fixed throughout the audit. This assumption is useful for
explaining the core idea, but it is restrictive in blackbox settings.
There, the analyst will typically refine the classifier as more
samples become available. To justify such scenarios mathematically, we derive in this section a Gaussian approximation for the estimated type-I and type-II errors under evolving classifiers. This approximation is (as we recall) the central tool for proving Theorem \ref{theo:main}.\\
The key point is the following. Reusing previous samples to update
the classifier makes the sequence of classification outcomes
dependent. However, sequential Gaussian approximations are not
limited to independent data. They remain valid for weakly dependent
and approximately stationary processes under suitable assumptions.
We therefore formulate the adaptive classifier sequence as a weakly
dependent process and invoke the Gaussian approximation principle
developed in~\cite{kutta:kokoszka:2025}. The resulting theorem is an
adaptive analogue of Theorem~\ref{theo:main}.\\
Let $(X_k,Y_k)_{k\ge 1}$ denote the pairs of outputs sampled from
the two neighboring databases, with $X_k\sim P$ and $Y_k\sim Q$.
At time $k$, the classifier may depend on all previous observations,
but not on the current pair $(X_k,Y_k)$. We write
\[
  \hat \phi_k(x)
  =
  \mathcal H_k\bigl(x,\varepsilon_k,(X_{k-1},Y_{k-1}),\ldots,(X_1,Y_1)\bigr),
\]
where $\mathcal H_k$ is measurable and where $(\varepsilon_k)_{k\ge1}$
is an iid sequence, independent of the data, representing possible
randomization used by the classifier construction. The classification
outcome at time $k$ is encoded by
\[
  R_k
  :=
  \begin{pmatrix}
    \hat \phi_k(X_k) \\
    1-\hat \phi_k(Y_k)
  \end{pmatrix}
  \in \{0,1\}^2 .
\]
Thus the empirical error pair used by the auditor is
$ \hat\theta_k :=
  \begin{pmatrix}
    \hat \alpha_k \\
    \hat \beta_k
  \end{pmatrix}
  := \frac1k\sum_{i=1}^k R_i .$
The first component estimates the type-I classification error and
the second component estimates the type-II classification error.
We next state the assumptions under which adaptive monitoring remains
valid. For this purpose, we first need a notion of dependence across data.
\begin{definition}
Let $(Z_k)_{k\ge1}$ be a time series with values in a separable
metric space. Its $\alpha$-mixing coefficients are defined by
\[
  \alpha_Z(h)
  :=
  \sup_{m\ge1}
  \sup_{\substack{
      A\in \sigma(Z_1,\ldots,Z_m)\\
      B\in \sigma(Z_{m+h},Z_{m+h+1},\ldots)
  }}
  \left|
    \mathbb P(A\cap B)-\mathbb P(A)\mathbb P(B)
  \right|.
\]
We say that $(Z_k)_{k\ge1}$ is $\alpha$-mixing if
$\alpha_Z(h)\to0$ as $h\to\infty$.
\end{definition}
The intuition of mixing coefficients is simple. In the case where the two
$\sigma$-algebras are independent, it should hold for any pair of events that
$\mathbb P(A\cap B)=\mathbb P(A)\mathbb P(B)$, or equivalently that
$\mathbb P(A\cap B)-\mathbb P(A)\mathbb P(B)=0$. Consequently, a natural
generalization to weak dependence, where $\sigma$-algebras contain little
information about each other, is that the difference be small. A detailed
discussion of mixing can be found in~\cite{bradley:2007}. Mixing has for a
long time been the most popular way to describe dependence. Nowadays there
are alternatives, such as physical dependence, and we conjecture that similar
results as derived below would also hold under this dependence concept. Next,
we formulate the main technical assumptions for this section.
\begin{ass}
\label{ass:adaptive}
The adaptive classification process $(R_k)_{k\ge1}$ satisfies the
following conditions.
\begin{itemize}
  \item[A1)] The process $(R_k)_{k\ge1}$ is $\alpha$-mixing with
  coefficients satisfying
  \[
    \alpha_R(h) \le C_\alpha h^{-p}
  \]
  for some constants $C_\alpha>0$ and $p>4$.
  \item[A2)] There exists a deterministic measurable classifier
  $\phi$ such that, with $R_k^\circ := \begin{pmatrix}
      \phi(X_k)\\
      1-\phi(Y_k)
    \end{pmatrix},$
  we have
  \[
    \mathbb P(R_k\neq R_k^\circ) \le C_\phi k^{-\nu}
  \]
  for some constants $C_\phi,\nu>0$.
\end{itemize}
\end{ass}
Assumption~A1 means that the dependence induced by reusing earlier
observations for classifier updates is weak enough. The assumption of
polynomial decay of mixing coefficients is relatively weak.
Assumption~A2 means that the adaptive classifier stabilizes asymptotically
toward a deterministic classifier $\phi$. We do not assume that $\phi$ is
optimal. Optimality is relevant for power, but not for validity of the
sequential confidence adjustment.\\
We next prove the main Gaussian approximation needed to justify the adaptive
auditor. The statement is scalar and applies to either coordinate of $R_k$.
This is sufficient, since the confidence adjustment in the auditor is applied
coordinatewise.\\
For $j\in\{1,2\}$, let $e_1=(1,0)^\top$ and $e_2=(0,1)^\top$, and
define
\[
  U_{j,k}:=e_j^\top R_k,
  \qquad
  U_{j,k}^\circ:=e_j^\top R_k^\circ .
\]
Thus
\[
  U_{1,k}=\hat\phi_k(X_k),
  \qquad
  U_{2,k}=1-\hat\phi_k(Y_k).
\]
Write $S_j(k):=\sum_{i=1}^k\bigl(U_{j,i}-\mathbb E U_{j,i}\bigr)$ and
\[
  D_{j,M}(k)
  :=
  \frac{S_j(k)}
       {\sqrt{k\log(20+k/M)}} ,
  \qquad k\ge M .
\]
\begin{lemma}[Proof in Appendix~\ref{app:adaptive-classifier-proofs}]
\label{lem:adaptive_gaussian_approx}
Suppose Assumption~\ref{ass:adaptive} holds. Fix
$j\in\{1,2\}$ and assume that
\[
  \sigma_j^2:=\operatorname{Var}(U_{j,1}^\circ)>0.
\]
Then, on a suitable probability space,
\[
  \sup_{k\ge M}D_{j,M}(k)
  =
  \sup_{x\ge1}
  \frac{\sigma_j W_j(x)}
       {\sqrt{x\log(20+x)}}
  +o_P(1),
  \qquad M\to\infty,
\]
where $W_j$ is a standard Brownian motion.
\end{lemma}
We now prove the bounded law of the iterated logarithm that was used
in the previous proof. The next result uses the definition of the
Prokhorov metric, taken from~\cite{kutta:kokoszka:2025}.
\begin{definition}
Let $(\mathcal{M}, d_\mathcal{M})$ be a generic metric space.
We define for a set $A \subset \mathcal{M}$ and $\chi>0$ the
$\chi$-environment
\[
A^\chi:=\{x \in \mathcal{M}|\,\, \exists a \in A:
d_\mathcal{M}(x,a) <\chi\}.
\]
We call a set $A$ measurable if it is a Borel set in $\mathcal{M}$
with respect to $d_\mathcal{M}$. Now we can define the Prokhorov
distance of two probability measures $\mathbb{P}_1, \mathbb{P}_2$ as
\begin{align*}
&\pi(\mathbb{P}_1,\mathbb{P}_2) := \\
&\inf\Big\{ \chi>0 :
\mathbb{P}_1(A)\le \mathbb{P}_2(A^\chi)+\chi
\ \text{and}\
\mathbb{P}_2(A)\le \mathbb{P}_1(A^\chi)+\chi
\ \forall A
\Big\}.
\end{align*}

\end{definition}
The Prokhorov metric is a somewhat technical concept: it measures the
distance between two probability measures, taking into account the
topology of the underlying metric space. This is different from
alternatives such as the total variation metric, which poses stricter
requirements that do not, in the same way, refer to topology. A
drawback of total variation is that total variation convergence often
fails to hold, even when weak convergence holds. In contrast,
Prokhorov convergence is closely related to weak convergence; on
separable spaces it is even equivalent to it. It therefore gives a
natural metric for proximity of distributions.\\
In the next lemma, we impose proximity of two stochastic processes
with respect to the Prokhorov metric. The underlying metric space is
that of bounded functions on the unit interval $[0,1]$, equipped with
the supremum norm. We also use the notation that, with only minor
abuse, we write $\pi(X,Y)$ for random variables $X,Y$ when we mean
$\pi(\mathbb{P}_X,\mathbb{P}_Y)$, identifying random variables with
their distributions. The proof is an adaptation of the proof strategy
of Corollary~2 in~\cite{kutta:kokoszka:2025}.\\
For the next lemma, we define for a general triangular array of
bounded random variables $(U_{i,M})_{i,M}$ the sum
\[
 S_M'(k)=\sum_{i=1}^k
  \bigl(U_{i,M}-\mathbb E U_{i,M}\bigr).
\]
\begin{lemma}[Proof in Appendix~\ref{app:adaptive-classifier-proofs}]
\label{lem:lil_scale_bound}
Define the partial sum process
\[
  P_{M,n}(t)
  :=
  \frac{1}{\sqrt n}S_M'(\lfloor nt\rfloor),
  \qquad 0\le t\le1.
\]
Suppose it satisfies, uniformly over $M$, the Prokhorov approximation
bound
\[
 \pi\bigl(P_{M,n},\sigma W\bigr)
  \le C n^{-\tau}
  \tag{PA}
  \label{eq:prokhorov_approx}
\]
for some constants $C,\tau>0$, where $W$ is standard Brownian motion
on $[0,1]$ and $\sigma^2>0$. Also assume that the individual variables
$U_{i,M}$ are bounded. Then
\[
  \sup_{k\ge1}
  \left|
    \frac{S_M'(k)}
         {\sqrt{k\log(\log(k+e))}}
  \right|
  =
  \mathcal O_P(1),
  \qquad M\to\infty .
\]
\end{lemma}

For the next result, suppose that Algorithm~\ref{alg:seq-audit} is run
with the outputs of the adaptive classifiers $R_k$ described in this
section. Recall that $\phi$ denotes the limiting classifier appearing
in Assumption~\ref{ass:adaptive}.
\begin{theo}[Proof in Appendix~\ref{app:adaptive-classifier-proofs}]
\label{theo:adaptive_validity}
Suppose the assumptions of Lemma~\ref{lem:adaptive_gaussian_approx}
hold for both coordinates $j=1,2$.
\begin{itemize}
    \item[a)] Suppose that the claimed $f_{\claim}$-DP guarantee is
    satisfied. Then
    {\small
    \begin{align*}
        \limsup_{M\to\infty}
      \mathbb{P}\!\left(\textnormal{\textsc{APT} rejects for some } k\ge M\right)
      \le \gamma .
    \end{align*}
    }
    \item[b)] Suppose that the claimed $f_{\claim}$-DP guarantee is
    violated and that $f_{\claim}(\alpha_\phi)>\beta_\phi.$
    Then, for $k_{\max}=\infty$,
    {\small
        \begin{align*}
            \lim_{M\to\infty}
          \mathbb{P}\!\left(\textnormal{\textsc{APT} rejects for some } k\ge M\right)
          =1 .
        \end{align*}
    }
\end{itemize}
\end{theo}

%% file: 104-full-version-supplementary-material.tex
\section{Proofs for the Adaptive-Classifier Theory}
\label{app:adaptive-classifier-proofs}
\begin{proof}[Proof of Lemma~\ref{lem:adaptive_gaussian_approx}]
We prove the statement for a fixed coordinate $j$ and suppress $j$
from the notation. Thus write
\[
  U_k=e_j^\top R_k,
  \qquad
  U_k^\circ=e_j^\top R_k^\circ,
\]
\[
  S(k)=\sum_{i=1}^k(U_i-\mathbb E U_i),
  \qquad
  D_M(k)=\frac{S(k)}{\sqrt{k\log(20+k/M)}}.
\]
The proof follows an early--late decomposition of the monitoring
indices. Let $\zeta\in(0,1/2)$ be a sufficiently small parameter,
specified below, and define
\[
  \mathcal E_M:=\{M\le k\le M^{1+\zeta}\},
  \qquad
  \mathcal L_M:=\{k>M^{1+\zeta}\}.
\]
Then
\[
  \sup_{k\ge M}D_M(k)
  =
  \max\left\{
    \sup_{k\in\mathcal E_M}D_M(k),
    \sup_{k\in\mathcal L_M}D_M(k)
  \right\}.
\]
We first replace the first few observations by their limiting
counterparts. For this purpose, put
\[
  m_M:=\lfloor M^\zeta\rfloor
\]
and define
\[
  U_{i,M}
  :=
  \begin{cases}
    U_i^\circ, & 1\le i\le m_M,\\
    U_i,       & i>m_M.
  \end{cases}
\]
Let
\[
  S_M'(k):=\sum_{i=1}^k
  \bigl(U_{i,M}-\mathbb E U_{i,M}\bigr),
  \quad
  D_M'(k):=
  \frac{S_M'(k)}
       {\sqrt{k\log(20+k/M)}} .
\]
Since $U_i,U_i^\circ\in\{0,1\}$, the centered sums defining $D_M(k)$
and $D_M'(k)$ can differ only through the first $m_M$ terms. Hence
\[
  \sup_{k\ge M}
  |D_M(k)-D_M'(k)|
  \le
  \frac{2m_M}{\sqrt{M\log(21)}}
  =
  O(M^{\zeta-1/2})
  =
  o(1).
\]
It is therefore enough to prove the Gaussian approximation for
$D_M'$.\\
We next consider the early range $\mathcal E_M$. The variables
$U_{i,M}$ may be identified with random functions $f_{U_{i,M}}$ that
are constant on $[0,1]$, by defining
\[
  f_{U_{i,M}}(x)\equiv U_{i,M}.
\]
Under this identification, the functional assumptions in
Theorem~4.1 of~\cite{kutta:kokoszka:2025} reduce to ordinary moment,
weak-dependence, and approximation conditions. The moment conditions
are automatic because $U_{i,M}\in\{0,1\}$. The Hölder modulus of
continuity of $f_{U_{i,M}}$ is identically zero. The required weak
dependence follows from Assumption~A1; moreover, the limiting
variables $U_i^\circ$ are independent over time, so replacing the
first $m_M$ terms by $U_i^\circ$ does not increase the relevant serial
dependence. Finally, Assumption~A2 gives the required polynomial
approximation to the stationary limiting process $(U_i^\circ)_{i\ge1}$.
Choosing $\zeta>0$ sufficiently small, the rate restrictions in
Theorem~4.1 of~\cite{kutta:kokoszka:2025} are satisfied on the
expanding range $\mathcal E_M$.\\
Consequently, by Theorem~4.1 of~\cite{kutta:kokoszka:2025}, we may
construct a standard Brownian motion $W$ such that
\[
  \sup_{k\in\mathcal E_M}
  \left|
    \frac{S_M'(k)}
         {\sqrt{k\log(20+k/M)}}
    -
    \frac{\sigma W(k/M)}
         {\sqrt{(k/M)\log(20+k/M)}}
  \right|
  =
  o_P(1).
\]
Equivalently,
\[
  \sup_{k\in\mathcal E_M}D_M'(k)
  =
  \sup_{1\le x\le M^\zeta}
  \frac{\sigma W(x)}
       {\sqrt{x\log(20+x)}}
  +o_P(1),
\]
where we have used the Brownian scaling relation.\\
It remains to control the late range. First notice that the
corresponding Brownian tail is negligible, since
\[
  \sup_{x>M^\zeta}
  \frac{|W(x)|}{\sqrt{x\log(20+x)}}
  \to 0
  \qquad\text{a.s.}
\]
by the law of the iterated logarithm for Brownian motion. The more
difficult part is now to show that
\[
  \sup_{k\in\mathcal L_M}|D_M'(k)|=o_P(1).
\]
To establish it, it will suffice to show that
\begin{align} \label{e:LIL}
  \sup_{k \ge 1}
  \left|
    \frac{S_M'(k)}
         {\sqrt{k\log(\log(k+e))}}\right|
  =
  \mathcal{O}_P(1).
\end{align}
This result follows from Lemma~\ref{lem:lil_scale_bound}, whose
condition is verified by Theorem~4.1 of~\cite{kutta:kokoszka:2025}
as explained below. With \eqref{e:LIL} in hand, we conclude that
\begin{align*}
    & \sup_{k\in\mathcal L_M}|D_M'(k)| \\
    =&  \sup_{k >M^{1+\zeta}}
  \left|
    \frac{S_M'(k)}
         {\sqrt{k \log(\log(k+e))}}
    \sqrt{\frac{\log(\log(k+e))}{\log(20+k/M)} }
  \right| \\
  \le & \mathcal{O}_P(1)
  \sup_{k >M^{1+\zeta}}
  \sqrt{\frac{\log(\log(k+e))}{\log(20+k/M)} }.
\end{align*}
For $k>M^{1+\zeta}$ we have
\[
  M<k^{1/(1+\zeta)}
  \qquad\text{and hence}\qquad
  \frac{k}{M}>k^{\zeta/(1+\zeta)}.
\]
Therefore, for all sufficiently large $k$,
\[
  \log(20+k/M)
  \ge
  c_\zeta \log k
\]
with $c_\zeta>0$. Since
\[
  \frac{\log(\log(k+e))}{\log k}\to0,
\]
we obtain
\[
  \sup_{k >M^{1+\zeta}}
  \sqrt{\frac{\log(\log(k+e))}{\log(20+k/M)} }
  \to0.
\]
Consequently,
\[
  \sup_{k\in\mathcal L_M}|D_M'(k)|=o_P(1).
\]
Combining the early and late ranges yields first
\[
  \sup_{k\ge M}D_M(k)
  =
  \sup_{1\le x\le M^\zeta}
  \frac{\sigma W(x)}
       {\sqrt{x\log(20+x)}}
  +o_P(1),
\]
and then, using the negligible Brownian tail,
\[
  \sup_{1\le x\le M^\zeta}
  \frac{\sigma W(x)}
       {\sqrt{x\log(20+x)}}
  =
  \sup_{x\ge1}
  \frac{\sigma W(x)}
       {\sqrt{x\log(20+x)}}
  +o_P(1).
\]
This proves the claim.
\end{proof}

\begin{proof}[Proof of Lemma~\ref{lem:lil_scale_bound}]
Let
\[
  n_\ell:=\lceil e^\ell\rceil,
  \qquad \ell\ge1 .
\]
For $n_\ell\le k\le n_{\ell+1}$, we have
\[
  k\ge n_\ell,
  \qquad
  \log(\log(k+e))\ge c\log \ell
\]
for some numerical constant $c>0$ and all sufficiently large $\ell$.
Consequently,
\begin{align*}
  \max_{n_\ell\le k\le n_{\ell+1}}
  \frac{|S_M'(k)|}
       {\sqrt{k\log(\log(k+e))}}
  &\le
  C
  \max_{n_\ell\le k\le n_{\ell+1}}
  \frac{|S_M'(k)|}
       {\sqrt{n_\ell\log \ell}}        \\
  &\le
  C
  \frac{\sqrt{n_{\ell+1}}}{\sqrt{n_\ell}}
  \frac{\|P_{M,n_{\ell+1}}\|_\infty}
       {\sqrt{\log \ell}}              \\
  &\le
  C
  \frac{\|P_{M,n_{\ell+1}}\|_\infty}
       {\sqrt{\log \ell}},
\end{align*}
where $\|\cdot\|_\infty$ denotes the supremum norm on $[0,1]$ and
where we used $n_{\ell+1}/n_\ell\le C e$.\\
We now analyze the right-hand side. Let
\[
  \delta_\ell:= C n_{\ell+1}^{-\tau}.
\]
By the Prokhorov approximation \eqref{eq:prokhorov_approx}, we get
for any $A>0$
\[
  \mathbb P\left(
    \frac{\|P_{M,n_{\ell+1}}\|_\infty}{\sqrt{\log \ell}}
    > A\sigma
  \right)
  \le
  \mathbb P\left(
    \|\sigma W\|_\infty
    >
    A\sigma\sqrt{\log\ell}-\delta_\ell
  \right)
  +
  \delta_\ell .
\]
Since $n_{\ell+1}\ge e^\ell$, we have
\[
  \delta_\ell\le C e^{-\tau\ell}.
\]
For all large $\ell$, the threshold on the Gaussian term is bounded
below by $cA\sigma\sqrt{\log\ell}$ for a numerical constant $c>0$.
Thus it suffices to bound
\[
  \mathbb P\left(
    \|W\|_\infty> cA\sqrt{\log\ell}
  \right).
\]
By Fernique's theorem there exists $C_0>0$ such that
\[
  \mathbb E
  \exp\left(
    \frac{\|W\|_\infty^2}{C_0^2}
  \right)
  <\infty .
\]
Choose $A$ sufficiently large so that
\[
  \lambda_A:=\frac{2}{c^2A^2}
  \le \frac1{C_0^2}.
\]
Then, by Markov's inequality,
\[
  \mathbb P\left(
    \|W\|_\infty> cA\sqrt{\log\ell}
  \right)
  \le
  \mathbb E\exp\{\lambda_A\|W\|_\infty^2\}\ell^{-2}
  \le
  C_A\ell^{-2}.
\]
Therefore,
\[
  \sum_{\ell=1}^\infty
  \mathbb P\left(
    \frac{\|P_{M,n_{\ell+1}}\|_\infty}{\sqrt{\log \ell}}
    > A\sigma
  \right)
  \le
  \sum_{\ell=1}^\infty
  \{C_A\ell^{-2}+Ce^{-\tau\ell}\}
  <\infty,
\]
with the right-hand side independent of $M$.\\
It remains to translate this blockwise summability into the asserted
$\mathcal O_P(1)$ bound. For any $L\ge1$,
\begin{align*}
   &\mathbb P\left(
      \sup_{k\ge n_L}
      \frac{|S_M'(k)|}
           {\sqrt{k\log(\log(k+e))}}
      > A\sigma
    \right)\\
   &\le
   \sum_{\ell=L}^\infty
   \mathbb P\left(
     C
     \frac{\|P_{M,n_{\ell+1}}\|_\infty}
          {\sqrt{\log \ell}}
     > A\sigma
   \right).
\end{align*}
By the summability just proved, the right-hand side tends to zero as
$L\to\infty$, uniformly in $M$.\\
Now fix $\epsilon>0$. Choose $L$ so large that the last tail
probability is smaller than $\epsilon$. For the finite initial part,
boundedness of the variables $U_{i,M}$ implies that
\[
  \max_{1 \le k \le n_L}
  \frac{|S_M'(k)|}
       {\sqrt{k\log(\log(k+e))}}
\]
is bounded by a finite deterministic constant, depending only on $L$.
Increasing $A$ if necessary, we therefore get
\[
  \mathbb P\left(
    \sup_{1\le k<\infty}
    \frac{|S_M'(k)|}
         {\sqrt{k\log(\log(k+e))}}
    > A\sigma
  \right)
  \le \epsilon .
\]
This proves
\[
  \sup_{k\ge1}
  \left|
    \frac{S_M'(k)}
         {\sqrt{k\log(\log(k+e))}}
  \right|
  =
  \mathcal O_P(1).
\]
\end{proof}

\begin{proof}[Proof of Theorem~\ref{theo:adaptive_validity}]
Write
\[
  \bar\theta_k
  :=
  \mathbb E\hat\theta_k
  =
  \begin{pmatrix}
    \bar\alpha_k\\
    \bar\beta_k
  \end{pmatrix}
  :=
  \frac1k\sum_{i=1}^k \mathbb E R_i .
\]
Since the claimed $f_{\claim}$-DP guarantee holds, every
classifier has an error pair lying on or above the claimed tradeoff
curve.  Hence, for each $i$,
$  \mathbb E R_i$
lies on or above the curve $f_{\claim}$. Since the epigraph of
$f_{\claim}$ is convex, the running average $\bar\theta_k$ also lies
on or above the curve. Thus
\[
  \bar\beta_k\ge f_{\claim}(\bar\alpha_k),
  \qquad k\ge1 .
\]
By Lemma~\ref{lem:adaptive_gaussian_approx}, applied coordinatewise
and combined with the confidence adjustment
$\textproc{ConfAdj}$ used in Algorithm~\ref{alg:seq-audit}, we obtain
paired one-sided confidence intervals
\[
(I_{k,M}^{\alpha},I_{k,M}^{\beta})
  =
  \bigl([0,\hat T_{\alpha}^{k,M}],
        [0,\hat T_{\beta}^{k,M}]\bigr)
\]
such that
\[
  \liminf_{M\to\infty}
  \mathbb P\left(
    (\bar\alpha_k,\bar\beta_k)
    \in I_{k,M}^{\alpha}\times I_{k,M}^{\beta}
    \text{ for all } k\ge M
  \right)
  \ge 1-\gamma .
\]
On this event, for every $k\ge M$,
\[
  \bar\alpha_k\le \hat T_{\alpha}^{k,M},
  \qquad
  \bar\beta_k\le \hat T_{\beta}^{k,M}.
\]
Since $f_{\claim}$ is non-increasing,
\[
  f_{\claim}(\hat T_{\alpha}^{k,M})
  \le
  f_{\claim}(\bar\alpha_k)
  \le
  \bar\beta_k
  \le
  \hat T_{\beta}^{k,M}.
\]
Therefore the rejection condition
\[
  \hat T_{\beta}^{k,M}<f_{\claim}(\hat T_{\alpha}^{k,M})
\]
cannot occur on the simultaneous confidence event. Hence a false
rejection can only occur on the complement of that event, whose
asymptotic probability is at most $\gamma$. This proves part~a).\\
For part~b), Assumption~\ref{ass:adaptive} implies that
$\mathbb E R_i\to(\alpha_\phi,\beta_\phi)^\top$, and hence
\[
  \bar\theta_k\to(\alpha_\phi,\beta_\phi)^\top .
\]
If $f_{\claim}(\alpha_\phi)>\beta_\phi$, then the running
mean pair $\bar\theta_k$ eventually lies below the claimed curve by a
positive margin. The Gaussian approximation in
Lemma~\ref{lem:adaptive_gaussian_approx} implies that the empirical
pair and its confidence adjustment also get arbitrarily close to $\bar\theta_k$. Consequently, for large enough $k$, the inequality
\[
  \hat T_{\beta}^{k,M}<f_{\claim}(\hat T_{\alpha}^{k,M})
\]
holds with probability tending to one. This proves part~b).
\end{proof}

\section{Additional Simulations}\label{app:add_sim}
In this section, we provide a bigger picture and compare our sequential procedure to two non-sequential references: a baseline that uses the same techniques as our methodology, but with a fixed batch size, i.e. it performs only a one-shot $t$-test at fixed sample size $n$. Second, we compare ourselves to the $f$-DP auditor proposed in \cite{askingeneral} - a state-of-the-art method. \\
Sequential monitoring necessarily introduces additional sampling costs through time-uniform guarantees, relative to a test with a fixed sample size $n$ - at least if the test is efficient and $n$ is optimally calibrated. Recall that such optimal calibration is practically impossible but serves as an optimal theoretical benchmark. Accordingly we cannot expect, to outperform the non-sequential baseline in terms of power. \\
We report power curves of the three methods in Figures \ref{ref:fig_comp_1} and \ref{ref:fig_comp_2}. The blue line reports the rejection rate at any sample size $n$, the green line the rejection rate for a one-sided $t$-test (using the rest of our new subroutines to make an optimal classifier) with a fixed sample size and the orange line, the method by \cite{askingeneral}, also for a fixed sample size $n$. The comparison is somewhat favorable to the non-sequential methods, because the fixed sample tests can exhaust their entire significance level $\gamma$ (here $0.05$ for all tests), whereas the sequential procedure has a significance level $<0.05$ on any finite interval. Yet, it is difficult to come up with another comparison. \\
Our simulation results have two implications: First, comparing the green and blue line, we see that (as expected) the non-sequential method has higher power than the sequential version. However, the power difference is always moderate, and sequential testing requires the same sample size times a moderate factor as the non-sequential version for equal power. Again, recall that the gap is actually even smaller because we are comparing a sequential procedure for an infinite time horizon with a fixed-batch test. Using our sequential procedure would effectively prevent analysts from using unnecessarily large samples (several orders of magnitude too large) compared to the theoretical optimum.\\
Second, compared to the $f$-DP auditor \cite{askingeneral}, our method achieves consistently higher empirical rejection rates at all stopping times. This is a surprisingly positive result and reflects that relying on generic concentration bounds (Hoeffding in the case of \cite{askingeneral}) is oftentimes suboptimal for statistical inference. Here, even when adding the cost of sequential testing to our procedure, we still substantially outperform \cite{askingeneral}. 
We also believe that our new, geometric method to calibrate optimal classifiers (see Section \ref{sub:crit_eta}) contributes to this discrepancy.

\begin{figure}[H]
    \centering
    \includegraphics[width=.9\linewidth]{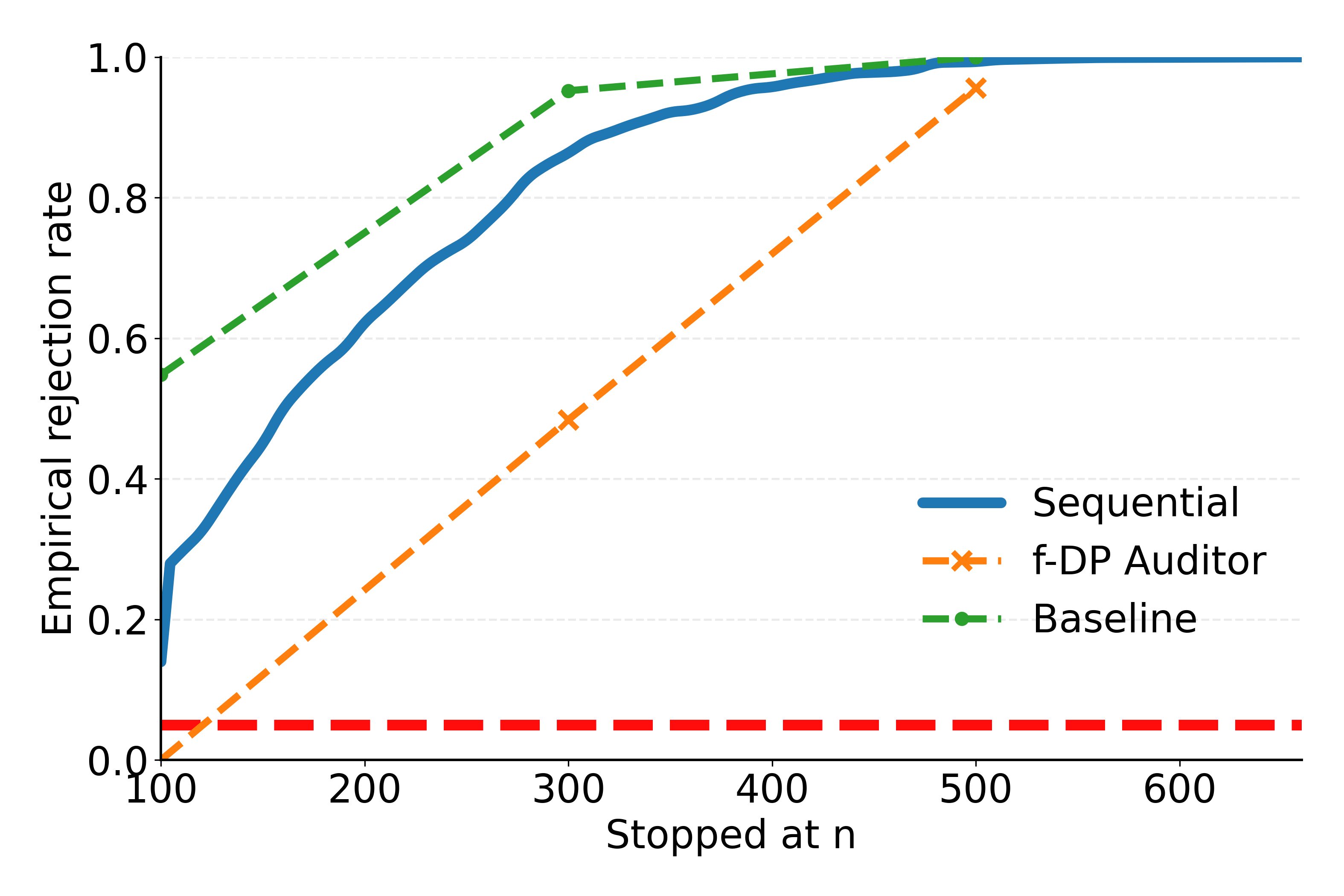}
    \caption{\label{ref:fig_comp_1} Comparison for $\mu_{\text{claim}}=0.5$ (in reality $\mu=1$) between one-shot t-test (green curve), Algorithm \ref{alg:seq-audit} (blue curve) and \cite{askingeneral} (orange curve).}
\end{figure}

\begin{figure}[H]
    \centering
    \includegraphics[width=0.9\linewidth]{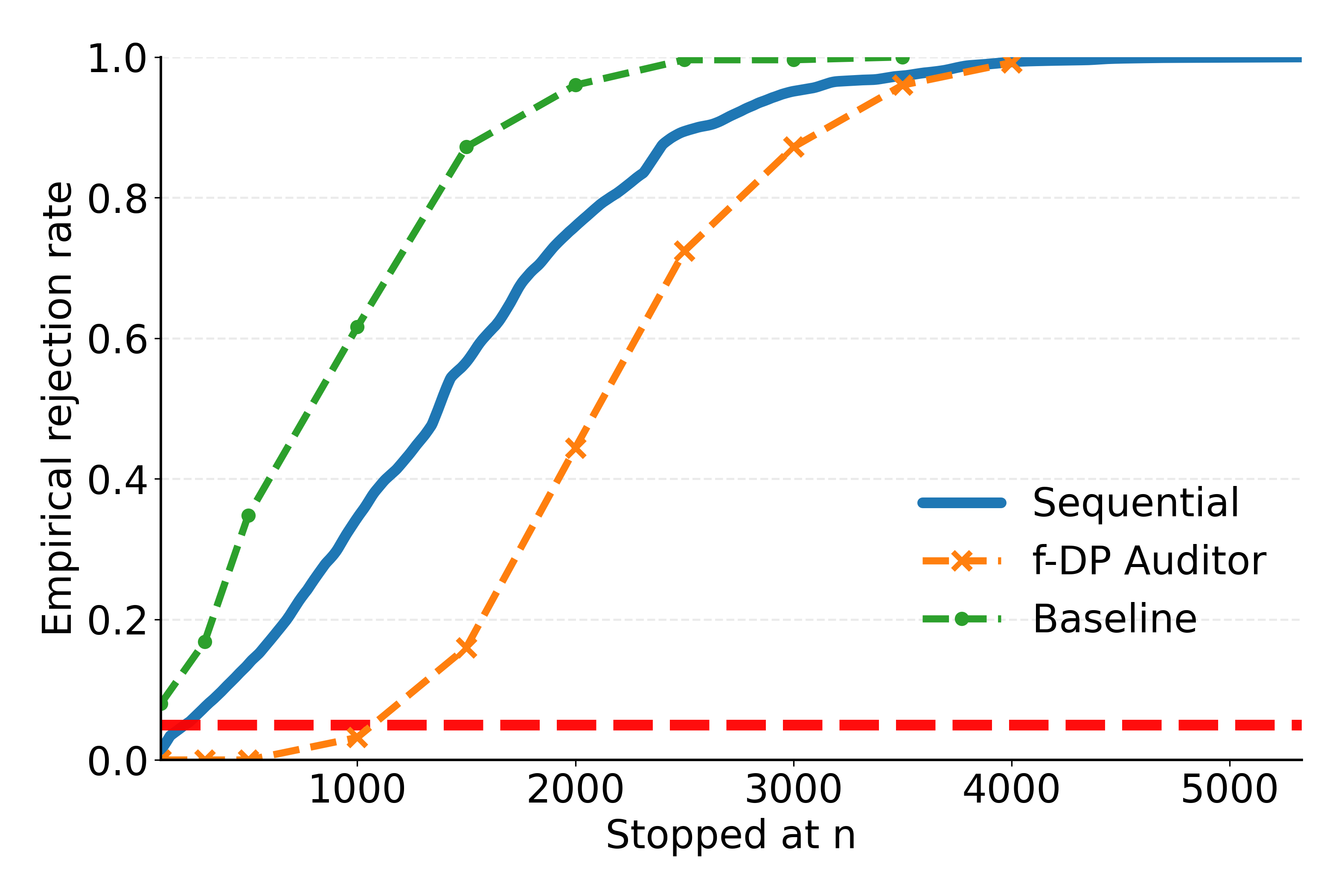}
    \caption{ \label{ref:fig_comp_2} Comparison for $\mu_{\text{claim}}=0.8$ (in reality $\mu=1$) between one-shot t-test (green curve), Algorithm \ref{alg:seq-audit} (blue curve) and \cite{askingeneral} (orange curve).}
\end{figure}